\documentclass[10pt,reqno]{amsart}

\usepackage{hyperref}
\usepackage{latexsym,amsmath, bm}
\usepackage{enumerate}
\usepackage{amsfonts}
\usepackage{amssymb} 
\usepackage{geometry}
\usepackage{latexsym}
\usepackage{fixmath}
\usepackage{faktor}
\usepackage{mathtools}

\usepackage[T1]{fontenc}
\usepackage{fourier}
\usepackage{bbm}
\usepackage{color}

\usepackage{subcaption}

\usepackage{graphicx}
\usepackage{fullpage}

\usepackage[boxed, linesnumbered, noend, noline]{algorithm2e}
\SetKwInput{KwData}{Input}
\SetKwInput{KwResult}{Output}

\numberwithin{equation}{section}

\newcommand\kdef{k=\lceil n^\theta\rceil}

\newcommand{\cchen}{c_{\mathrm{ls}}}
\newcommand{\cchenges}{c_{\mathrm{CS}}}

\newcommand{\poi}{p_{01}}
\newcommand{\pii}{p_{11}}
\newcommand{\pio}{p_{10}}

\DeclareMathOperator*{\argmin}{arg\,min}

\newcommand{\cex}{c_{\mathrm{ex}}}
\newcommand{\ccentre}{c_{\mathrm{ex},2}}
\newcommand{\cbound}{c_{\mathrm{ex},1}}
\newcommand{\cyz}{c_{\mathrm{ex},0}}
\newcommand{\cchan}{c_{\mathrm{Sh}}}
\newcommand{\dchan}{d_{\mathrm{Sh}}}

\newcommand{\cn}{c_n}
\newcommand{\dn}{d_n}

\newcommand{\wx}[1]{\vW_{{#1}}}
\newcommand{\wxj}[2]{\vW_{{#1},{#2}}}

\renewcommand\log{\ln}
\renewcommand\subset{\subseteq}
\newcommand\disteq{\sim}

\newcommand\SPARC{\ASPIV}
\newcommand\ASPIV{{\tt SPARC}}
\newcommand\SPEX{{\tt SPEX}}
\newcommand\SPIV{{\tt SPIV}}
\newcommand\SPAX{\ASPIV}
\newcommand\DD{{\tt DD}}

\newcommand{\fz}{\mathfrak z}
\newcommand{\fU}{\mathfrak U}

\newcommand{\vX}{\vec X}

\renewcommand{\epsilon}{\eps}

\newcommand{\Gsp}{\mathbf G_{\mathrm{sc}}}
\newcommand{\Gsc}{\Gsp}
\newcommand{\Gcc}{\mathbf G_{\mathrm{cc}}}

\newcommand\vY{\vec Y}

\newcommand\vm{\vec m}

\newcommand\aSIGMA{\SIGMA'}
\newcommand\dSIGMA{\SIGMA''}
\newcommand\DSIGMA{\dSIGMA'}

\newcommand\optDs{\cD}

\newcommand\vU{\vec U}

\newcommand\nix{\,\cdot\,}

\newcommand\vW{\vec W}

\newcommand\dd{{\mathrm d}}

\renewcommand{\vec}[1]{\boldsymbol{#1}}

\newcommand\KL[2]{D_{\mathrm{KL}}\bc{{{#1}\|{#2}}}}
\newcommand\Kl[2]{D_{\mathrm{KL}}({#1}\|{#2})}

\newcommand\SIGMA{\vec\sigma}

\newtheorem{definition}{Definition}[section]
\newtheorem{claim}[definition]{Claim}

\newtheorem{remark}[definition]{Remark}
\newtheorem{theorem}[definition]{Theorem}
\newtheorem{lemma}[definition]{Lemma}
\newtheorem{proposition}[definition]{Proposition}
\newtheorem{corollary}[definition]{Corollary}

\newtheorem{fact}[definition]{Fact}

\newcommand\fE{\mathfrak{E}}

\newcommand\cA{\mathcal{A}}

\newcommand\cD{\mathcal{D}}

\newcommand\cE{\mathcal{E}}
\newcommand\cU{\mathcal{U}}
\newcommand\cN{\mathcal{N}}

\newcommand\cS{\mathcal{S}}

\newcommand\cI{\mathcal{I}}

\newcommand\cL{\mathcal{L}}
\newcommand\cM{\mathcal{M}}

\newcommand\cP{\mathcal{P}}
\newcommand\cX{\mathcal{X}}
\newcommand\cY{\mathcal{Y}}

\newcommand\cZ{\mathcal{Z}}

\def\cE{{\mathcal E}}

\newcommand\vZ{\vec Z}

\newcommand\eul{\mathrm{e}}
\newcommand\eps{\varepsilon}

\newcommand\ZZpos{\mathbb{Z}_{\geq0}}

\newcommand\Erw{\mathbb{E}}
\newcommand{\vecone}{\mathbb{1}}

\newcommand{\Bin}{{\rm Bin}}
\newcommand{\Hyp}{{\rm Hyp}}

\newcommand{\Be}{{\rm Be}}

\newcommand\bc[1]{\left({#1}\right)}
\newcommand\cbc[1]{\left\{{#1}\right\}}
\newcommand\bcfr[2]{\bc{\frac{#1}{#2}}}

\newcommand\brk[1]{\left\lbrack{#1}\right\rbrack}

\newcommand\norm[1]{\left\|{#1}\right\|}
\newcommand\abs[1]{\left|{#1}\right|}

\newcommand\QQ{\mathbb{Q}}
\newcommand\RR{\mathbb{R}}
\newcommand\RRpos{\RR_{\geq0}}
\def\?#1{}
\makeatletter
\def\whp{w.h.p\@ifnextchar-{.}{\@ifnextchar.{.\?}{\@ifnextchar,{.}{\@ifnextchar){.}{\@ifnextchar:{.:\?}{.\ }}}}}}
\makeatother

\makeatletter
\def\Whp{W.h.p\@ifnextchar-{.}{\@ifnextchar.{.\?}{\@ifnextchar,{.}{\@ifnextchar){.}{\@ifnextchar:{.:\?}{.\ }}}}}}
\makeatother

\newcommand\pr{\mathbb{P}} 
\renewcommand\Pr{\pr} 

\newcommand\Lem{Lemma}
\newcommand\Prop{Proposition}
\newcommand\Thm{Theorem}

\newcommand\Cor{Corollary}
\newcommand\Sec{Section}

\newcommand{\ceil}[1]{\left\lceil#1\right\rceil}

 \def\G{{\vec G}}

\def\ex{{\mathbb E}}
\def\pr{{\mathbb P}}

\newcommand{\remove}[1]{}

\newcommand{\zeroplus}{V_{0}^+}

\newcommand{\zeroplusi}[1]{V_{0}^+[{#1}]}

\begin{document}
	
\title{Noisy group testing via spatial coupling}
	
\thanks{Amin Coja-Oghlan and Lena Krieg are supported by DFG CO 646/3 and DFG CO 646/5.
	Max Hahn-Klimroth and Olga Scheftelowitsch are supported by DFG CO 646/5.}

\author{Amin Coja-Oghlan, Max Hahn-Klimroth, Lukas Hintze, Dominik Kaaser, Lena Krieg, Maurice Rolvien, Olga~Scheftelowitsch}
\address{Amin Coja-Oghlan, {\tt amin.coja-oghlan@tu-dortmund.de}, TU Dortmund, Faculty of Computer Science and Faculty of Mathematics, 12 Otto-Hahn-St, Dortmund 44227, Germany.}
\address{Max Hahn-Klimroth, {\tt maximilian.hahnklimroth@tu-dortmund.de}, TU Dortmund, Faculty of Computer Science, 12 Otto-Hahn-St, Dortmund 44227, Germany.}
\address{Lukas Hintze, {\tt lukas.rasmus.hintze@uni-hamburg.de}, Universit\"at Hamburg, Fakult\"at f\"ur Mathematik, Informatik und Naturwissenschaften, Informatik, Vogt-K\"olln-Str.\ 30, 22527 Hamburg, Germany.}
\address{Lena Krieg, {\tt lena.krieg@tu-dortmund.de}, TU Dortmund, Faculty of Computer Science, 12 Otto-Hahn-St, Dortmund 44227, Germany.}
\address{Maurice Rolvien, {\tt maurice.rolvien@tu-dortmund.de}, TU Dortmund, Faculty of Computer Science, 12 Otto-Hahn-St, Dortmund 44227, Germany.}
\address{Dominik Kaaser, {\tt dominik@kaaser.at}, TU Hamburg, Institute for Data Engineering, E-19, Blohmstra\ss e 15, 21079 Hamburg, Germany}
\address{Olga Scheftelowitsch, {\tt olga.scheftelowitsch@tu-dortmund.de}, TU Dortmund, Faculty of Computer Science, 12 Otto-Hahn-St, Dortmund 44227, Germany.}

\begin{abstract}%
	We study the problem of identifying a small set $k\sim n^\theta$, $0<\theta<1$, of infected individuals within a large population of size $n$ by testing groups of individuals simultaneously.
	All tests are conducted concurrently.
	The goal is to minimise the total number of tests required.
	In this paper we make the (realistic) assumption that tests are noisy, i.e.\ that a group that contains an infected individual may return a negative test result or one that does not contain an infected individual may return a positive test results with a certain probability.
	The noise need not be symmetric.
	We develop an algorithm called \SPARC\ that correctly identifies the set of infected individuals up to $o(k)$ errors with high probability with the asymptotically minimum number of tests.
	Additionally, we develop an algorithm called \SPEX\ that exactly identifies the set of infected individuals \whp with a number of tests that matches the information-theoretic lower bound for the constant column design, a powerful and well-studied test design.
\hfill {\em MSc: 05C80,	62B10, 68P30, 68R05}
\end{abstract}

\maketitle

\section{Introduction}\label{sec_intro}

\subsection{Background and motivation.}\label{sec_motiv}
Few mathematical disciplines offer as abundant a supply of easy-to-state but hard-to-crack problems as combinatorics does.
Group testing is a prime example.
The problem goes back to the 1940s~\cite{Dorfman_1943}.
Within a population of size $n$ we are to identify a subset of $k$ infected individuals.
To this end we test groups of individuals simultaneously.
In an idealised scenario called `noiseless group testing' each test returns a positive result if and only if at least one member of the group is infected.
All tests are conducted in parallel.
The problem is to devise a (possibly randomised) test design that minimises the total number of tests required.

Noiseless group testing has inspired a long line of research, which has led to optimal or near-optimal results for several parameter regimes~\cite{Aldridge_2019,opt}.
But the assumption of perfectly accurate tests is unrealistic.
Real tests are noisy~\cite{Long}.
More precisely, in medical terms the {\em sensitivity} of a test is defined as the probability that a test detects an actual infection, viz.\ that a group that contains an infected individual displays a positive test result.
Moreover, the {\em specificity} of a test refers to the probability that a healthy group returns a negative test result.
If these accuracies are reasonable high (say 99\%), one might be tempted to think that noiseless testing provides a good enough approximation.
Yet remarkably we will discover that this is far from correct.
Even a seemingly tiny amount of noise has an enormous impact not only on the number of tests required, but also on the choice of a good test design; we will revisit this point in \Sec~\ref{sec_sim}.
Hence, in group testing, like in several other inference problems, the presence of noise adds substantial mathematical depth.
As a rough analogy think of error-correcting linear codes.
In the absence of noise the decoding problem just boils down to solving linear equations.
By contrast, the noisy version, the closest vector problem, is NP-hard~\cite{Dinur}.

In the present paper we consider a very general noise model that allows for arbitrary sensitivities and specificities.
To be precise, we assume that if a test group contains an infected individual, then the test displays a positive result with probability $p_{11}$ and a negative result with probability $p_{10}=1-p_{11}$.
Similarly, if the group consists of healthy individuals only, then the test result will display a negative outcome with probability $p_{00}$ and positive result with probability $p_{01}=1-p_{00}$.
Every test is subjected to noise independently.

Under the common assumption that the number $k$ of infected individuals scales as a power $k\sim n^\theta$ of the population size $n$ with an exponent $0<\theta<1$ we contribute new {\em approximate} and {\em exact recovery algorithms} \SPARC\ and \SPEX.
These new algorithms come with randomised test designs.
We will identify a threshold $m_{\SPARC}=m_{\SPARC}(n,k,\vec p)$ such that \SPARC\ correctly identifies the set of infected individuals up to $o(k)$ errors with high probability over the choice of the test design, provided that at least $(1+\eps)m_{\SPARC}$ tests are deployed.
\SPARC\ is efficient, i.e.\ has polynomial running time in terms of $n$.
By contrast, we will prove that with $(1-\eps)m_{\SPARC}$ tests it is {\em impossible} to identify the set of infected individuals up to $o(k)$ errors, regardless the choice of test design or the running time allotted.
In other words, \SPARC\ solves the approximate recovery problem optimally.

In addition, we develop a polynomial time algorithm \SPEX\ that correctly identifies the status of {\em all} individuals \whp, provided that at least $(1+\eps)m_{\SPEX}$ tests are available, for a certain $m_{\SPEX}(n,k,\vec p)$.
Exact recovery has been the focus of much of the previous literature on group testing~\cite{Aldridge_2019}.
In particular, for noisy group testing the best previous exact recovery algorithm is the {\tt DD} algorithm from~\cite{Maurice}.
{\tt DD} comes with a simple randomised test design called the {\em constant column design}.
Complementing the positive result on \SPEX, we show that on the constant column design exact recovery is information-theoretically impossible with $(1-\eps)m_{\SPEX}$ tests.
As a consequence, the number $m_{\SPEX}$ of tests required by \SPEX\ is an asymptotic lower bound on the number of tests required by {\em any} algorithm on the constant column design, including {\tt DD}.
Indeed, as we will see in \Sec~\ref{sec_sim}, for most choices of the specificity/sensitivity and of the infection density $\theta$, \SPEX\ outperforms {\tt DD} dramatically.

Throughout the paper we write $\vec p=(p_{00},p_{01},p_{10},p_{11})$ for the noisy channel to which the test results are subjected.
We may assume without loss that
\begin{align}\label{eqnoise}
	p_{11}&>p_{01},
\end{align}
i.e.\ that a test is more likely to display a positive result if the test group actually contains an infected individual; for otherwise we could just invert the test results.
A test design can be described succinctly by a (possibly random) bipartite graph $G=(V\cup F,E)$, where $V$ is a set of $n$ individuals and $F$ is a set of $m$ tests.
Write $\SIGMA\in\{0,1\}^V$ for the (unknown) vector of Hamming weight $k$ whose $1$-entries mark the infected individuals.
Further, let $\aSIGMA\in\{0,1\}^F$ be the vector of {\em actual test results}, i.e.\ $\aSIGMA_a=1$ if and only if $\SIGMA_v=1$ for at least one individual $v$ in test $a$.
Finally, let $\dSIGMA\in\{0,1\}^F$ be the vector of {\em displayed tests results}, where noise has been applied to $\aSIGMA$, i.e.,
\begin{align}\label{eqnoisemodel}
		\pr\brk{\dSIGMA_a=\sigma''\mid G,\,\aSIGMA_a=\sigma'}&=p_{\sigma'\,\sigma''}&&\mbox{independently for every }a\in F.
	\end{align}
The objective is to infer $\SIGMA$ from $\dSIGMA$ given $G$.
As per common practice in the group testing literature, we assume throughout that $\kdef$ and that $k$ and the channel $\vec p$ are known to the algorithm~\cite{Aldridge_2019}.%
		\footnote{These assumption could be relaxed at the expense of increasing the required number of tests (details omitted).}

\subsection{Approximate recovery}\label{sec_intro_apx}
The first main result provides an algorithm that identifies the status of all but $o(k)$ individuals correctly with the optimal number of tests.
For a number $z\in[0,1]$ let $h(z)=-z\log(z)-(1-z)\log(1-z)$.
Further, for $y,z\in[0,1]$ let $\KL yz=y\log(y/z)+(1-y)\log((1-y)/(1-z))$ signify the Kullback-Leibler divergence. 
We use the convention that $0\log0=0\log\frac00=0$.
Given the channel $\vec p$ define
\begin{align}\label{eqPhil}
	\phi&=\phi(\vec p)=\frac{h(p_{00})-h(p_{10})}{p_{00}-p_{10}}\enspace,&
	\cchan&=\cchan(\vec p)=\frac1{\KL{p_{10}}{(1-\tanh(\phi/2))/2}}\enspace.
\end{align}
The value $1/\cchan=\KL{p_{10}}{(1-\tanh(\phi/2))/2}$ equals the capacity of the $\vec p$-channel~\cite[\Lem~F.1]{Maurice}.
Let $$m_{\SPARC}(n,k,\vec p)=\cchan k\log(n/k).$$

\begin{theorem}\label{thm_alg_apx}
	For any $\vec p$, $0<\theta<1$ and $\eps>0$ there exists $n_0=n_0(\vec p,\theta,\eps)$ such that for every $n>n_0$ there exist a randomised test design $\Gsc$ with $m\leq(1+\eps)m_{\SPARC}(n,k,\vec p)$ tests and a deterministic polynomial time inference algorithm $\ASPIV$ such that
	\begin{align}\label{eq_alg_apx}
		\pr\brk{\|\ASPIV(\Gsc,\dSIGMA_{\Gsc})-\SIGMA\|_1<\eps k}>1-\eps.
	\end{align}
\end{theorem}

In other words, once the number of tests exceeds $m_{\SPARC}=\cchan k\log(n/k)$, \SPARC\ applied to the test design $\Gsc$ identifies the status of all but $o(k)$ individuals correctly \whp.
The test design $\Gsc$ employs an idea from coding theory called `spatial coupling'~\cite{Takeuchi_2011,Wang}.
As  we will elaborate in \Sec~\ref{sec_over}, spatial coupling blends a randomised and a topological construction.
A closely related design has been used in noiseless group testing~\cite{opt}.

The following theorem shows that \Thm~\ref{thm_alg_apx} is optimal in the strong sense that it is information-theoretically impossible to approximately recover the set of infected individuals with fewer than $(1-\eps)m_{\SPARC}$ tests on \emph{any} test design.
In fact, approximate recovery \whp is impossible even if we allow \emph{adaptive} group testing where tests are conducted one-by-one and the choice of the next group to be tested may depend on all previous results.

\begin{theorem}\label{thm_inf_apx}
	For any $\vec p$, $0<\theta<1$ and $\eps>0$ there exist $\delta=\delta(\vec p,\theta,\eps)>0$ and $n_0=n_0(\vec p,\theta,\eps)$ such that for all $n>n_0$, all adaptive test designs with $m\leq(1-\eps)m_{\SPARC}(n,k,\vec p)$ tests in total and any function $\cA:\{0,1\}^m\to\{0,1\}^n$ we have
	\begin{align}\label{eq_inf_apx}
		\pr\brk{\|\cA(\dSIGMA)-\SIGMA\|_1<\delta k}<1-\delta.
	\end{align}
\end{theorem}

\Thm~\ref{thm_inf_apx} and its proof are a relatively simple adaptation of \cite[\Cor~2.3]{Maurice}, where $\cchan k\log(n/k)$ was established as a lower bound on the number of tests required for {\em exact} recovery.

\subsection{Exact recovery}\label{sec_intro_ex}
How many tests are required in order to infer the set of infected individuals precisely, not just up to $o(k)$ mistakes?
Intuitively, apart from an information-theoretic condition such as \eqref{eqPhil}, exact recovery requires a kind of local stability condition.
More precisely, imagine that we managed to correctly diagnose all individuals $y\neq x$ that share a test with individual $x$.
Then towards ascertaining the status of $x$ itself only those tests are relevant that contain $x$ but no other infected individual $y$; for the outcome of these tests hinges on the status of $x$.
Hence, to achieve exact recovery we need to make certain that it is possible to tell the status of $x$ itself from these tests \whp.

The required number of tests to guarantee local stability on the test design $\Gsc$ from \Thm~\ref{thm_alg_apx} can be expressed in terms of a mildly involved optimisation problem.
For $c,d>0$ and $\theta\in(0,1)$ let
\begin{align}\label{eqYinterval}
	\cY(c,d,\theta)&=\cbc{y\in[0,1]:cd(1-\theta)\KL y{\exp(-d)}<\theta}.
\end{align}
This set is a non-empty interval, because $y\mapsto\KL y{\exp(-d)}$ is convex and $y=\exp(-d)\in\cY(c,d,\theta)$.
Let
\begin{align}\label{eqcyz}
	\cyz(d,\theta)&=
	\begin{cases}
		\inf\cbc{c>0:\inf_{y\in\cY(c,d,\theta)}cd(1-\theta)\bc{\KL y{\exp(-d)}+y\KL{p_{01}}{p_{11}}}\geq\theta}&\mbox{if }p_{11}<1,\\
		\inf\{c>0:0\not\in\cY(c,d,\theta)\}&\mbox{otherwise}.
	\end{cases}
\end{align}
If $p_{11}=1$ let $\fz(y)=1$ for all $y\in\cY(c,d,\theta)$.
Further, if $p_{11}<1$ then the function $z\mapsto\KL z{p_{11}}$ is strictly decreasing on $[p_{01},p_{11}]$;
therefore, for any $c>\cyz(d,\theta)$ and $y\in\cY(c,d,\theta)$ there exists a unique $\fz(y)=\fz_{c,d,\theta}(y)\in[p_{01},p_{11}]$ such that
\begin{align}\label{eq_def_zy}
	cd(1-\theta)\bc{\KL y{\exp(-d)}+y\KL{\fz(y)}{p_{11}}}=\theta.
\end{align}
In either case set
\begin{align}\label{eqsep1}
	\cbound(d,\theta)&= 
	\begin{cases}
		\inf\cbc{c>\cyz(d,\theta):\inf_{y\in\cY(c,d,\theta)}cd(1-\theta)\bc{\KL y{\exp(-d)}+y\KL{\fz(y)}{p_{01}}}\geq1}&\mbox{if }p_{01}>0,\\
		\cyz(d,\theta)&\mbox{otherwise}.
	\end{cases}
\end{align}
Finally, define
\begin{align}\label{eqccentre}
	\ccentre(d)&=1\big/\bc{h(p_{00}\exp(-d)+p_{10}(1-\exp(-d)))-\exp(-d) h(p_{00})-(1-\exp(-d))h(p_{10})},\\
	\cex(\theta)&=\inf_{d>0}\max\{\cbound(d,\theta),\ccentre(d)\},\label{eqcex}\\
	m_{\SPEX}(n,k,\vec p)&=\cex(\theta) k\log(n/k).\nonumber
\end{align}

\begin{theorem}\label{thm_alg_ex}
	For any $\vec p$, $0<\theta<1$ and $\eps>0$ there exists $n_0=n_0(\vec p,\theta,\eps)$ such that for every $n>n_0$ there exist a randomised test design $\Gsp$ with $m\leq(1+\eps)m_{\SPEX}(n,k,\vec p)$ tests and a deterministic polynomial time inference algorithm $\SPEX$ such that
	\begin{align}\label{eq_alg_ex}
		\pr\brk{\SPEX(\Gsc,\dSIGMA_{\Gsc})=\SIGMA}>1-\eps.
	\end{align}
\end{theorem}

Like \SPARC\ from \Thm~\ref{thm_alg_apx}, \SPEX\ uses the spatially coupled test design $\Gsc$.
Crucially, apart from the numbers $n$ and $m$ of individuals and tests, the value of $d$ at which the infimum \eqref{eqcex} is attained also enters into the construction of that test design.
Specifically, the average size of a test group equals $dn/k$.
Remarkably, while the optimal value of $d$ for approximate recovery turns out to depend on the channel $\vec p$ only, a different value of $d$ that also depends on $k$ may be the right choice to facilitate exact recovery.
We will revisit this point in \Sec~\ref{sec_sim}.

\subsection{Lower bound on the constant column design}\label{sec_lower}

Unlike in the case of approximate recovery we do not have a proof that the positive result on exact recovery from \Thm~\ref{thm_alg_ex} is optimal for {\em any} choice of test design.
However, we can show that exact recovery with $(1-\eps)\cex k\log(n/k)$ tests is impossible on the {\em constant column design} $\Gcc$.
Under $\Gcc$ each of the $n$ individuals independently joins an equal number $\Delta$ of tests, drawn uniformly without replacement from the set of all $m$ available tests.
Let $\Gcc=\Gcc(n,m,\Delta)$ signify the outcome.
The following theorem shows that exact recovery on $\Gcc$ is information-theoretically impossible with fewer than $m_{\SPEX}$ tests.

%

\begin{theorem}\label{thm_inf_ex}
	For any $\vec p$, $0<\theta<1$ and $\eps>0$ there exists $n_0=n_0(\vec p,\theta,\eps)$ such that for every $n>n_0$ and all $m\leq(1-\eps)m_{\SPEX}(n,k,\vec p)$, $\Delta>0$ and any $\cA_{\Gcc}:\{0,1\}^m\to\{0,1\}^n$ we have
	\begin{align}\label{eq_inf_ex}
		\pr\brk{\cA_{\Gcc}(\dSIGMA_{\G})=\SIGMA}<\eps.
	\end{align}
\end{theorem}

An immediate implication of \Thm~\ref{thm_inf_ex} is that the positive result from \Thm~\ref{thm_alg_ex} is at least as good as the best prior results on exact noisy recovery from~\cite{Maurice}, which are based on running a simple combinatorial algorithm called {\tt DD} on $\Gcc$.
In fact, in \Sec~\ref{sec_sim} we will see that the new bound from \Thm~\ref{thm_alg_ex} improves over the bounds from~\cite{Maurice} rather significantly for most $\theta,\vec p$.

The proof of \Thm~\ref{thm_inf_ex} confirms the combinatorial meaning of the threshold $\cex(\theta)$.
Specifically, for $c=m/(k\log(n/k))<\ccentre(d)$ from \eqref{eqccentre} a moment calculation reveals that \whp\ $\Gcc$ contains `solutions' $\sigma$ of at least the same posterior likelihood as the true $\SIGMA$ such that $\sigma$ and $\SIGMA$ differ significantly, i.e.\ $\|\SIGMA-\sigma\|_1=\Omega(k)$.
By contrast, the threshold $\cbound(d,\theta)$ marks the onset of local stability.
This means that for $c<\cbound(d,\theta)$ there will be numerous $\sigma$ close to but not identical to $\SIGMA$ (i.e.\ $0<\|\SIGMA-\sigma\|_1=o(k)$) of the at least same posterior likelihood.
In either case any inference algorithm, efficient or not, is at a loss identifying the actual $\SIGMA$.

In recent independent work Chen and Scarlett~\cite{scarlettchen} obtained \Thm~\ref{thm_inf_ex} in the special case of symmetric noise (i.e.\ $p_{00}=p_{11}$).
While syntactically their expression for the threshold $\cex(\theta)$ differs from \eqref{eqYinterval}--\eqref{eqcex}, it can be checked that both formulas yield identical results (see Appendix~\ref{apx_scarlettchen}).
Apart from the information-theoretic lower bound (which is the part most relevant to the present work), Chen and Scarlett also proved that it is information-theoretically possible (by means of an exponential algorithm) to infer $\SIGMA$ \whp on the constant column design with $m\geq(1+\eps)\cex(\theta)k\log(n/k)$ tests if $p_{00}=p_{11}$.
Hence, the bound $\cex(\theta)$ is tight in the case of symmetric noise.


\subsection{Examples}\label{sec_sim}
We illustrate the improvements that \Thm s~\ref{thm_alg_apx} and~\ref{thm_alg_ex} contribute by way of concrete examples of channels $\vec p$.
Specifically, for the binary symmetric channel and the $Z$-channel it is possible to obtain partial analytic results on the optimisation behind $\cex(\theta)$ from \eqref{eqcex} (see Appendices~\ref{sec_lukas_bsc}--\ref{sec_lukas_Z}).
As we will see, even a tiny amount of noise has a dramatic impact on both the number of tests required and the parameters that make a good test design.

\subsubsection{The binary symmetric channel}\label{sec_sim_bsc}
Although symmetry is an unrealistic assumption from the viewpoint of applications~\cite{Long}, the expression \eqref{eqPhil} and the optimisations \eqref{eqsep1}--\eqref{eqcex} get much simpler on the binary symmetric channel, i.e.\ in the case $p_{00}=p_{11}$.
For instance, the value of $d$ that minimises $\ccentre(d)$ from \eqref{eqccentre} turns out to be $d=\log2$.
The parameter $d$ enters into the constructions of the randomised test designs $\Gsc$ and $\Gcc$ in a crucial role.
Specifically, the average size of a test group equals $dn/k$.
In effect, any test is actually negative with probability $\exp(-d+o(1))$ (see \Prop~\ref{prop_basic} and \Lem~\ref{lemma_basic} below).
Hence, if $d=\log2$ then \whp\ about half the tests contain an infected individual.
In effect, since $p_{11}=p_{00}$, after the application of noise about half the tests display a positive result \whp.

In particular, in the noiseless case $p_{11}=p_{00}=1$ a bit of calculus reveals that the value $d=\log2$ also minimises the other parameter $\cbound(d,\theta)$ from \eqref{eqsep1} that enters into $\cex(\theta)$ from \eqref{eqcex}.
Therefore, in noiseless group testing $d=\log2$ unequivocally is the optimal choice, the optimisation on $d$ \eqref{eqcex} effectively disappears and we obtain
\begin{align}\label{eqcnoiseless}
	\cex(\theta)&=\max\cbc{\frac\theta{(1-\theta)\log^22},\frac1{\log2}},
\end{align}
thereby reproducing the optimal result on noiseless group testing from~\cite{opt}.

But remarkably, as we verify analytically in Appendix~\ref{sec_lukas_bsc} at positive noise $\frac12<p_{11}=p_{00}<1$ the value of $d$ that minimises $\cbound(d,\theta)$ does \emph{not} generally equal $\log 2$.
Hence, if we aim for exact recovery, then at positive noise it is no longer optimal for all $0<\theta<1$ to aim for about half the tests being positive/negative. 
The reason is the occurrence of a phase transition in terms of $\theta$ where the `local stability' term $\cbound(d,\theta)$ takes over as the overall maximiser in \eqref{eqccentre}.
Consequently, the objective of minimising $\cbound(d,\theta)$ and the optimal choice $d=\log2$ for $\ccentre(d)$ clash.
In effect, the overall minimiser $d$ for $\cex(d,\theta)$ depends on both $\vec p$ and the infection density parameter $\theta$ in a non-trivial way.
Thus, the presence and level of noise has a discernible impact on the choice of a good test design.%
\footnote{This observation confirms a hypothesis stated in~\cite[Appendix~F]{Maurice}.
	As mentioned in \Sec~\ref{sec_lower}, independent work of Chen and Scarlett~\cite{scarlettchen} on the case of symmetric noise reaches the same conclusion.}

\begin{figure}
	\begin{subfigure}[l]{0.49\textwidth}
		\includegraphics[height=63mm]{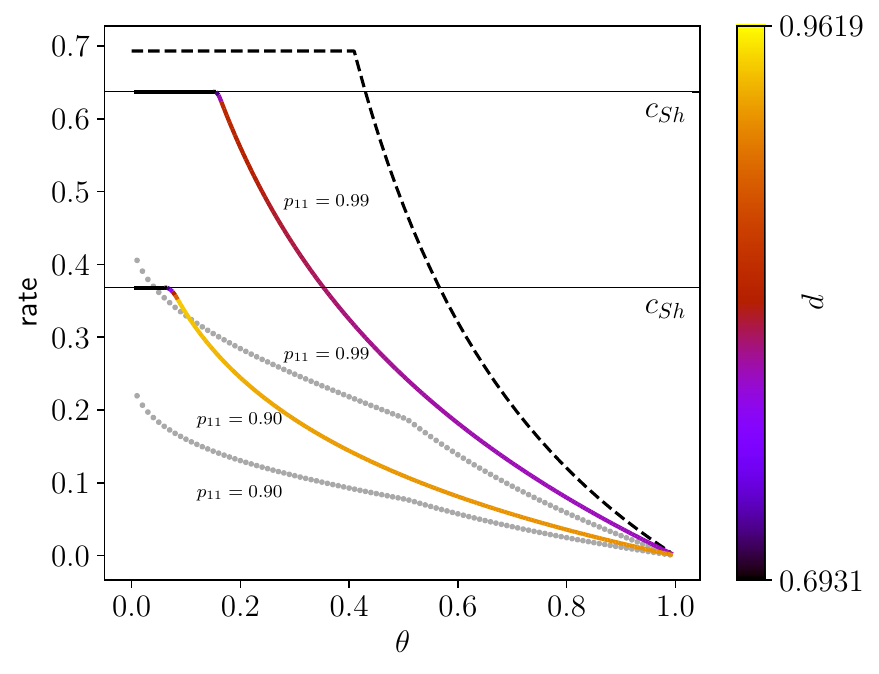}
		\caption{Rates on the binary symmetric channel for \\$p_{00}=p_{11}=0.99$ and $p_{00}=p_{11}=0.9$.}
		\label{left}
	\end{subfigure}
	\begin{subfigure}[r]{0.49\textwidth}
		\includegraphics[height=63mm]{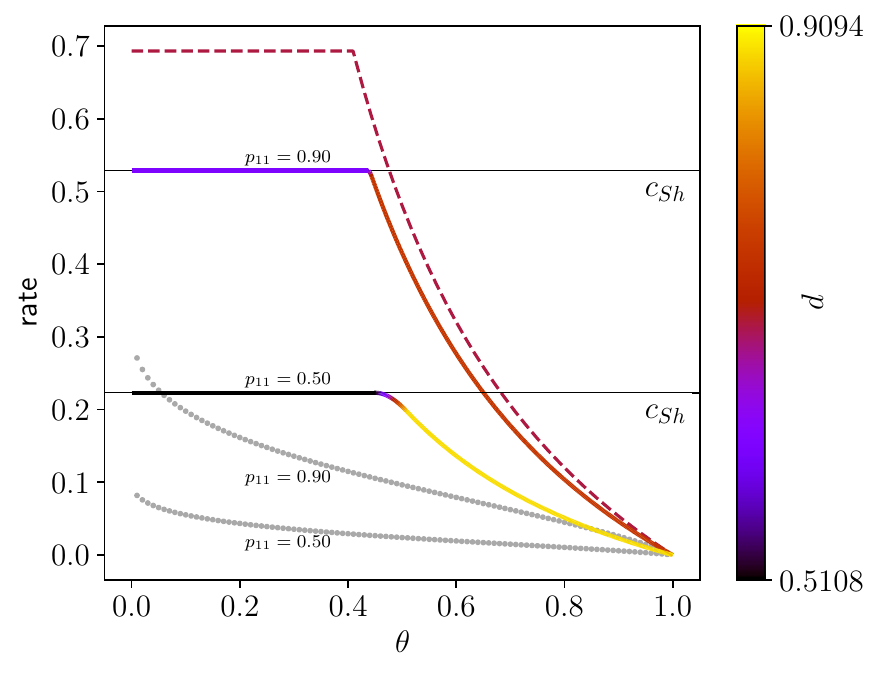}
		\caption{Rates on the $Z$-channel ($p_{00}=1$) with $p_{11}=0.9$ and \\$p_{11}=0.5$.}
		\label{right}
	\end{subfigure}
	\caption{Information rates on different channels in nats. The horizontal axis displays the infection density parameter $0<\theta<1$. The colour indicates the optimal value of $d$ for a given $\theta$.
	}\label{fig_sym}
\end{figure}

Figure~\ref{fig_sym} displays the performance of the algorithms $\SPARC$ and $\SPEX$ from \Thm s~\ref{thm_alg_apx} and~\ref{thm_alg_ex} on the binary symmetric channel.
For the purpose of graphical representation the figure does not display the values of $\cex(\theta)$, which diverge as $\theta\to1$, but the value $1/\cex(\theta)$.
This value has a natural information-theoretic interpretation: it is the average amount of information that a test reveals about the set of infected individuals, measured in `nats'.
In other words, the plots display the information rate of a single test (higher is better).
The optimal values of $d$ are colour-coded into the curves.
While in the noiseless case $d=\log2$ remains constant, in the noisy cases $d$ varies substantially with both $\theta$ and $p_{00}=p_{11}$.

For comparison the figure also displays the rate in the noiseless case (dashed line on top) and the best previous rates realised by the {\tt DD} algorithm on $\Gcc$ from~\cite{Maurice} (dotted lines).
As is evident from Figure~\ref{fig_sym}, even noise as small as $1\%$ already reduces the rate markedly: the upper coloured curve remains significantly below the noiseless black line.
That said, Figure~\ref{fig_sym} also illustrates how the rate achieved by \SPEX\ improves over the best previous algorithm {\tt DD} from~\cite{Maurice}.
Somewhat remarkably, the $10\%$-line for $\cex(\theta)$ intersects the $1\%$-line for {\tt DD} for an interval of $\theta$.
Hence, for these $\theta$ the algorithm from \Thm~\ref{thm_alg_ex} at $10\%$ noise requires fewer tests than {\tt DD} at $1\%$.

Figure~\ref{fig_sym} also illustrates how approximate and exact recovery compare.
Both coloured curves start out as black horizontal lines.
These bits of the curves coincide with the rate of the \SPARC\ algorithm from \Thm~\ref{thm_alg_apx}.
The rate achieved by \SPARC, which does not depend on $\theta$, is therefore just the extension of this horizontal line to the entire interval $0<\theta<1$ at the height $\cchan$ from  \eqref{eqPhil}.
Hence, particularly for large $\theta$ approximate recover achieves {\em much} better rates than exact recovery.

\subsubsection{The $Z$-channel}\label{sec_sim_Z}
In the case $p_{00}=1$ of perfect specificity,  known as the $Z$-channel, it is possible to derive simple expressions for the optimisation problems~\eqref{eqYinterval}--\eqref{eqcex} (see Appendix~\ref{sec_lukas_Z}):
%
\begin{align}
	\cbound(d,\theta)&= -\frac{\theta}{d(1-\theta)  \log(1 - \exp(-d) \pii)},\label{eqLukZ1}\\
	\ccentre(d)&= \bc{h(p_{10}+(1-p_{10})\exp(-d))-(1-\exp(-d))h(p_{10})}^{-1}\label{eqLukZ2}.
\end{align}
As in the symmetric case, there is a tension between the value of $d$ that minimises \eqref{eqLukZ1} and the objective of minimising \eqref{eqLukZ2}.
Recall that since the size of the test groups is proportional to $d$, the optimiser $d$ has a direct impact on the construction of the test design.

Figure~\ref{right} displays the rates achieved by of \SPEX\ (solid line) and, for comparison, the {\tt DD} algorithm from~\cite{Maurice} (dotted grey) on the $Z$-channel with sensitivities $p_{11}=0.9$ and $p_{11}=0.5$.
Additionally, the dashed red line indicates the noiseless rate.
Once again the optimal value of $d$ is colour-coded into the solid \SPEX\ line.
Remarkably, the \SPEX\ rate at $p_{11}=0.5$ (high noise) exceeds the {\tt DD} rate at $p_{11}=0.9$ for a wide range of $\theta$.
As in the symmetric case the horizontal $\cchan$-lines indicate the performance of the \SPARC\ approximate recovery algorithm.

\subsubsection{General (asymmetric) noise}\label{sec_sim_gen}
While in the symmetric case the term $\ccentre(d)$ from \eqref{eqccentre} attains its minimum simply at $d=\log 2$, with $\phi$ from \eqref{eqPhil} the minimum for general $\vec p$ is attained at 
\begin{align}\label{eqdcentre}
	d&=\dchan=\log(p_{11}-p_{01})-\log\bc{(1-\tanh(\phi/2))/2-p_{10}}&&\mbox{\cite[\Lem~F.1]{Maurice}}.
\end{align}
Once again the design of $\Gsc$ (as well as $\Gcc$) ensures that \whp a $\exp(-d)$-fraction of tests are actually negative \whp.
The choice \eqref{eqdcentre} ensures that under the $\vec p$-channel the mutual information between the channel input and the channel output is maximised~\cite[\Lem~F.1]{Maurice}:
\begin{align}\label{eqcchandchan}
	\frac1{\cchan}=\frac1{\ccentre(\dchan)}=\KL{p_{10}}{(1-\tanh(\phi/2))/2},
\end{align}
As can be checked numerically, the second contributor $\cbound(d,\theta)$ to $\cex(\theta)$ may take its minimum at  another $d$.
However, we are not aware of a simple explicit expressions for $\ccentre(\theta)$ from \eqref{eqccentre} for general noise.

\subsection{Related work}\label{sec_related}

\noindent
The monograph of Aldridge, Johnson and Scarlett~\cite{Aldridge_2019} provides an excellent overview of the group testing literature.
The group testing problem comes in various different flavours: non-adaptive (where all tests are conducted concurrently) or adaptive (where tests are conducted in subsequent stages such that the tests at later stages may depend on the outcomes of earlier stages), as well as noiseless or noisy.
An important result of Aldridge shows that noiseless non-adaptive group testing does not perform better than plain individual testing if $k=\Omega(n)$, i.e.\ if the number of infected individuals is linear in the size of the population~\cite{Aldridge_linea2}.
Therefore, research on non-adaptive group testing focuses on the case $k\sim n^\theta$ with $0<\theta<1$.
For non-adaptive noiseless group testing with this scaling of $k$ two different test designs (Bernoulli and constant column) and various elementary algorithms have been proposed~\cite{Chan_2011}.
Among these elementary designs and algorithms the best performance to date is achieved by the {\tt DD} greedy algorithm on the constant column design~\cite{Johnson_2019}.
However, the {\tt DD} algorithm does not match the information-theoretic bound on the constant column design for all $\theta$~\cite{Coja_2019}.

Coja-Oghlan, Gebhard, Hahn-Klimroth and Loick proposed a more sophisticated test design for noiseless group testing based on spatial coupling~\cite{opt}, along with an efficient inference algorithm called {\tt SPIV}.
Additionally, they improved the information-theoretic lower bound for non-adaptive noiseless group testing.
The number of tests required by the {\tt SPIV} algorithm matches this new lower bound. 
In effect, the combination of {\tt SPIV} with the spatially coupled test design solves the noiseless non-adaptive group testing problem optimally both for exact and approximate recovery.

The present article deals with the {\em noisy} non-adaptive variant of group testing.
A noisy version of the efficient {\tt DD} algorithm was previously studied on both the Bernoulli and the constant column design~\cite{Maurice,Sca20}.
The best previous exact recovery results for general noise were obtained by Johnson, Gebhard, Loick and Rolvien~\cite{Maurice} by means of {\tt DD} on the constant column design (see \Thm~\ref{thm_DD} below).
\Thm~\ref{thm_inf_ex} shows in combination with \Thm~\ref{thm_alg_ex} that the new {\tt SPEX} algorithm performs at least as well as {\em any} algorithm on the constant column design, including and particularly {\tt DD}.

Apart from the articles~\cite{Maurice,Sca20} that dealt with the same general noise model as we consider here,
several contributions focused on special noise models, particularly symmetric noise ($p_{00}=p_{11}$). 
In this scenario Chen and Scarlett~\cite{scarlettchen} recently determined the information-theoretically optimal number of tests required for exact recovery on the Bernoulli and constant column designs.
The constant column design outperforms the Bernoulli design.
The information-theoretic threshold identified by Chen and Scarlett matches the threshold $\cex(\theta)$ from \eqref{eqccentre} in the special case of symmetric noise (see Appendix~\ref{apx_scarlettchen}).
However, Chen and Scarlett do not investigate the issue of {\em efficient} inference algorithms.
Instead, they pose the existence of an efficient inference algorithm that matches the information-theoretic threshold as an open problem.
\Thm~\ref{thm_alg_ex} applied to symmetric noise answers their question in the affirmative.

A further contribution of Scarlett and Cevher~\cite{Scarlett_2016} contains a result on approximate recovery under the assumption of symmetric noise.
In this case Scarlett and Cevher obtain matching information-theoretic upper and lower bounds, albeit without addressing the question of efficient inference algorithms.
\Thm~\ref{thm_alg_apx} applied to the special case of symmetric noise provides a polynomial time inference algorithm that matches their lower bound.

From a practical viewpoint non-adaptive group testing (where all tests are conducted in parallel) is preferable because results are available more rapidly than in the adaptive setting, where several rounds of testing are required.
That said, adaptive schemes may require a smaller total number of tests.
The case of noiseless adaptive group testing has been studied since the seminal work of Dorfman~\cite{Dorfman_1943} from the 1940s.
For the case $k\sim n^\theta$ a technique known as generalized binary splitting gets by with the optimal number of tests~\cite{Hwang_1972, Baldassini_2013}. 
Aldridge~\cite{Aldridge_linea1} extended this approach to the case $k=\Theta(n)$, obtaining near-optimal rates.
Recently there has been significant progress on upper and lower bounds for noisy and adaptive group testing, although general optimal results remain elusive~\cite{Scarlett_2018,Teo22}. 

Beyond group testing in recent years important progress has been made on several inference problems by means of a combination of spatial coupling and message passing ideas.
Perhaps the most prominent case in point is the compressed sensing problem~\cite{Donoho_2013,Krzakala_2012}.
Further applications include the pooled data problem~\cite{Lenka_pooled_data,Lenka_pooled_data2,MaxNoela} and CDMA~\cite{Takeuchi_2011}, a signal processing problem.
The basic idea of spatial coupling, which we are going to discuss in some detail in \Sec~\ref{sec_gsc}, goes back to work on capacity-achieving linear codes~\cite{Kumar_2014,Kudekar_2010_2,Kudekar_2013,Takeuchi_2011,Wang}.
The {\tt SPIV} algorithm from~\cite{opt} combines a test design inspired by spatial coupling with a combinatorial inference algorithm.
A novelty of the present work is that we replace this elementary algorithm by a novel variant of the Belief Propagation message passing algorithm~\cite{MM,Pearl} that lends itself to a rigorous analysis.

\subsection{Organisation}
After introducing a bit of notation and recalling some background in \Sec~\ref{Sec_pre} we given an outline of the proofs of the main results in \Sec~\ref{sec_over}.
Subsequently \Sec~\ref{sec_alg_apx} deals with the details of the proof of \Thm~\ref{thm_alg_apx}.
Moreover, \Sec~\ref{sec_exact} deals with the proof of \Thm~\ref{thm_alg_ex}, while in \Sec~\ref{sec_cc_lower} we prove \Thm~\ref{thm_inf_ex}.
The proof of \Thm~\ref{thm_inf_apx}, which is quite short and uses arguments that are well established in the literature, follows in Appendix \ref{sec_inf_apx}.
Appendix~\ref{sec_lemma_endgame_misclassified} contains the proof of a routine expansion property of the test design $\Gsc$.
Finally, in Appendices~\ref{sec_lukas_bsc} and \ref{sec_lukas_Z} we investigate the optimisation problems \eqref{eqsep1}--\eqref{eqcex} on the binary symmetric channel and the $Z$-channel, and in Appendix \ref{apx_scarlettchen} we compare our result to the recent result of Chen and Scarlett \cite{scarlettchen}. 

\subsection{Preliminaries}\label{Sec_pre}
As explained in \Sec~\ref{sec_motiv}, a test design is a bipartite graph $G=(V\cup F,E)$ whose vertex set consists of a set $V$ individuals and a set $F$ of tests.
The ground truth, i.e.\ the set of infected individuals is encoded as a vector $\SIGMA\in\{0,1\}^V$ of Hamming weight $k$.
Since we will deal with randomised test designs, we may assume that $\SIGMA$ is a {\em uniformly random} vector of Hamming weight $k$ (by shuffling the set of individuals).
Also recall that for a test $a$ we let $\aSIGMA_a\in\{0,1\}$ denote the {\em actual} result of test $a$ (equal to one if and only if $a$ contains an infected individual), while $\dSIGMA_a\in\{0,1\}$ signifies the {\em displayed} test result obtained via \eqref{eqnoisemodel}.
It is convenient to introduce the shorthands
\begin{align*}
	V_0&=\cbc{x\in V:\SIGMA_x=0},&
	V_1&=\cbc{x\in V:\SIGMA_x=1},&
	F_0&=\cbc{a\in F:\aSIGMA_a=0},&
	F_1&=\cbc{a\in F:\aSIGMA_a=1},\\
	F^-&=\cbc{a\in F:\dSIGMA_a=0},&
	F^+&=\cbc{a\in F:\dSIGMA_a=1},&
	F_0^\pm&=F^\pm\cap F_0,&
	F_1^\pm&=F^\pm\cap F_1
\end{align*}
for the set of infected/healthy individuals, the set of actually negative/positive tests, the set of negatively/positively displayed tests and the tests that are actually negative/positive and display a positive/negative result, respectively.
For each node $u$ of $G$ we denote by $\partial u=\partial_Gu$ the set of neighbours of $u$.
For an individual $x\in V$ we also let $\partial^\pm x=\partial_G^\pm x=\partial_G x\cap F^\pm$ be the set of positively/negatively displayed tests that contain $x$.

We need Chernoff bounds for the binomial and the hypergeometric distribution.
Recall that the {\em hypergeometric distribution} $\Hyp(L,M,N)$ is defined by
\begin{align}\label{eqHyp}
	\pr\brk{\Hyp(L,M,N)=k}&=\binom Mk\binom{L-M}{N-k}\binom LN^{-1}&&(k\in\{0,1,\ldots,\min\{M,N\}\}).
\end{align}
(Out of a total of $L$ items of which $M$ are special we draw $N$ items without replacement and count the number of special items in the draw.)
The mean of the hypergeometric distribution equals $MN/L$.

\begin{lemma}[{\cite[Equation 2.4]{Janson_2011}} ] \label{lem_chernoff}
Let $\vX$ be a binomial random variable with parameters $N,p$.
Then
\begin{align}\label{eqChernoff1}
    \Pr\brk{\vX \geq {qN}} &\leq \exp \bc{-N\KL{q}{p}} \quad \text{for $p<q<1$,} \\
    \Pr\brk{\vX \leq {qN}} &\leq \exp \bc{-N\KL{q}{p}} \quad \text{for $0<q<p$.}\label{eqChernoff2}
\end{align}
\end{lemma}

\begin{lemma}[\cite{Hoeffding}] \label{lem_hyperchernoff}
	For a  hypergeometric variable $\vX\sim\Hyp(L,M,N)$ the bounds \eqref{eqChernoff1}--\eqref{eqChernoff2} hold with $p=M/L$.
\end{lemma}

Throughout we use asymptotic notation $o(\nix),\omega(\nix),O(\nix),\Omega(\nix),\Theta(\nix)$ to refer to limit $n\to\infty$.
It is understood that the constants hidden in, e.g., a $O(\nix)$-term may depend on $\theta,\vec p$ or other parameters, and that a $O(\nix)$-term may have a positive or a negative sign.
To avoid case distinctions we sometimes take the liberty of calculating with the values $\pm\infty$.
The usual conventions $\infty+\infty=\infty\cdot\infty=\infty$ and $0\cdot\infty=0$ apply.
Furthermore, we set $\tanh(\pm\infty)=\pm1$.
Also recall that $0\log0=0\log\frac00=0$.
Additionally, $\log0=-\infty$ and $\frac 1 0 = \log\frac10=\infty$.

Finally, for two random variables $\vX,\vY$ defined on the same finite probability space $(\Omega,\pr\brk\nix)$ we write
\begin{align*}
I(\vX,\vY)&=\sum_{\omega,\omega'\in\Omega}\pr\brk{\vX=\omega,\,\vY=\omega'}\log\frac{\pr\brk{\vX=\omega,\,\vY=\omega'}}{\pr\brk{\vX=\omega}\pr\brk{\vY=\omega'}}
\end{align*}
for the {\em mutual information} of $\vX,\vY$.
We recall that $I(\vX,\vY)\geq0$.

\section{Overview}\label{sec_over}

\noindent
We proceed to survey the functioning of the algorithms \SPARC\ and \SPEX. 
To get started we briefly discuss the best previously known algorithm for noisy group testing, the {\tt DD} algorithm from~\cite{Maurice}, which operates on the constant column design.
We will discover that {\tt DD} can be viewed as a truncated version of the Belief Propagation (BP) message passing algorithm.
BP is a generic heuristic for inference problems backed by physics intuition~\cite{MM,Zdeborova_2016}.
Yet unfortunately BP is notoriously difficult to analyse.
Even worse, it seems unlikely that full Belief Propagation will significantly outperform {\tt DD} on the constant column design; evidence of this was provided in~\cite{ilias} in the noiseless case.
The basic issue is the lack of a good initialisation of the BP messages.

To remedy this issue we resort to the spatially coupled test design $\Gsc$, which combines a randomised and a spatial construction.
The basic idea resembles domino toppling.
Starting from an easy-to-diagnose `seed', the algorithm works its way forward in a well-defined direction until all individuals have been diagnosed.
Spatial coupling has proved useful in related inference problems, including noiseless group testing~\cite{opt}.
Therefore, a natural stab at group testing would be to run BP on $\Gsc$.
Indeed, the BP intuition provides a key ingredient to the \SPARC\ algorithm, namely the update equations (see \eqref{eqWtau} below), of which the correct choice of the weights (Eq.\ \eqref{eqweights} below) is the most important ingredient.
But in light of the difficulty of analysing textbook BP, \SPARC\ relies on a modified version of BP that better lends itself to a rigorous analysis.
Furthermore, the \SPEX\ algorithm for exact recovery combines \SPARC\ with a clean-up step.

\SPARC\ and \SPEX\ can be viewed as generalised versions of the noiseless group testing algorithm called \SPIV\ from~\cite{opt}.
However, \cite{opt} did not exploit the connection with BP.
Instead, in the noiseless case the correct weights were simply `guessed' based on combinatorial intuition, an approach that it seems difficult to generalise.
Hence, the present, systematic derivation of the weights \eqref{eqweights} also casts new light on the noiseless case.
In fact, we expect that the paradigm behind \SPARC\ and \SPEX, namely to use BP heuristically to find the correct parameters for a simplified message passing algorithm, potentially generalises to other inference problems as well.

\subsection{The {\tt DD} algorithm}\label{sec_dd}

\noindent
The {\tt DD} algorithm from~\cite{Maurice} utilises the constant column design $\Gcc$.%
	\footnote{The article~\cite{Maurice} also investigates the performance of {\tt DD} on the Bernoulli design, which turns out to be inferior.}
Thus, each individual independently joins an equal number $\Delta$ of random tests.
Given the displayed test results, {\tt DD} first declares certain individuals as uninfected by thresholding the number of negatively displayed tests.
More precisely, {\tt DD} declares as uninfected any individual that appears in at least $\alpha\Delta$ negatively displayed tests, with $\alpha$ a diligently chosen threshold.
Having identified the respective individuals as uninfected, {\tt DD} looks out for tests $a$ that display a positive result and that only contain a single individual $x$ that has not been identified as uninfected yet.
Since such tests $a$ hint at $x$ being infected, in its second step {\tt DD} declares as infected any individual $x$ that appears in at least $\beta\Delta$ positively displayed tests $a$ where all other individuals $y\in\partial a\setminus x$ were declared uninfected by the first step.
Once again $\beta$ is a carefully chosen threshold.
Finally, {\tt DD} declares as uninfected all remaining individuals.

The {\tt DD} algorithm exactly recovers the infected set \whp provided the total number $m$ of tests is sufficiently large such that the aforementioned thresholds $\alpha,\beta$ exist.
The required number of tests, which comes out in terms of a mildly delicate optimisation problem, was determined in~\cite{Maurice}.
Let
\begin{align}\label{eqq}
	q_0^-&=\exp(-d)p_{00}+(1-\exp(-d))p_{10},&q_0^+&=\exp(-d)p_{01}+(1-\exp(-d))p_{11}.
\end{align}

\begin{theorem}[{\cite[\Thm~2.2]{Maurice}}]\label{thm_DD}
	Let $\eps>0$ and with $\alpha\in(p_{10},q_0^-)$ and $\beta\in(0,\exp(-d)p_{11})$ let
		\begin{align*}
			c_{\DD}&=\min_{\alpha,\beta,d} \max\cbc{c_{\DD,1}(\alpha,d),c_{\DD,2}(\alpha,d),c_{\DD,3}(\beta,d),c_{\DD,4}(\alpha,\beta,d)},\qquad\mbox{where}\\
			c_{\DD,1}(\alpha,d)&=\frac{\theta}{(1-\theta)\KL\alpha{p_{10}}},\qquad
			c_{\DD,2}(\alpha,d)=\frac1{d\KL{\alpha}{q_0^-}},\qquad
			c_{\DD,3}(\beta,d)=\frac{\theta}{d(1-\theta)\KL{\beta}{p_{11}\exp(-d)}},\\
			c_{\DD,4}(\alpha,\beta,d)&=\max_{(1-\alpha)\vee\beta\leq z\leq1}\bc{d(1-\theta)\bc{\KL z{q_0^+}+\vecone\cbc{\beta>\frac{z\exp(-d)p_{01}}{q_0^+}}z\KL{\beta/z}{\exp(-d)p_{01}/q_0^+}}}^{-1}.
		\end{align*}
	If $m\geq(1+\eps)c_{\DD}k\log(n/k)$, then there exists $\Delta>0$ and $0\leq \alpha,\beta\leq1$ such that the {\tt DD} algorithm outputs $\SIGMA$ \whp.
\end{theorem}

The distinct feature of {\tt DD} is its simplicity.
However, the thresholding that {\tt DD} applies does seem to leave something on the table.
For a start, whether {\tt DD} identifies a certain individual $x$ as infected depends only on the results of tests that have distance at most three from $x$ in the graph $\Gcc$.
Moreover, it seems wasteful that {\tt DD} takes only those positively displayed tests into consideration where all but one individual were already identified as uninfected.

\subsection{Belief Propagation}\label{sec_bp}
Belief Propagation is a message passing algorithm that is expected to overcome these deficiencies.
In fact, heuristic arguments suggests that BP might be the ultimate recovery algorithm for a wide class of inference algorithms on random graphs~\cite{Zdeborova_2016}.
That said, rigorous analyses of BP are few and far between.

Following the general framework from~\cite{MM}, in order to apply BP to a group testing design $G=(V\cup F,E)$ we equip each test $a\in F$ with a weight function
\begin{align}\label{eqpsidef}
	\psi_a=\psi_{G,\dSIGMA,a}&:\{0,1\}^{\partial a}\to\RRpos,&
	\sigma_{\partial a_i}&\mapsto
	\begin{cases}
		\vecone\{\|\sigma\|_{1}=0\}p_{00}+\vecone\{\|\sigma\|_{1}>0\}p_{10}&\mbox{ if $a\in F^-$},\\
		\vecone\{\|\sigma\|_{1}=0\}p_{01}+\vecone\{\|\sigma\|_{1}>0\}p_{11}&\mbox{ if $a\in F^+$}.
		\end{cases}
\end{align}
Thus, $\psi_a$ takes as argument a $\{0,1\}$-vector $\sigma=(\sigma_x)_{x\in\partial a}$ indexed by the individuals that take part in test $a$.
The weight $\psi_a(\sigma)$ equals the probability of observing the result that $a$ displays if the infection status were $\sigma_x$ for every individual $x\in\partial a$.
In other words, $\psi_a$ encodes the posterior under the $\vec p$-channel.
Given $G,\dSIGMA$ the weight functions give rise to the total weight of $\sigma\in\{0,1\}^V$ by letting
\begin{align}\label{eqpsiG}
	\psi(\sigma)=\psi_{G,\dSIGMA}(\sigma)&=\vecone\cbc{\|\sigma\|_1=k}\prod_{a\in F}\psi_{a}(\sigma_{\partial a}).
\end{align}
Thus, we just multiply up the contributions \eqref{eqpsidef} of the various tests, and add in the prior assumption that precisely $k$ individuals are infected.
The total weight \eqref{eqpsiG} induces a probability distribution
\begin{align}\label{eqBoltzmann}
	\mu_{G,\dSIGMA}(\sigma)&=\psi_{G,\dSIGMA}(\sigma)/Z_{G,\dSIGMA},&\mbox{where}&&Z_{G,\dSIGMA}&=\sum_{\sigma\in\{0,1\}^V}\psi_{G,\dSIGMA}(\sigma).
\end{align}
A simple application of Bayes' rule shows that $\mu_G$ matches the posterior of the ground truth $\SIGMA$ given the test results.

\begin{fact}\label{fact_Nishi}
	For any test design $G$ we have $\mu_{G,\dSIGMA}(\sigma)=\pr\brk{\SIGMA=\sigma\mid \dSIGMA}$.
\end{fact}

BP is a heuristic to calculate the marginals of $\mu_{G,\dSIGMA}$ or, in light of Fact~\ref{fact_Nishi}, the posterior probabilities $\pr[\SIGMA_{x_i}=1\mid\dSIGMA]$.
To this end BP associates {\em messages} with the edges of $G$.
Specifically, for any adjacent individual/test pair $x,a$ there is a message $\mu_{x\to a,t}(\nix)$ from $x$ to $a$, and another one $\mu_{a\to x,t}(\nix)$ in the reverse direction.
The messages are updated in rounds and therefore come with a time parameter $t\in\ZZpos$.
Moreover, being probability distributions on $\{0,1\}$, the messages always satisfy
\begin{align}\label{eq_msg_norm}
	\mu_{x\to a,t}(0)+\mu_{x\to a,t}(1)=\mu_{a\to x,t}(0)+\mu_{a\to x,t}(1)=1.
\end{align}
The intended semantics is that $\mu_{x\to a,t}(1)$ estimates the probability that $x$ is infected given all known information except the result of test $a$.
Analogously, $\mu_{a\to x,t}(1)$ estimates the probability that $x$ is infected if we disregard all other tests $b\in\partial x\setminus\{a\}$.

This slightly convoluted interpretation of the messages facilitates simple heuristic formulas for computing the messages iteratively.
To elaborate, in the absence of a better a priori estimate, at time $t=0$ we simply initialise in accordance with the prior, i.e.,
\begin{align}\label{eqBPinit}
	\mu_{x\to a,0}(0)&=1-k/n&\mu_{x\to a,0}(1)&=k/n.
\end{align}
Subsequently we use the weight function \eqref{eqpsidef} to update the messages: inductively for $t\geq1$ and $r\in\{0,1\}$ let
\begin{align}\label{eqBPupdate1}
	\mu_{a\to x,t}(r)&\propto
		\begin{cases}
			p_{11}&\mbox{if }r=1,\,a\in F^+,\\
			p_{11}+(p_{01}-p_{11})\prod_{y\in\partial a\setminus\{x\}}\mu_{y\to a,t-1}(0)&\mbox{if }r=0,\,a\in F^+,\\
			p_{10}&\mbox{if }r=1,\,a\in F^-,\\
			p_{10}+(p_{00}-p_{10})\prod_{y\in\partial a\setminus\{x\}}\mu_{y\to a,t-1}(0)&\mbox{if }r=0,\,a\in F^-,
		\end{cases}\\
	\mu_{x\to a,t}(r)&\propto\bcfr{k}{n}^r\bc{1-\frac kn}^{1-r}\prod_{b\in \partial x\setminus\{a\}}\mu_{b\to x,t-1}(r).\label{eqBPupdate2}
\end{align}
The $\propto$-symbol hides the necessary normalisation to ensure that the messages satisfy \eqref{eq_msg_norm}.
Furthermore, the $(k/n)^r$ and $(1-k/n)^{1-r}$-prefactors in \eqref{eqBPupdate2} encode the prior that precisely $k$ individuals are infected.
The expressions \eqref{eqBPupdate1}--\eqref{eqBPupdate2} are motivated by the hunch that for most tests $a$ the values $(\SIGMA_y)_{y\in \partial a}$ should be stochastically dependent primarily through their joint membership in test $a$.
An excellent exposition of BP can be found in \cite{MM}.

How do we utilise the BP messages to infer the actual infection status of each individual?
The idea is to perform the update \eqref{eqBPupdate1}--\eqref{eqBPupdate2} for a `sufficiently large' number of rounds, say until an approximate fixed point is reached.
The (heuristic) BP estimate of the posterior marginals after $t$ rounds then reads
\begin{align}\label{eqBPmarg}
	\mu_{x,t}(r)&\propto\bcfr{k}{n}^r\bc{1-\frac kn}^{1-r}\prod_{b\in \partial x}\mu_{b\to x,t-1}(r)&&(r\in\{0,1\}).
\end{align}
Thus, by comparison to \eqref{eqBPupdate2} we just take the incoming messages from {\em all} tests $b\in\partial x$ into account.
In summary, we `hope' that after sufficiently many updates we have $\mu_{x,t}(r)\approx\pr[\SIGMA_x=r\mid\dSIGMA]$.
We could then, for instance, declare the $k$ individuals with the greatest $\mu_{x,t}(1)$ infected, and everybody else uninfected.


For later reference we point out that the BP updates \eqref{eqBPupdate1}--\eqref{eqBPupdate2} and \eqref{eqBPmarg} can be simplified slightly by passing to log-likelihood ratios.
Thus, define
\begin{align}\label{eqBPLog}
	\eta_{x\to a,t}&=\log\frac{\mu_{G,\dSIGMA,x\to a,t}(1)}{\mu_{G,\dSIGMA,x\to a,t}(0)},&
	\eta_{G,\dSIGMA,a\to x,t}&=\log\frac{\mu_{G,\dSIGMA,a\to x,t}(1)}{\mu_{G,\dSIGMA,a\to x,t}(0)},
\end{align}
with the initialisation $\eta_{G,\dSIGMA,x\to a,0}=\log(k/(n-k))\sim(\theta-1)\log n$ from \eqref{eqBPinit}.
Then \eqref{eqBPupdate1}--\eqref{eqBPmarg} transform into
\begin{align}\label{eqBPLogUpdate1}
	\eta_{a\to x,t}&=\begin{cases}
		\log p_{11}-\log\brk{p_{11}+(p_{01}-p_{11})\prod_{y\in\partial a\setminus\{x\}}\frac12\bc{1-\tanh\bc{\frac12\eta_{y\to a,t-1}}}}&\mbox{if }a\in F^+,\\
		\log p_{10}-\log\brk{p_{10}+(p_{00}-p_{10})\prod_{y\in\partial a\setminus\{x\}}\frac12\bc{1-\tanh\bc{\frac12\eta_{y\to a,t-1}}}}&\mbox{if }a\in F^-,\\
	\end{cases}\\
	\eta_{x\to a,t}&=(\theta-1)\log n+\sum_{b\in\partial x\setminus\{a\}}\eta_{b\to x,t},\label{eqBPLogUpdate2},\quad \eta_{x,t}=(\theta-1)\log n+\sum_{b\in\partial x}\eta_{b\to x,t}.
\end{align}
In this formulation BP ultimately diagnoses the $k$ individuals with the largest $\eta_{x,t}$ as infected.

Under the assumptions of \Thm~\ref{thm_DD} the {\tt DD} algorithm can be viewed as the special case of BP with $t=2$ applied to $\Gcc$.
Indeed, the analysis of {\tt DD} evinces that on $\Gcc$ the largest $k$ BP estimates \eqref{eqBPmarg} with $t\geq2$ correctly identify the infected individuals \whp%
\footnote{In the noiseless case {\tt DD} is actually a special case of a discrete message passing algorithm called Warning Propagation~\cite{MM}.}

It is therefore an obvious question whether BP on the constant column fits the bill of \Thm~\ref{thm_alg_apx}.
Clearly, BP remedies the obvious deficiencies of {\tt DD} by taking into account information from a larger radius around an individual (if we iterate beyond $t=2$).
Also in contrast to {\tt DD}'s hard thresholding the update rules \eqref{eqBPupdate1}--\eqref{eqBPupdate2} take information into account in a more subtle, soft manner.
Nonetheless, we do not expect that {\tt BP} applied to the constant column design meets the information-theoretically optimal bound from \Thm~\ref{thm_alg_apx}.
In fact, there is strong evidence that {\tt BP} does not suffices to meet the information-threshold for all $\theta$ even in the noiseless case~\cite{ilias}.
The fundamental obstacle appears to be the `cold' initialisation \eqref{eqBPinit}, which (depending on the parameters) can cause the BP messages to approach a meaningless fixed point.
Yet for symmetry reasons on the constant column design no better starting point than the prior \eqref{eqBPinit} springs to mind; after all, $\Gcc$ is a nearly biregular random graph, and thus all individuals look alike.
To overcome this issue we will employ a different type of test design that enables a warm start for BP.
This technique goes by the name of spatial coupling.

\subsection{Spatial coupling}\label{sec_gsc}

The thrust of spatial coupling is to blend a randomised construction, in our case the constant column design, with a spatial arrangement so as to provide a propitious starting point for BP.
Originally hailing from coding theory, spatial coupling has also been used in previous work on noiseless group testing~\cite{opt}.
In fact, the construction that we use is essentially identical to that from~\cite{opt} (with suitably adapted parameters).

To set up the spatially coupled test design $\Gsp$ we divide the set $V=V_n=\{x_1,\ldots,x_n\}$ of individuals into
\begin{align}\label{eqell}
	\ell=\lceil\sqrt{\log n}\rceil
\end{align}
pairwise disjoint \emph{compartments} $V[1],\ldots,V[\ell]\subset V$ such that $\lfloor n/\ell\rfloor\leq|V[i]|\leq\lceil n/\ell\rceil$.
We think of these compartments as being spatially arranged so that $V[i+1]$ comes `to the right' of $V[i]$ for $1\leq i<\ell$.
More precisely, we arrange the compartments in a ring topology such that $V[\ell]$ is followed again by $V[1]$.
Hence, for notational convenience let $V[\ell+j]=V[j]$ and $V[1-j]=V[\ell-j+1]$ for $1\leq j\leq\ell$.
Additionally, we introduce $\ell$ compartments $F[1],\ldots,F[\ell]$ of tests arranged in the same way: think of $F[i]$ as sitting `above' $V[i]$.
We assume that the total number $m$ of tests in $F[1]\cup\ldots\cup F[\ell]$ is divisible by $\ell$ and satisfies $m=\Theta(k\log(n/k))$.
Hence, let each compartment $F[i]$ contain precisely $m/\ell$ tests.
As in the case of the individuals we let $F[\ell+j]=F[j]$ for $0<j\leq\ell$.
Additionally, let 
\begin{align}\label{eqs}
	s&=\lceil\log\log n\rceil&&\mbox{and}&\Delta&=\Theta(\log n)
\end{align}
be integers such that $\Delta$ is divisible by $s$.
Construct $\Gsp$ by letting each $x\in V[i]$ join precisely $\Delta/s$ tests from $F[i+j-1]$ for $j=1,\ldots,s$. 
These tests are chosen uniformly without replacement and independently for different $x$ and $j$.
Additionally, $\Gsp$ contains a compartment $F[0]$ of
\begin{align}\label{eqF0}
	m_0&=2c_{\DD}\frac{ks}\ell\log(n/k)=o(k\log(n/k))
\end{align}
tests.
Every individual $x$ from the first $s$ compartments $V[1],\ldots,V[s]$ joins an equal number $\Delta_0$ of tests from $F[0]$.
These tests are drawn uniformly without replacement and independently.
For future reference we let
\begin{align}\label{eqcndn}
	c=\cn&=\frac{m}{k\log(n/k)},&d=\dn&=\frac{k\Delta}{m};
\end{align}
the aforementioned assumptions on $m,\Delta$ ensure that $c,d=\Theta(1)$ and the total number of tests of $\Gsp$ comes to
\begin{align}\label{eqtotaltests}
	\sum_{i=0}^{\ell}|F[i]|&=(c+o(1))k\log(n/k).
\end{align}

In summary, $\Gsp$ consists of $\ell$ equally sized compartments $V[i]$, $F[i]$ of tests plus one extra serving $F[0]$ of tests.
Each individual $V[i]$ joins random tests in the $s$ consecutive compartments $F[i+j-1]$ with $1\leq j\leq s$.
Additionally, the individuals in the first $s$ compartments $V[1]\cup\cdots\cup V[s]$, which we refer to as the {\em seed}, also join the tests in $F[0]$.

We will discover momentarily how $\Gsp$ facilitates inference via BP.
But first let us make a note of some basic properties of $\Gsp$.
Recall that $\SIGMA\in\{0,1\}^V$, which encodes the true infection status of each individual, is chosen uniformly and independently of $\Gsp$ from all vectors of Hamming weight $k$.
Let $V_1=\{x\in V:\SIGMA_x=1\}$, $V_0=V\setminus V_1$ and let $V_r[i]=V_r\cap V[i]$ be the set of individuals with infection status $r\in\{0,1\}$ in compartment $i$.
Furthermore, recall that $\aSIGMA_a\in\{0,1\}$ denotes the actual result of test $a\in F=F[0]\cup\cdots\cup F[\ell]$, and that $\dSIGMA_a$ signifies the displayed result of $a$ as per~\eqref{eqnoisemodel}.
For $r\in\{0,1\}$ and $0\leq i\leq\ell$ let
\begin{align*}
	F_r[i]&=F_r\cap F[i],&F_r^+[i]&=F_r[i]\cap F^+,&F_r^-[i]&=F_r[i]\cap F^-.
\end{align*}
Thus, the subscript indicates the actual test result, while the superscript indicates the displayed result. 
In \Sec~\ref{sec_prop_basic} we will prove the following. 

\begin{proposition}\label{prop_basic}
	The test design $\Gsp$ enjoys the following properties with probability $1-o(n^{-2})$.
	\begin{description}
		\item[G1]  The number of infected individuals in the various compartments satisfy
			\begin{align}\label{eqG1}
				\frac k\ell-\sqrt{\frac k\ell}\log n\leq\min_{i\in[\ell]}|V_1[i]|\leq\max_{i\in[\ell]}|V_1[i]|\leq \frac k\ell+\sqrt{\frac k\ell}\log n.
			\end{align}
		\item[G2] For all $i\in[\ell]$ the numbers of tests that are actually/displayed positive/negative satisfy
			\begin{align}
				\frac{m}{\ell}\exp(-d)p_{00}-\sqrt m\ln^3 n&\leq\abs{F_0^-[i]}\leq\frac m{\ell}\exp(-d)p_{00}+\sqrt m\ln^3 n,\label{eqG3_1}\\
				\frac{m}{\ell}\exp(-d)p_{01}-\sqrt m\ln^3 n&\leq\abs{F_0^+[i]}\leq\frac m{\ell}\exp(-d)p_{01}+\sqrt m\ln^3 n,\label{eqG3_2}\\
				\frac{m}{\ell}(1-\exp(-d))p_{10}-\sqrt m\ln^3 n&\leq\abs{F_1^-[i]}\leq\frac m{\ell}(1-\exp(-d))p_{10}+\sqrt m\ln^3 n,\label{eqG3_3}\\
				\frac{m}{\ell}(1-\exp(-d))p_{11}-\sqrt m\ln^3 n&\leq\abs{F_1^+[i]}\leq\frac m{\ell}(1-\exp(-d))p_{11}+\sqrt m\ln^3 n.\label{eqG3_4}
			\end{align}
	\end{description}
\end{proposition}

\subsection{Approximate recovery}
We are going to exploit the spatial structure of $\Gsp$ in a manner reminiscent of domino toppling.
To get started we will run {\tt DD} on the seed $V[1]\cup\cdots\cup V[s]$ and the tests $F[0]$ only; this is our first domino.
The choice \eqref{eqF0} of $m_0$ ensures that {\tt DD} diagnoses all individuals in $V[1]\cup\cdots\cup V[s]$ correctly \whp.
The seed could then be used as an informed starting point from which we could iterate BP to infer the status of the individuals in $V[s+1]\cup\ldots\cup V[\ell]$.
However, this algorithm appears to be difficult to analyse.
Instead we will show under that the assumptions of \Thm~\ref{thm_alg_apx} a modified, `paced' version of BP that diagnoses one compartment (or `domino') at a time and then re-initialises the messages ultimately classifies all but $o(k)$ individuals correctly.
Let us flesh this strategy out in detail.

\subsubsection{The seed}
Recall that each individual $x\in V[1]\cup\cdots\cup V[s]$ independently joins $\Delta_0$ random tests from $F[0]$.
In its the initial step $\SPARC$ runs {\tt DD} on the test design $\G_0$ comprising $V[1]\cup\cdots\cup V[s]$ and the tests $F[0]$ only.
Throughout $\SPARC$ maintains a vector $\tau\in\{0,1,*\}^{V}$ that represents the algorithm's current estimate of the ground truth $\SIGMA$, with $*$ indicating `undetermined as yet'.

\IncMargin{1em}
\begin{algorithm}[h!]
 \KwData{$\G$, $\dSIGMA$}
 \KwResult{an estimate of $\SIGMA$}
 Let $(\tau_x)_{x\in V[1]\cup\cdots\cup V[s]}\in\cbc{0,1}^{V[1]\cup\cdots\cup V[s]}$ be the result of applying {\tt DD} to $V[1]\cup\cdots\cup V[s]$ and $F[0]$\;
   Set $\tau_{x}=*$ for all individuals $x\in V\setminus(V[1]\cup\cdots\cup V[s])$\;
 \caption{\SPARC, steps 1--2}\label{SC_algorithm1}
 \label{Alg_SC}
\end{algorithm}
\DecMargin{1em}

Since \Prop~\ref{prop_basic} shows that the seed contains $(1+o(1))ks/\ell$ infected individuals \whp, the choice \eqref{eqF0} of $m_0$ and  \Thm~\ref{thm_DD} imply that {\tt DD} will succeed to diagnose the seed correctly for a suitable $\Delta_0$.

\begin{proposition}[{\cite[\Thm~2.2]{Maurice}}] \label{prop_seed}
	There exist $\Delta_0=\Theta(\log n)$ such that the output of~\DD\ satisfies $\tau_x=\SIGMA_x$ for all $x\in V[1]\cup\cdots\cup V[s]$ \whp.
\end{proposition}

\subsubsection{A combinatorial condition}

To simplify the analysis of the message passing step in \Sec~\ref{sec_weights} we observe that certain individuals can be identified as likely uninfected purely on combinatorial grounds.
More precisely, consider $x\in V[i]$ for $s<i\leq\ell$.
If $x$ is infected, then any test $a\in\partial x$ is actually positive.
Hence, we expect that $x$ appears in about $p_{11}\Delta$ tests that display a positive result.
In fact, the choice \eqref{eqs} of $s$ ensures that \whp even within each separate compartment $F[i+j-1]$, $1\leq j\leq s$ the individual $x$ appears in about $p_{11}\Delta/s$ positively displayed tests.
Thus, let
\begin{align}\label{eqVi+}
	V^+[i]&=\cbc{x\in V[i]:\sum_{j=1}^s\abs{|\partial x\cap F^+[i+j-1]|-\Delta p_{11}/s}\leq\log^{4/7}n}.
\end{align}
The following proposition confirms that all but $o(k/s)$ infected individuals $x\in V[i]$ belong to $V^+[i]$.
Additionally, the proposition determines the approximate size of $V^+[i]$.
For notational convenience we define
\begin{align*}
	V_0^+[i]&=V^+[i]\cap V_0,&V_1^+[i]&=V^+[i]\cap V_1,&V^+&=\bigcup_{s<i\leq \ell}V^+[i].
\end{align*}
Recall $q_0^+$ from \eqref{eqq}.
The proof of the following proposition  can be found in \Sec~\ref{sec_plausible}.

\begin{proposition}\label{prop_plausible}
	\Whp we have
	\begin{align}\label{eq_prop_plausible1}
		\sum_{i=s+1}^\ell |V_1[i]\setminus V^+[i]|&\leq k\exp(-\Omega(\log^{1/7}n))&&\mbox{and}\\
	\label{eq_prop_plausible2}
		\abs{V^+[i]\setminus V_1[i]}&\leq\frac n\ell\exp\bc{-\Delta\KL{p_{11}}{q_0^+}+O(\log^{4/7}n)}
					&&\mbox{for all $s+1\leq i\leq \ell$}.
	\end{align}
\end{proposition}

\subsubsection{Belief Propagation redux}\label{sec_weights}
The main phase of the \SPARC\ algorithm employs a simplified version of the BP update rules \eqref{eqBPupdate1}--\eqref{eqBPupdate2} to diagnose one compartment $V[i]$, $s<i\leq\ell$, after the other.
The textbook way to employ BP would be to diagnose the seed, initialise the BP messages emanating from the seed accordingly and then run BP updates until the messages converge.
However, towards the proof of \Thm~\ref{thm_alg_apx} this way of applying BP seems too complicated both to run and to analyse.
Instead, \SPARC\ relies on a `paced' version of BP.
Rather than updating the messages to convergence from the seed, we perform one round of message updates, then diagnose the next compartment, re-initialise the messages to coincide with the newly diagnosed compartment and proceed to the next undiagnosed compartment.

We work with the log-likelihood versions of the BP messages from \eqref{eqBPLog}--\eqref{eqBPLogUpdate2}.
Suppose we aim to process compartment $V[i]$, $s<i\leq\ell$, having completed $V[1],\ldots,V[i-1]$ already.
Then for a test $a\in F[i+j-1]$, $j\in[s]$, and a adjacent variable $x\in V[i+j-s]\cup V[i+j-1]$ we initialise
\begin{align}\label{eqetainit}
	\eta_{x\to a,0}&=\begin{cases}
		-\infty&\mbox{ if }\tau_x=0,\\
		+\infty&\mbox{ if }\tau_x=1,\\
		\log(k/(n-k))&\mbox{ if }\tau_x=*.
		\end{cases}
\end{align}
The third case above occurs if and only if $x\in V[i]\cup\cdots\cup V[i+j-1]$; that is, if $x$ belongs to an as yet undiagnosed compartment.
For the compartments that have been diagnosed already we set $\eta_{x\to a,0}$ to $\pm\infty$, depending on whether $x$ has been classified as infected or uninfected.

Let us now investigate the ensuing messages $\eta_{a\to x,1}$ for $x\in V[i]$ and tests $a\in\partial x\cap F[i+j-1]$.
A glimpse at \eqref{eqBPLogUpdate1} reveals that for any test $a$ that contains an individual $y\in V[i+j-s]\cup\cdots\cup V[i-1]$ with $\tau_y=1$ we have $\eta_{a\to x,1}=0$. 
This is because \eqref{eqetainit} ensures that $\eta_{y\to a,0}=\infty$ and $\tanh(\infty)=1$.
Hence, the test $a$ contains no further information as to the status of $x$.
Therefore, we call a test $a\in F[i+j-1]$ \emph{informative} towards $V[i]$ if $\tau_y=0$ for all $y\in\partial a\cap(V[i+j-s]\cup\cdots\cup V[i-1])$.

Let $\vW_{i,j}(\tau)$ be the set of all informative $a\in F[i+j-1]$.
Then any $a\in \vW_{i,j}(\tau)$ receives $\eta_{y\to a,0}=-\infty$ from all individuals $y$ that have been diagnosed already, i.e.\ all $y\in V[h]$ with $i+j-s\leq h<i$.
Another glance at the update rule shows that the corresponding terms simply disappear from the product on the r.h.s.\ of \eqref{eqBPLogUpdate1}, because $\tanh(-\infty)=-1$.
Consequently, only the factors corresponding to undiagnosed individuals $y\in V[i]\cup\cdots\cup V[i+j-1]$ remain.
Hence, with $r=\vecone\{a\in F^+\}$ the update rule \eqref{eqBPLogUpdate1} simplifies to
\begin{align}\label{eqetasimplified}
	\eta_{a\to x,1}&=
	\log p_{1r}-\log\brk{p_{1r}+(p_{0r}-p_{1r})\bc{1-k/n}^{\abs{\partial a\cap (V[i]\cup\cdots\cup V[i+j-1])}-1}}.
\end{align}

The only random element in the expression \eqref{eqetasimplified} is the number $\abs{\partial a\cap (V[i]\cup\cdots\cup V[i+j-1])}$ of members of test $a$ from compartments $V[i]\cup\cdots\cup V[i+j-1]$.
But by the construction of $\Gsc$ this number has a binomial distribution with mean
\begin{align*}
	\ex\abs{\partial a\cap (V[i]\cup\cdots\cup V[i+j-1])}&=\frac{j\Delta n}{ms}+o(1)=\frac{djn}{ks}+o(1)&&\mbox{[using \eqref{eqcndn}]}.
\end{align*}
Since the fluctuations of $\abs{\partial a\cap(V[i]\cup\cdots\cup V[i+j-1])}$ are of smaller order than the mean, we conclude that \whp \eqref{eqetasimplified} can be well approximated by a deterministic quantity:
\begin{align}\label{eqetaapx}
	\eta_{a\to x,1}&=\begin{cases}
	w_j^++o(1),&\mbox{ if }a\in F^+,\\
	-w_j^-+o(1),&\mbox{ if }a\in F^-
	\end{cases},\qquad\mbox{where}\\
\label{eqweights}
	w_j^+&=\log\frac{p_{11}}{p_{11}+(p_{01}-p_{11})\exp(-dj/s)}\geq 0,&
	w_j^-&=-\log\frac{p_{10}}{p_{10}+(p_{00}-p_{10})\exp(-dj/s)}\geq 0.
\end{align}
Note that in the case $p_{10} = 0$, the negative test weight $W_j^-$ evaluates to $w_j^- = \infty$, indicating that individual contained in negative test definitely are uninfected.
Finally, the messages \eqref{eqetaapx} lead to the BP estimate of the posterior marginal of $x$ via \eqref{eqBPLogUpdate2}, i.e.\ by summing on all informative tests $a\in\partial x$.
To be precise, letting
\begin{align*}
	\vW_{x,j}^\pm(\tau)&=\partial x\cap \vW_{i,j}(\tau)\cap F^{\pm}
\end{align*}
be the positively/negatively displayed informative tests adjacent to $x$ and setting
\begin{align}\label{eqWtau}
	\vW_x^+(\tau)& = \sum_{j=1}^{s} w_j^+\abs{\vW_{x,j}^+(\tau)},&
	\vW_x^-(\tau)&=\sum_{j=1}^{s}w_j^- \abs{\vW_{x,j}^-(\tau)},
\end{align}
we obtain
\begin{align}
	\eta_{x,1}=&\vW_x^+(\tau)-\vW_x^-(\tau)+\mbox{`lower order fluctuations'}.\label{eqetaxapx}
\end{align}

One issue with the formula \eqref{eqetaapx} is the analysis of the `lower order fluctuations', which come from the random variables $|\partial a\cap (V^+[i]\cup\cdots\cup V^+[i+s-1])|$.
Of course, one could try to analyses theses deviations caefully by resorting to some kind of a normal approximation.
But for our purposes this is unnecessary.
It turns out that we may simply use the sum on the r.h.s.\ of \eqref{eqetaxapx} to identify which individuals are infected.
Specifically, instead of computing the actual BP approximation $\eta_{x,1}$ after one round of updating, we just compare $\vW_x^+(\tau)$ and $\vW_x^-(\tau)$ with the values that we would expect these random variables to take if $x$ were infected.
These conditional expectations work out to be
\begin{align}\label{eqWest}
	W^+&=p_{11}\Delta \sum_{j=1}^{s}\exp(d(j-s)/s)w_j^+,&
	W^-&= \begin{cases} p_{10}\Delta \sum_{j=1}^{s}\exp(d(j-s)/s)w_j^- &\mbox{ if } p_{10}>0\\ 0 &\mbox{ otherwise. } \end{cases}
\end{align}
Thus, \SPARC\ will diagnose $V[i]$ by comparing $\vW_x^{\pm}(\tau)$ with $W^{\pm}$.
Additionally, \SPARC\ takes into account that infected individuals likely belong to $V^+[i]$, as we learned from \Prop~\ref{prop_plausible}.

\IncMargin{1em}
\begin{algorithm}[ht]
  \setcounter{AlgoLine}{2}
  \For{$i=s+1, \dots, \ell $}{
 \For{$x \in V[i]$}{
	 \If{$x\not\in V^+[i]$ or $\vW_x^+(\tau) < (1-\zeta) W^+\mbox{ or }\vW_x^-(\tau) > (1+\zeta) W^-$}{
 $\tau_x = 0$ \tcp*[h]{classify as uninfected}}
 \Else{$\tau_x = 1$ \tcp*[h]{classify as infected}}
  }}
  \Return $\tau$
\caption{\SPARC, steps 3--9.}
 \end{algorithm}
\DecMargin{1em}

Let
\begin{align}\label{eqzeta}
	\zeta=(\log\log\log n)^{-1}
\end{align}
be a term that tends to zero slowly enough to absorb error terms.
The following proposition, which we prove in \Sec~\ref{sec_prop_dist_psi}, summarises the analysis of phase~3.
Recall from \eqref{eqcndn} that $c=m/(k\log(n/k))$.

\begin{proposition} \label{prop_dist_psi}
	Assume that for a fixed $\eps>0$ we have
	\begin{align}\label{eqprop_dist_psi}
		c>\ccentre(d)+\eps.
	\end{align}
	Then \whp the output $\tau$ of \SPARC\ satisfies
	$$\sum_{x\in V^+}\vecone\cbc{\tau_x\neq\SIGMA_x}\leq k\exp\bc{-\Omega\bcfr{\log n}{(\log\log n)^5}}.$$
\end{proposition}

The proof of \Prop~\ref{prop_dist_psi}, which can be found in \Sec~\ref{sec_prop_dist_psi}, is the centerpiece of the analysis of \SPARC.
The proof is based on a large deviations analysis that bounds the number of individuals $x\in V^+[i]$ whose corresponding sums $\vW_x^\pm(\tau)$ deviate form their conditional expectations given $\SIGMA_x$.
We have all the pieces in place to complete the proof of the first theorem.

\begin{proof}[Proof of \Thm~\ref{thm_alg_apx}]
	Setting $d=\dchan$ from \eqref{eqdcentre} and invoking~\eqref{eqcchandchan}, we see that the theorem is an immediate consequence of \Prop s~\ref{prop_seed}, \ref{prop_plausible} and \ref{prop_dist_psi}.
\end{proof}

\subsection{Exact recovery}\label{sec_disc_ex}

As we saw in \Sec~\ref{sec_intro_ex} the threshold $\cbound(d,\theta)$ encodes a local stability condition.
This condition is intended to ensure that \whp no other `solution' $\sigma\neq\SIGMA$ with $\|\sigma-\SIGMA\|_1=o(k)$ of a similarly large posterior likelihood exists.
In fact, because $\Gsp$ enjoys fairly good expansion properties the test results provide sufficient clues for us to home in on $\SIGMA$ once we get close, provided the number of tests is as large as prescribed by \Thm~\ref{thm_alg_ex}.
Thus, the idea behind exact recovery is to run \SPARC\ first and then apply local corrections to fully recover $\SIGMA$; a similar strategy was employed in the noiseless case in~\cite{opt}.

Though this may sound easy and a simple greedy strategy does indeed do the trick the noiseless case~\cite{opt}, in the presence of noise it takes a good bit of care to get the local search step right.
Hence, as per \eqref{eqcex} suppose that $c,d$ from \eqref{eqcndn} satisfy $c>\max\{\ccentre(d),\cbound(d,\theta)\}+\eps$.
Also suppose that we already ran \SPARC\ to obtain $\tau\in\{0,1\}^V$ with $\|\tau-\SIGMA\|_1=o(k)$ (as provided by \Prop~\ref{prop_dist_psi}).
How can we set about learning the status of an individual $x$ with perfect confidence?

Assume for the sake of argument that $\tau_y=\SIGMA_y$ for all $y$ that share a test $a\in\partial x$ with $x$.
If $a$ contains another infected individual $y\neq x$, then unfortunately nothing can be learned from $a$ about the status of $x$.
In this case we call the test $a$ {\em tainted}.
By contrast, if $\tau_y=0$ for all $y\in\partial a\setminus\{x\}$, i.e.\ if $a$ is {\em untainted}, then the displayed result $\dSIGMA_a$ hinges on the infection status of $x$ itself.
Hence, the larger the number of untainted positively displayed $a\in\partial x$, the likelier $x$ is infected.
Consequently, to accomplish exact recovery we are going to threshold the number of untainted positively displayed $a\in\partial x$.
But crucially, to obtain an optimal algorithm we cannot just use a scalar, one-size-fits-all threshold.
Instead, we need to carefully choose a threshold function that takes into account the total number of untainted tests $a\in\partial x$.

To elaborate, let
\begin{align}\label{eqYx}
	\vY_x&=\abs{\cbc{a\in\partial x\setminus F[0]:\forall y\in\partial a\setminus\{x\}:\SIGMA_y=0}}
\end{align}
be the total number of untainted tests $a\in\partial x$; to avoid case distinctions we omit seed tests $a\in F[0]$.
Routine calculations reveal that $\vY_x$ is well approximated by a binomial variable with mean $\exp(-d)\Delta$.
Therefore, the fluctuations of $\vY_x$ can be estimated via the Chernoff bound.
Specifically, the numbers of infected/uninfected individuals with $\vY_x=\alpha\Delta$ can be approximated as
\begin{align}\label{eqChernoffvY1}
	\ex\abs{\cbc{x\in V_1:\vY_x=\alpha\Delta}}&=k\exp\bc{-\Delta\KL{\alpha}{\exp(-d)}+o(\Delta)},\\
	\ex\abs{\cbc{x\in V_0:\vY_x=\alpha\Delta}}&=n\exp\bc{-\Delta\KL{\alpha}{\exp(-d)}+o(\Delta)}.\label{eqChernoffvY2}
\end{align}
Consequently, since $\kdef=o(n)$ `atypical' values of $\vY_x$ occur more frequently on healthy than on infected individuals.
In fact, recalling \eqref{eqYinterval}, we learn from a brief calculation that for $\alpha\not\in\cY(c,d,\theta)$ not a single $x\in V_1$ with $\vY_x=\alpha\Delta$ exists \whp.
Hence, if $\vY_x/\Delta\not\in\cY(c,d,\theta)$ we deduce that $x$ is uninfected.

For $x$ such that $\vY_x/\Delta\in\cY(c,d,\theta)$ more care is required.
In this case we need to compare the number
\begin{align}\label{eqZdef}
	\vZ_x&=\abs{\cbc{a\in\partial^+ x\setminus F[0]:\forall y\in\partial a\setminus\{x\}:\SIGMA_y=0}}
\end{align}
of {\em positively displayed} untainted tests to the total number $\vY_x$ of untainted tests.
Since the test results are put through the $\vec p$-channel independently, $\vZ_x$ is a binomial variable given $\vY_x$.
The conditional mean of $\vZ_x$ equals $p_{11}\vY_x$ if $x$ is infected, and $p_{01}\vY_x$ otherwise.
Therefore, the Chernoff bound shows that
\begin{align}\label{eqChernoffvZ}
	\pr\brk{\vZ_x=\alpha\beta\Delta\mid\vY_x=\alpha\Delta}&=\exp\bc{-\alpha\Delta\KL{\beta}{p_{\SIGMA_x1}}+o(\Delta)}.
\end{align}

In light of \eqref{eqChernoffvZ} we set up the definition \eqref{eqsep1} of $\cbound(d,\theta)$ so that $\fz(\nix)$ can be used as a threshold function to tell infected from uninfected individuals.
Indeed, 
given $\vY_x,\vZ_x$ we should declare $x$ uninfected if either $\vY_x/\Delta\not\in\cY(c,d,\theta)$ or $\vZ_x/\vY_x<\fz(\vZ_x/\Delta)$, and infected otherwise; then the choice of $\fz(\nix)$ and $\cex(\theta)$ would ensure that all individuals get diagnosed correctly \whp.
Figure~\ref{fig_lena} displays a characteristic specimen of the function $\fz(\nix)$ and the corresponding rate function from \eqref{eqsep1}.

\begin{figure}
	\includegraphics[height=55mm]{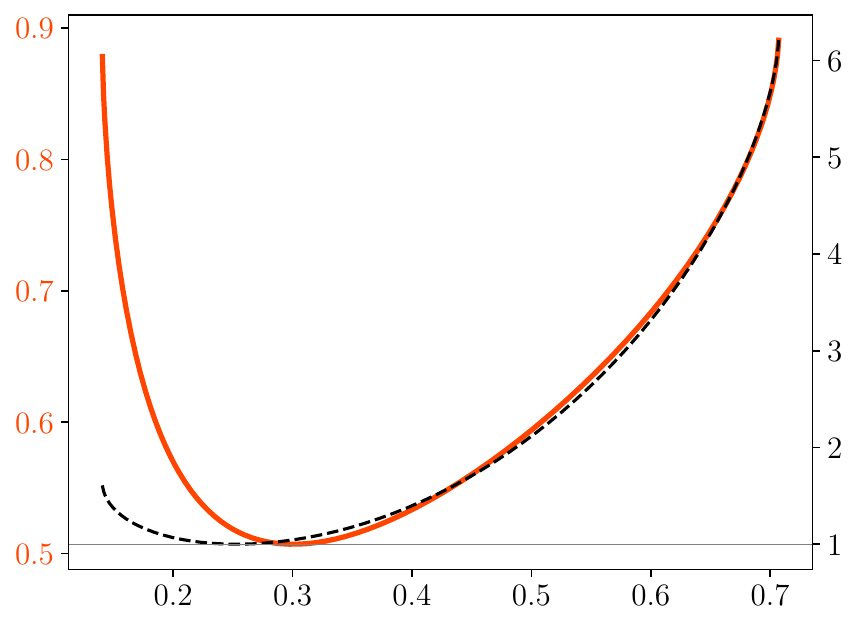}
	\caption{The threshold function $\fz(\nix)$ (red) on the interval $\cY(\cbound(d,\theta),d,\theta)$ and the resulting large deviations rate $\cbound(d,\theta)d(1-\theta)\bc{\KL \alpha{\exp(-d)}+\alpha\KL{\fz(\alpha)}{p_{01}}}$ (black) with $\theta=1/2$, $p_{00}=0.972$, $p_{11}=0.9$ at the optimal choice of $d$.}\label{fig_lena}
\end{figure} 

Yet trying to distil an algorithm from these considerations, we run into two obvious obstacles.
First, the threshold $\fz(\nix)$ may be hard to compute precisely.
Similarly, the limits of the interval $\cY(c,d,\theta)$ may be irrational (or worse).
The following proposition, which we prove in \Sec~\ref{sec_lem_yz}, remedies these issues.

\begin{proposition}\label{lem_yz}
	Let $\eps>0$ and assume that $c>\cex(d,\theta)+\eps$.
	Then there exist $\delta>0$ and an open interval $\emptyset\neq\cI=(l,r)\subset[\delta,1-\delta]$ with endpoints $l,r\in\QQ$ 
	such that for any $\eps'>0$ there exist $\delta'>0$ and a step function $\cZ:\cI\to(p_{01},p_{11})\cap\QQ$ such that the following conditions are satisfied.
	\begin{description}
		\item[Z1] $cd(1-\theta)\KL{y}{\exp(-d)}>\theta+\delta$ for all $y\in(0,1)\setminus(l+\delta,r-\delta)$.
		\item[Z2] $cd(1-\theta)\bc{\KL y{\exp(-d)}+y\KL{\cZ(y)}{p_{11}}}>\theta+\delta$ for all $y\in\cI$.
		\item[Z3] $cd(1-\theta)\bc{\KL y{\exp(-d)}+y\KL{\cZ(y)}{p_{01}}}>1+\delta$ for all $y\in\cI$.
		\item[Z4] If $y,y'\in\cI$ satisfy $|y-y'|<\delta'$, then $|\cZ(y)-\cZ(y')|<\eps'$.
	\end{description}
\end{proposition}

The second obstacle is that in the above discussion we assumed that $\tau_y=\SIGMA_y$ for all $y\in\partial a\setminus\{x\}$.
But all that the analysis of \SPARC\ provides is that \whp $\tau_y\neq\SIGMA_y$ for at most $o(k)$ individuals $y$.
To cope with this issue we resort to the expansion properties of $\Gsc$.
Roughly speaking, we show that \whp for any small set $S$ of individuals (such as $\{y:\tau_y\neq\SIGMA_y\}$) the set of individuals $x$ that occur in `many' tests that also contain a second individual from $S$ is significantly smaller than $S$.
As a consequence, as we apply the thresholding procedure repeatedly the number of misclassified individuals decays geometrically.
We thus arrive at the following algorithm.

\IncMargin{1em}
\begin{algorithm}[ht]
 \KwData{$\Gsp$, $\dSIGMA$}
 \KwResult{an estimate of $\SIGMA$}
  Let $\tau^{(1)}$ be the output of $\SPAX(\Gsp,\dSIGMA)$\;
  \For{ $i = 1,\dots, \lceil\log n\rceil$} {
  For all $x\in V[s+1]\cup\cdots\cup V[\ell]$ calculate\\
  $\displaystyle \quad Y_x(\tau^{(i)})=\sum_{a\in\partial x\setminus F[0]}\vecone\cbc{\forall y\in\partial a\setminus\cbc x:\tau^{(i)}_y=0}$,\quad 
   $\displaystyle Z_x(\tau^{(i)})=\sum_{a\in\partial x\setminus F[0]:\dSIGMA_a=1}\vecone\cbc{\forall y\in\partial a\setminus\cbc x:\tau^{(i)}_y=0}$ 
   \;
  Let 
  $\displaystyle\tau_x^{(i+1)}=
  	\begin{cases}
  	\tau_x^{(i)}&\mbox{ if }x\in V[1]\cup\cdots\cup V[s],\\
	\vecone\cbc{
		Y_x(\tau^{(i)})/\Delta\in\cI\mbox{ and }Z_x(\tau^{(i)})/\Delta>\cZ(Y_x(\tau^{(i)})/\Delta)}&\mbox{ otherwise }
  	\end{cases}$\;
  	}
  \KwRet{$\tau^{(\lceil\log n\rceil)}$}
  \caption{The $\SPEX$ algorithm.}\label{SC_algorithm}
\end{algorithm}
\DecMargin{1em}

\begin{proposition}\label{prop_endgame}
	Suppose that $c>\cex(d,\theta)+\eps$ for a fixed $\eps>0$. 
	There exists $\eps'>0$, a rational interval $\cI$ and a rational step function $\cZ$ such that \whp for all $1\leq i<\log n$ we have 
		$$\|\SIGMA-\tau^{(i+1)}\|_1\leq\frac13\|\SIGMA-\tau^{(i)}\|_1.$$
\end{proposition}

\noindent
We prove \Prop~\ref{prop_endgame} in \Sec~\ref{sec_prop_endgame}.

\begin{proof}[Proof of \Thm~\ref{thm_alg_ex}]
	Since \Prop~\ref{prop_endgame} shows that the number of misclassified individuals decreases geometrically as we iterate Steps 3--5 of \SPEX, we have $\tau^{(\lceil\log n\rceil)}=\SIGMA$ \whp.
	Furthermore, thanks to \Thm~\ref{thm_alg_apx} and \Prop~\ref{lem_yz} \SPEX\ is indeed a polynomial time algorithm.
\end{proof}

\begin{remark}\label{rem_noiseless}\upshape
	In the noiseless case $p_{00}=p_{11}=1$ \Thm~\ref{thm_alg_ex} reproduces the analysis of the \SPIV\ algorithm from~\cite{opt}.
	One key difference between \SPARC\ and \SPEX\ on the one hand and \SPIV\ on the other is that the former are based on Belief Propagation, while the latter relies on combinatorial intuition.
	More precisely, the \SPIV\ algorithm infers from positive tests by means of a weighted sum identical to $\vW_x^+(\tau)$ from \eqref{eqWtau} with the special values $p_{00}=p_{11}=1$ and $d=\log2$.
	In the noiseless case the weights $w_j^+$ were `guessed' based on combinatorial intuition.
	Furthermore, in noiseless case we can be certain that any individual contained in a negative test is healthy, and therefore the \SPIV\ algorithm only takes negative tests into account in this simple, deterministic manner.
	By contrast, in the noisy case the negative tests give rise to a second weighted sum $\vW^-_x(\tau)$.
	An important novelty is that rather than `guessing' the weights $w_j^\pm$, here we discovered how they can be derived systematically from the BP formalism.
	Apart from shedding new light on the noiseless case as well, we expect that this type of approach can be generalised to quite a few other inference problems as well.
	A second novelty in the design of \SPEX\ is the use of the threshold function $\cZ(\nix)$ that depends on the number untainted tests.
	The need for such a threshold function is encoded in the optimisation problem \eqref{eqsep1} that gives rise to a non-trivial choice of the coefficient $d$ that governs the size of the tests.
	This type of optimisation is unnecessary in the noiseless case, where simply $d=\log2$ is the optimal choice for all $\theta\in(0,1)$.
\end{remark}

\subsection{The constant-column lower bound}
Having completed the discussion of the algorithms \SPARC\ and \SPEX\ for \Thm s~\ref{thm_alg_apx} and~\ref{thm_alg_ex}, we turn to the proof of \Thm~\ref{thm_inf_ex} on the information-theoretic lower bound for the constant column design.
The goal is to show that for $m<(\cex(\theta)-\eps)k\log(n/k)$ no algorithm, efficient or otherwise, can exactly recover $\SIGMA$ on $\Gcc$.

To this end we are going to estimate the probability of the actual ground truth $\SIGMA$ under the posterior given the test results.
We recall from Fact~\ref{fact_Nishi} that the posterior coincides with the Boltzmann distribution from \eqref{eqBoltzmann}.
The following proposition, whose proof can be found in \Sec~\ref{sec_cc_ground}, estimates the Boltzmann weight $\psi_{\Gcc,\dSIGMA}(\SIGMA)$ of the ground truth.
Recall from \eqref{eqcndn} that $d=k\Delta/m$.
Also recall the weight function from \eqref{eqpsiG}.

\begin{proposition}\label{prop_cc_ground}
	\whp the weight of $\SIGMA$ satisfies
	\begin{align}\label{eqprop_cc_ground}
		\frac1m\log\psi_{\Gcc,\dSIGMA}(\SIGMA)&=-\exp(-d)h(p_{00})-(1-\exp(-d))h(p_{11})+O(m^{-1/2}\log^3).
	\end{align}
\end{proposition}

We are now going to argue that for $c<\cex(\theta)$ the partition function $Z_{\Gcc,\dSIGMA}$ dwarfs the Boltzmann weight \eqref{eqprop_cc_ground} \whp.
The definition \eqref{eqcex} of $\cex(\theta)$ suggests that there are two possible ways how this may come about.
The first occurs when $c=m/(k\log(n/k))$ is smaller than the local stability bound $\cbound(d,\theta)$ from \eqref{eqsep1}.
In this case we will essentially put the analysis of \SPEX\ into reverse gear.
That is, we will show that there are plenty of individuals $x$ whose status could be flipped from infected to healthy or vice versa without reducing the posterior likelihood.
In effect, the partition function will be far greater than the Boltzmann weight of $\SIGMA$.

\begin{proposition}\label{prop_XYZ}
	Assume that there exists $y\in\cY(c,d,\theta)$ such that there exists $z\in(p_{01},p_{11})$ such that
	\begin{align}\label{eq_prop_XYZ_1}
	cd(1-\theta)\bc{\KL y{\exp(-d)}+y\KL{z}{p_{11}}}&<\theta&&\mbox{and}\\
	cd(1-\theta)\bc{\KL y{\exp(-d)}+y\KL{z}{p_{01}}}&<1.\label{eq_prop_XYZ_2}
\end{align}
Then \whp there exist $n^{\Omega(1)}$ pairs $(v,w)\in V_1\times V_0$ such that for the configuration $\SIGMA^{[v,w]}$ obtained from $\SIGMA$ by inverting the $v,w$-entries we have 
	\begin{align}\label{eq_prop_XYZ_3}
		\mu_{\Gcc,\dSIGMA}(\SIGMA^{[v,w]})=\mu_{\Gcc,\dSIGMA}(\SIGMA).
	\end{align}
\end{proposition}

\noindent
We prove \Prop~\ref{prop_XYZ} in \Sec~\ref{sec_cc_stab}.

The second case is that $c$ is smaller than the entropy bound $\ccentre(d)$ from \eqref{eqccentre}.
In this case we will show by way of a moment computation that $Z_{\Gcc,\dSIGMA}$ exceeds the Boltzmann weight of $\SIGMA$ \whp.
More precisely, in \Sec~\ref{sec_cc_mmt} we prove the following.

\begin{proposition}\label{prop_mmt}
	Let $\eps>0$.
	If $c<\ccentre(d)-\eps$, then
	\begin{align*}
		\pr\brk{ \log Z_{\Gcc,\dSIGMA}\geq k\log(n/k)\brk{1-c/\ccentre(d)+o(1)}}>1-\eps+o(1).
	\end{align*}
\end{proposition}

\begin{proof}[Proof of \Thm~\ref{thm_inf_ex}]
	\Prop~\ref{prop_XYZ} readily implies that $\mu_{\Gcc,\dSIGMA}(\SIGMA)=o(1)$ \whp if $c<\cbound(d,\theta)$.
	Furthermore, in the case $c<\ccentre(d)$ the assertion follows from \Prop s~\ref{prop_cc_ground} and~\ref{prop_mmt}.
\end{proof}

\section{Analysis of the approximate recovery algorithm}\label{sec_alg_apx}

\noindent
In this section we carry out the proofs of the various proposition leading up to the proof of \Thm~\ref{thm_alg_apx} (\Prop s~\ref{prop_basic}, \ref{prop_plausible} and \ref{prop_dist_psi}).
Throughout we work with the spatially coupled test design $\Gsc$ from \Sec~\ref{sec_gsc}.
Hence, we keep the notation and assumptions from \eqref{eqell}--\eqref{eqtotaltests}.
We also make a note of the fact that for any $x\in V[i]$ and any $a\in F[i+j-1]$, $j=1,\ldots,s$,
\begin{align}\label{eqax}
	\pr\brk{x\in\partial a}&=1-\pr\brk{x\not\in\partial a}=1-\binom{|F[i+j-1]|-1}{\Delta/s}\binom{|F[i+j-1]|}{\Delta/s}^{-1}=\frac{\Delta\ell}{ms}+O\bc{\bcfr{\Delta\ell}{ms}^2};
\end{align}
this is because the construction of $\Gsp$ ensures that $x$ joins $\Delta/s$ random tests in each of $F[i+j-1]$, drawn uniformly without replacement.

\subsection{Proof of \Prop~\ref{prop_basic}}\label{sec_prop_basic}
Property {\bf G1} is a routine consequence of the Chernoff bound and was previously established as part of \cite[\Prop~4.1]{opt}.
With respect to {\bf G2} we may condition on the event $\fE$ that the bound \eqref{eqG1} from {\bf G1} is satisfied.
Consider a test $a\in F[i]$.
Recall that $a$ comprises individuals from the compartments $V[i-s+1],\ldots,V[i]$.
Since the probability that a specific individual joins a specific test is given by \eqref{eqax} and since individuals choose the tests that they join independently, on $\fE$ for each $i-s+1\leq h\leq i$ we have
\begin{align}\label{eqG31}
	|V_1[h]\cap\partial a|\disteq\Bin\bc{\frac k\ell+O\bc{\sqrt{\frac k\ell}\log n},\frac{\Delta\ell}{ms}+O\bc{\bcfr{\Delta\ell}{m s}^2}}.
\end{align}
Combining \eqref{eqcndn} and \eqref{eqG31}, we obtain 
\begin{align}\label{eqG32}
	\ex[|V_1[h]\cap\partial a|\mid\fE]&=\frac{\Delta k}{ms}+O\bc{\sqrt{\frac\ell k}\log n}=\frac ds+O\bc{k^{-1/2}\log^{3/2}n}.
\end{align}
Further, combining \eqref{eqG31} and \eqref{eqG32}, we get
\begin{align*}
	\pr\brk{V_1[h]\cap\partial a=\emptyset\mid\fE}&=\exp\brk{\bc{\frac k\ell+O\bc{\sqrt{\frac k\ell}\log n}}\log\bc{1-\frac{\Delta\ell}{ms}+O\bc{\bcfr{\Delta\ell}{m s}^2}}}=\exp\bc{-\frac ds+O\bc{k^{-1/2}\log^{3/2}n}}.
\end{align*}
Multiplying these probabilities up on $i-s+1\leq h\leq i$, we arrive at the estimate
\begin{align*}
	\pr\brk{V_1[h]\cap\partial a=\emptyset\mid\fE}&=\exp\bc{-d+O\bc{k^{-1/2}\log^{8/5}n}}.
\end{align*}
Hence,
\begin{align}\label{eqG33}
	\ex[\abs{F_0[i]}\mid\fE]=\frac m\ell\exp(-d)+O(\sqrt m \log^{2}n).
\end{align}
To establish concentration, observe that the set $\partial x$ of tests that a specific infected individual $x\in V_1[h]$ joins can change $|F_0[i]|$ by $\Delta$ at the most.
Moreover, changing the neighbourhood $\partial x$ of a healthy individual cannot change the actual test results at all.
Therefore, by Azuma--Hoeffding,
\begin{align}\label{eqneg2}
	\pr\brk{\abs{|F_0[i]|-\Erw\brk{|F_0[i]| \mid\cE}}\geq \sqrt m\log^2n \mid\cE}&\leq2\exp\bc{-\frac{m\log^4n}{2k\Delta^2}}=o(n^{-3}).
\end{align}
Thus, \eqref{eqG33}, \eqref{eqneg2} and {\bf G1} show that with probability $1-o(n^{-2})$ for all $1\leq i\leq\ell$ we have
\begin{align}\label{eqG34}
	\abs{F_0[i]}&=\frac m\ell\exp(-d)+O(\sqrt m \log^{2}n),&\mbox{and thus }&&
	\abs{F_1[i]}&=\frac m\ell(1-\exp(-d))+O(\sqrt m \log^{2}n).
\end{align}
Since the actual test results are subjected to the $\vec p$-channel independently to obtain the displayed test results, the distributions of $|F_{0/1}^\pm[i]$ given $F_{0/1}[i]$ read
\begin{align}\label{eqG35}
	&\abs{F_0^-[i]}=\Bin(|F_0[i]|,p_{00}),
	&\abs{F_0^+[i]}=\Bin(|F_0[i]|,p_{01}),\\
	&\abs{F_1^-[i]}=\Bin(|F_1[i]|,p_{10}),
	& \abs{F_1^+[i]}=\Bin(|F_1[i]|,p_{11}).
\end{align}
Thus, {\bf G2} follows from \eqref{eqG34}, \eqref{eqG35} and \Lem~\ref{lem_chernoff} (the Chernoff bound).

\subsection{Proof of \Prop~\ref{prop_plausible}}\label{sec_plausible}
	Any $x\in V_1[i]$ has a total of $\Delta/s$ neighbours in each of $F[i],\ldots,F[i+s-1]$.
	Moreover, all tests $a\in\partial x$ are actually positive.
	Since the displayed result is obtained via the $\vec p$-channel independently for every $a$, the number of displayed positive neighbours $|\partial x\cap F^+[i+j-1]|$ is a binomial variable with distribution $\Bin(\Delta/s,p_{11})$.
	Since $\Delta=\Theta(\log n)$ and $s=\Theta(\log\log n)$, the first assertion \eqref{eq_prop_plausible1} is immediate from \Lem~\ref{lem_chernoff}.

	Moving on to the second assertion, we condition on the event $\fE$ that the bounds \eqref{eqG3_1}--\eqref{eqG3_4} hold for all $i\in[\ell]$.
	Then \Prop~\ref{prop_basic} shows that $\pr\brk\fE=1-o(n^{-2})$.
	Given $\fE$ we know that
	\begin{align}\label{eq_prop_plausible_1}
		|F^+[i]|&=\frac{q_0^+m}\ell+O(\sqrt m\log^3n),&
		|F^-[i]|&=\frac{(1-q_0^+)m}\ell+O(\sqrt m\log^3n).
	\end{align}
	Now consider an individual $x\in V_0[i]$.
	Also consider any test $a\in F[i+j-1]$ for $j\in[s]$.
	Then the actual result $\aSIGMA_a$ of $a$ is independent of the event $\{x\in a\}$, because $x$ is uninfected.
	Since the displayed result $\dSIGMA_a$ depends solely on $\aSIGMA_a$, we conclude that $\dSIGMA_a$ is independent of $\{x\in a\}$ as well.
	Therefore, \eqref{eq_prop_plausible_1} shows that on the event $\fE$ the number of displayed positive tests that $x$ is a member of has conditional distribution 
	\begin{align}\label{eq_prop_plausible_2}
		|\partial x\cap F^+[i+j-1]|\disteq\Hyp\bc{\frac m\ell,\frac{q_0^+m}\ell+O(\sqrt m\log^3n),\frac\Delta s}.
	\end{align}
	Since the random variables $(|\partial x\cap F^+[i+j-1]|)_{1\leq j\leq s}$ are mutually independent, \eqref{eq_prop_plausible2} follows from \Lem~\ref{lem_hyperchernoff}.

\subsection{Proof of \Prop~\ref{prop_dist_psi}}\label{sec_prop_dist_psi}


We reduce the proof of \Prop~\ref{prop_dist_psi} to three lemmas.
The first two estimate the sums from \eqref{eqWtau} when evaluated at the actual ground truth $\SIGMA$.

\begin{lemma}\label{lem_1dev}
	Assume that \eqref{eqprop_dist_psi} is satisfied.
	\whp we have 
	\begin{align}\label{eqlem_1dev}
		\sum_{s<i\leq\ell}\sum_{x\in V_1^+[i]}\vecone\cbc{\vW_x^+(\SIGMA) < (1-\zeta/2) W^+\mbox{ or }\vW_x^-(\SIGMA) > (1+\zeta/2) W^-}\leq k\exp\bc{-\Omega\bcfr{\log n}{(\log\log n)^4}}.
	\end{align}
\end{lemma}

\begin{lemma}\label{lem_0dev}
	Assume that \eqref{eqprop_dist_psi} is satisfied.
	\whp we have 
	\begin{align}\label{eqlem_0+dev}
		\sum_{s\leq i<\ell}\sum_{x\in\zeroplus[i]}\vecone\cbc{\vW_x^+(\SIGMA) \geq (1-2\zeta) W^+\mbox{ and }\vW_x^-(\SIGMA) \leq (1+2\zeta) W^-}&\leq k^{1-\Omega(1)}.
	\end{align}
\end{lemma}

\noindent
We defer the proofs of \Lem s~\ref{lem_1dev} and~\ref{lem_0dev} to \Sec s~\ref{Sec_lem_1dev} and \ref{sec_lem_0dev}.
While the proof of \Lem~\ref{lem_1dev} is fairly routine, the proof of \Lem~\ref{lem_0dev} is the linchpin of the entire analysis of \SPARC, as it effectively vindicates the BP heuristics that we have been invoking so liberally in designing the algorithm.

Additionally, to compare $\vW_x^\pm(\SIGMA)$ with the algorithm's estimate $\vW_x^\pm(\tau)$ we resort to the following expansion property of $\Gsp$.

\begin{lemma} \label{lemma_endgame_misclassified}
	Let $0<\alpha,\beta<1$ be such that $\alpha+\beta>1$.
	Then \whp for any $T \subset V$ of size $|T|\leq\exp(-\log^\alpha n)k$,
	\begin{align*}
		\abs{\cbc{x\in V:\sum_{a\in\partial x\setminus F[0]}\vecone\cbc{T\cap\partial a\setminus \cbc x\neq\emptyset}\geq\ln^{\beta}n}}\leq\frac{|T|}{8\log\log n}.
	\end{align*}
\end{lemma}

\noindent
In a nutshell, \Lem~\ref{lemma_endgame_misclassified} posits that for any `small' set $T$ of individuals there are even fewer individuals that share many tests with individuals from $T$.
\Lem~\ref{lemma_endgame_misclassified} is an generalisation of~\cite[\Lem~4.16]{opt}.
The proof, based on a routine union bound argument, is included in Appendix~\ref{sec_lemma_endgame_misclassified} for completeness.

\begin{proof}[Proof of \Prop~\ref{prop_dist_psi}]
	\Prop~\ref{prop_seed} shows that for all individuals $x$ in the seed $V[1],\ldots,V[s]$ we have $\tau_x=\SIGMA_x$ \whp.
	Let $\cM[i] = \cbc{x \in V^+[i]: \tau_x \neq \SIGMA_x}$ be the set of misclassified individuals in $V^+[i]$.
	We are going to show that \whp for all $1\leq i\leq \ell$ and for large enough $n$,
	\begin{align}\label{eq_lemma_endgame_misclassified_1}
		\abs{\cM[i]}\leq k\exp\bc{-\frac{\ln n}{(\ln\ln n)^{5}}}.
	\end{align}
	
	We proceed by induction on $i$.
	As mentioned above, \Prop~\ref{prop_seed} ensures that $\cM[1]\cup\cdots\cup\cM[s]=\emptyset$ \whp.
	Now assume that \eqref{eq_lemma_endgame_misclassified_1} is correct for all $i<h\leq\ell$; we are going to show that \eqref{eq_lemma_endgame_misclassified_1} holds for $i=h$ as well.
	To this end, recalling $\zeta$ from \eqref{eqzeta} and $W^\pm$ from \eqref{eqWest}, we define for $p_{10} > 0$
	\begin{align*}
		\cM_1[h]&=\cbc{x\in V_1^+[h]:\vW_x^+(\SIGMA) < (1- \zeta/2) W^+\mbox{ or }\vW_x^-(\SIGMA) > (1+\zeta/2) W^-},\\
		\cM_2[h]&=\cbc{x\in V_0^+[h]:\vW_x^+(\SIGMA) > (1- 2\zeta) W^+\mbox{ and }\vW_x^-(\SIGMA) < (1+2\zeta) W^-},\\
		\cM_3[h]&=\cbc{x\in V^+[h]:|\vW_x^+(\SIGMA)-\vW_x^+(\tau)|+|\vW_x^-(\SIGMA)-\vW_x^-(\tau)|>\zeta(W^+\wedge W^-)/8}
	\end{align*}
	and further for $p_{10} = 0$ 
		\begin{align*}	
			\cM_1[h]&=\cbc{x\in V_1^+[h]:\vW_x^+(\SIGMA) < (1- \zeta/2) W^+},\\
			\cM_2[h]&=\cbc{x\in V_0^+[h]:\vW_x^+(\SIGMA) > (1- 2\zeta) W^+\mbox{ and }\vW_x^-(\SIGMA) = 0 },\\
			\cM_3[h]&=\cbc{x\in V^+[h]:|\vW_x^+(\SIGMA)-\vW_x^+(\tau)|>\zeta(W^+)/8}.
		\end{align*}
	
	We claim that $\cM[h]\subseteq\cM_1[h]\cup\cM_2[h]\cup\cM_3[h]$.
	To see this, assume that $x\in\cM[h]\setminus(\cM_1[h]\cup\cM_2[h])$.
	Then for \SPARC\ to misclassify $x$, it must be the case that 
	\begin{align*}
		|\vW_x^+(\tau)-\vW_x^+(\SIGMA)|+|\vW_x^-(\tau)-\vW_x^-(\SIGMA)|>\frac\zeta8(W^+\wedge W^-),
	\end{align*}
	and thus $x\in\cM_3[h]$.

	Thus, to complete the proof we need to estimate $|\cM_1[h]|,|\cM_2[h]|,|\cM_3[h]|$.
	\Lem s~\ref{lem_1dev} and~\ref{lem_0dev} readily show that 
	\begin{align}\label{eq_prop_dist_psi_1}
		|\cM_1[h]|+|\cM_2[h]|\leq k\exp\bc{-\Omega\bcfr{\log n}{(\log\log n)^4}}.
	\end{align}
	Furthermore, in order to bound $|\cM_3[h]|$ we will employ \Lem~\ref{lemma_endgame_misclassified}.
	Specifically, consider $x\in\cM_3[h]$.
	Since $W^+\wedge W^-=\Omega(\Delta)=\Omega(\log n)$ for $p_{10}>0$ by \eqref{eqweights} and \eqref{eqWest}, there exists $j\in[s]$ such that
	\begin{align*}
		|\,|\vW_{x,j}^+(\tau)|-|\vW_{x,j}^+(\SIGMA)|\,|&>\log^{1/2}n&\mbox{or}&&|\,|\vW_{x,j}^-(\tau)|-|\vW_{x,j}^-(\SIGMA)|\,|>\log^{1/2}n.
	\end{align*}
	Since $W^+=\Omega(\log n)$  for $p_{10}=0$ there exists $j\in[s]$ such that
	\begin{align*}
		|\,|\vW_{x,j}^+(\tau)|-|\vW_{x,j}^+(\SIGMA)|\,|&>\log^{1/2}n
	\end{align*}
	
	Assume without loss that the left inequality holds.
	Then there are at least $\log^{1/2}n$ tests $a\in\partial x\cap F^+[h+j-1]$ such that
	$\partial a\cap(\cM[1]\cup\cdots\cup\cM[h-1])\neq\emptyset$.
	Therefore, \Lem~\ref{lemma_endgame_misclassified} shows that
	\begin{align*}
		|\cM_3[h]|\leq\frac{|\cM[h-s]\cup\cdots\cup\cM[h-1]|}{8\log\log n}.
	\end{align*}
	Hence, using induction to bound $|\cM[h-s]|,\ldots,|\cM[h-1]|$ and recalling from \eqref{eqs} that $s\leq1+\log\log n$, we obtain \eqref{eq_lemma_endgame_misclassified_1} for $i=h$. 
\end{proof}


\subsection{Proof of \Lem~\ref{lem_1dev}}\label{Sec_lem_1dev}

The thrust of \Lem~\ref{lem_1dev} is to verify that the definition \eqref{eqVi+} of the set $V^+[i]$ faithfully reflects the typical statistics of positively/negatively displayed tests in the neighbourhood of an infected individual $x\in V_1[i]$ with $s<i\leq\ell$.
Recall from the definition of $\Gsp$ that such an individual $x$ joins tests in $F[i+j-1]$ for $j\in[s]$.
Moreover, apart from $x$ itself a test $a\in F[i+j-1]\cap\partial x$ recruits from $V[i+j-s],\ldots,V[i+j-1]$.
In particular, $a$ recruits participants from the $s-j$ compartments $V[i+j-s],\ldots,V[i-1]$ preceding $V[i]$.
Let
\begin{align}\label{eqUij}
	\cU_{i,j}={\cbc{a\in F[i+j-1]:(V_1[1]\cup\cdots\cup V_1[i-1])\cap\partial a=\emptyset}}={\cbc{a\in F[i+j-1]:\bigcup_{h=i+j-s}^{i-1}V_1[h]\cap\partial a=\emptyset}}
\end{align}
be the set of tests in $F[i+j-1]$ that do not contain an infected individual from $V[i+j-s],\ldots,V[i-1]$.

\begin{claim}\label{Claim_Qxj}
	With probability $1-o(n^{-2})$ for all $s\leq i<\ell$, $j\in[s]$ we have
	\begin{align*}
		|\cU_{i,j}|&=\bc{1+O(n^{-\Omega(1)})}\frac m{\ell}\cdot\exp(d(j-s)/s).
	\end{align*}
\end{claim}
\begin{proof}
	We condition on the event $\fE$ that {\bf G1} from \Prop~\ref{prop_basic} holds.
	Then for any $a\in F[i+j-1]$ and any $i+j-s\leq h<i$ 
	the number $V_1[h]\cap\partial a$ of infected individuals in $a$ from $V[h]$ is a binomial variable as in \eqref{eqG31}.
	Since the random variables $(V_1[h]\cap\partial a)_h$ are mutually independent, we therefore obtain
		\begin{align}\label{eqClaim_Amin1_1a}
			\pr\brk{(V_1[i+j-s]\cup\cdots V_1[i-1])\cap\partial a=\emptyset\mid\fE}=\exp(d(j-s)/s)+O(n^{-\Omega(1)}).
		\end{align}
	Hence,
	\begin{align}\label{eqClaim_Amin1_2}
		\Erw\brk{|\cU_{i,j}|\mid\fE}&=\frac m\ell\exp\brk{d(j-s)/s+O(n^{-\Omega(1)})}.
	\end{align}
	Further, changing the set $\partial x$ of tests that a single $x\in V_1$ joins can alter $|\cU_{i,j}|$ by $\Delta$ at the most, while changing the set of neighbours of any $x\in V_0$ does not change $|\cU_{i,j}|$ at all.
	Therefore, Azuma--Hoeffding shows that
	\begin{align}\label{eqClaim_Amin1_3}
		\pr\brk{\abs{|\cU_{i,j}|-\Erw[|\cU_{i,j}|\mid\fE]}>\sqrt m\log^2n\mid\fE}&\leq2\exp\bc{-\frac{m\log^4n}{2k\Delta^2}}=o(n^{-2}).
	\end{align}
	The assertion follows from \eqref{eqClaim_Amin1_2}--\eqref{eqClaim_Amin1_3}.
\end{proof}

Let 
\begin{align}\label{eqq1}
	q_{1,j}^-&=p_{01}\exp(d(j-s)/s),&q_{1,j}^+&=p_{11}\exp(d(j-s)/s).
\end{align}

\begin{claim} \label{Claim_unexp}
	For all $s<i\leq\ell$, $x \in V_1[i]$ and $j \in [s]$ we have
	\begin{align*}
		\pr\brk{|\wxj xj^+|<(1-\zeta/2)q_{1,j}^+\frac\Delta s}+\pr\brk{|\wxj xj^-|>(1+\zeta/2)q_{1,j}^-\frac\Delta s}&\leq\exp\bc{-\Omega\bcfr{\log n}{(\log\log n)^4}}.
	\end{align*}
\end{claim}			
\begin{proof}
	Fix $s\leq i<\ell,1\leq j\leq s$ and $x\in V_1[i]$.
	In light of \Prop~\ref{prop_basic} and Claim~\ref{Claim_Qxj} the event $\fE$ that {\bf G1} from \Prop~\ref{prop_basic} holds and that $|\cU_{i,j}|=\exp(d(j-s)/s)\frac m\ell(1+O(n^{-\Omega(1)}))$ has probability
	\begin{align}\label{eqClaim_unexp_1}
		\pr\brk\fE=1-o(n^{-2}).
	\end{align}

	Let $\vU_{i,j}=|\partial x\cap\cU_{i,j}|$.
	Given $\cU_{i,j}$ we have
	\begin{align*}
		\vU_{i,j}\disteq\Hyp\bc{\frac m\ell+O(1),\exp\bcfr{d(j-s)}s\frac m\ell(1+O(n^{-\Omega(1)})),\frac\Delta s}.
	\end{align*}
	Therefore, \Lem~\ref{lem_hyperchernoff} shows that the event
	$$\fE'=\cbc{\abs{\vU_{i,j}-\frac\Delta s\exp\bcfr{d(j-s)}s}\geq\frac{\log n}{(\log\log n)^2}}$$
has conditional probability
	\begin{align}\label{eqClaim_unexp_2}
		\pr\brk{\fE'\mid\fE}\leq\exp\bc{-\Omega\bcfr{\log n}{(\log\log n)^4}}.
	\end{align}

	Finally, given $\vU_{i,j}$ the number $|\vW_{x,j}^+|$ of positively displayed $a\in\partial x\cap \cU_{i,j}$ has distribution $\Bin(\vU_{i,j},p_{11})$; similarly, $|\vW_{x,j}^-|$ has distribution $\Bin(\vU_{i,j},p_{01})$.
	Thus, since $\Delta=\Theta(\log n)$ while $s=O(\log\log n)$ by \eqref{eqs}, the assertion follows from \eqref{eqClaim_unexp_1}, \eqref{eqClaim_unexp_2} and \Lem~\ref{lem_chernoff}.
\end{proof}
			
\begin{proof}[Proof of \Lem~\ref{lem_1dev}]
	The lemma follows from Claim~\ref{Claim_unexp} and Markov's inequality.
\end{proof}
			
\subsection{Proof of \Lem~\ref{lem_0dev}}\label{sec_lem_0dev}
To prove \Lem~\ref{lem_0dev} we need to estimate the probability that an uninfected individual $x\in V[i]$, $s<i\leq\ell$, `disguises' itself to look like an infected individual.
In addition to the sets $\cU_{i,j}$ from \eqref{eqUij}, towards the proof of \Lem~\ref{lem_0dev} we also need to consider the sets
\begin{align*}
	\cP_{i,j}&=F_1[i+j-1]\cap\cU_{i,j},&\cN_{i,j}&=F_0[i+j-1]\cap\cU_{i,j},\\
	\cP^\pm_{i,j}&=F_1^{\pm}[i+j-1]\cap\cU_{i,j},&\cN^\pm_{i,j}&=F_0^{\pm}[i+j-1]\cap\cU_{i,j}.
\end{align*}
In words, $\cP_{i,j}$ and $\cN_{i,j}$ are the sets of actually positive/negative tests in $\cU_{i,j}$, i.e.\ actually positive/negative tests $a\in F[i+j-1]$ that do not contain an infected individual from $V[i+j-s]\cup\cdots\cup V[i-1]$.
Additionally, $\cP^\pm_{i,j},\cN^\pm_{i,j}$ discriminate based on the displayed tests results.
We begin by estimating the likely sizes of these sets.

\begin{claim}\label{lem_U}
	Let $s<i\leq\ell$ and $j\in[s]$.
	Then with probability $1-o(n^{-2})$ we have
\begin{align}
	|\cP_{i,j}|&=\bc{1+O(n^{-\Omega(1)})}\frac m{\ell}\cdot\frac{\exp(dj/s)-1}{\exp(d)},&|\cN_{i,j}|&=\bc{1+O(n^{-\Omega(1)})}\frac m{\ell}\cdot\exp(-d),\label{eq_lem_U_1}\\
	|\cP^+_{i,j}|&=\bc{1+O(n^{-\Omega(1)})}\frac m{\ell}\cdot\frac{p_{11}(\exp(dj/s)-1)}{\exp(d)},&|\cP^-_{i,j}|&=\bc{1+O(n^{-\Omega(1)})}\frac m{\ell}\cdot\frac{p_{10}(\exp(dj/s)-1)}{\exp(d)},\label{eq_lem_U_2}\\
	|\cN^+_{i,j}|&=\bc{1+O(n^{-\Omega(1)})}\frac m{\ell}\cdot p_{01}\exp(-d),&|\cN^-_{i,j}|&=\bc{1+O(n^{-\Omega(1)})}\frac m{\ell}\cdot p_{00}\exp(-d).\label{eq_lem_U_3}
\end{align}
\end{claim}
\begin{proof}
	Since $\cN_{i,j}=F_0[i+j-1]$, the second equation in \eqref{eq_lem_U_1} just follows from \Prop~\ref{prop_basic}, {\bf G2}.
	Furthermore, since $\cP_{i,j}=\cU_{i,j}\setminus\cN_{i,j}$, the first equation \eqref{eq_lem_U_1} is an immediate consequence of Claim~\ref{Claim_Qxj}.
	Finally, to obtain \eqref{eq_lem_U_2}--\eqref{eq_lem_U_3} we simply notice that given $|\cP_{i,j}|,|\cN_{i,j}|$ we have
	\begin{align*}
		|\cP^+_{i,j}|&=\Bin(|\cP_{i,j}|,p_{11}),& |\cP^-_{i,j}|&=\Bin(|\cP_{i,j}|,p_{10}),& |\cN^+_{i,j}|&=\Bin(|\cN_{i,j}|,p_{01}),& |\cN^-_{i,j}|&=\Bin(|\cN_{i,j}|,p_{00}).
	\end{align*}
	Hence, \eqref{eq_lem_U_2}--\eqref{eq_lem_U_3} just follow from \eqref{eq_lem_U_1} and \Lem~\ref{lem_chernoff}.
\end{proof}

Let $\fU$ be the event that {\bf G1--G2} from \Prop~\ref{prop_basic} hold and that the estimates \eqref{eq_lem_U_1}--\eqref{eq_lem_U_3} hold.
Then by \Prop~\ref{prop_basic} and Claim~\ref{lem_U} we have
\begin{align}\label{eq_U}
	\pr\brk\fU=1-o(n^{-2}).
\end{align}
To facilitate the following computations we let
\begin{align}
	q_{0,j}^-&=\exp(-d)p_{00}+(\exp(d(j-s)/s)-\exp(-d))p_{10},&q_{0,j}^+&=\exp(-d)p_{01}+(\exp(d(j-s)/s)-\exp(-d))p_{11}.\label{eqqj}
\end{align}
Additionally, we introduce the shorthand $\lambda=(\log\log n)/\log^{3/7}n$ for the error term from the definition \eqref{eqVi+} of $V^+[i]$.
Our next step is to determine the distribution of the random variables $|\wxj xj^{\pm}|$ that contribute to $\vW_x^\pm(\SIGMA)$ from \eqref{eqWtau}.

\begin{claim} \label{Lemma_wxj_xj}
	Let $s<i\leq\ell$ and $j\in[s]$.
	Given $\fU$ for every $x \in V_{0}^+[i]$ we have
	\begin{align}
		|\wxj xj^-| \sim \Hyp \bc{\bc{1 + O(n^{-\Omega(1)})}\frac {mq_0^-}{\ell},\bc{1 + O(n^{-\Omega(1)})} \frac{mq_{0,j}^-}{2\ell},(1+O(\lambda))\frac\Delta s p_{10}}\label{eqLemma_wxj-_xj},\\
		|\wxj xj^+| \sim \Hyp \bc{\bc{1 + O(n^{-\Omega(1)})}\frac {mq_0^+}{\ell},\bc{1 + O(n^{-\Omega(1)})} \frac{mq_{0,j}^+}{2\ell},\bc{1+O(\lambda)}\frac\Delta s p_{11}}\label{eqLemma_wxj+_xj}.
	\end{align}
\end{claim}
\begin{proof}
	The definition \eqref{eqVi+} of $V^+[i]$ prescribes that for any $x\in V_0^+[i]$ we have $\partial x\cap F^+[i+j-1]=(p_{11}+O(\lambda))\Delta/s$.
	The absence or presence of $x\in V_0[i]$ in any test $a$ does not affect the displayed results of $a$, because $x$ is uninfected.
	Therefore,  the conditional distributions of $|\vW^{\pm}_{x,j}|$ read
	\begin{align}\nonumber
		|\wxj xj^\pm|\sim \Hyp \bc{|F^\pm[i+j-1]|,|\cN_{i,j}^\pm|+|\cP_{i,j}^\pm|,(1+O(\lambda))\frac\Delta s p_{10}}.
	\end{align}
	Since on $\fU$ the bounds \eqref{eqG3_1}--\eqref{eqG3_4} and \eqref{eq_lem_U_1}--\eqref{eq_lem_U_3} hold, the assertion follows.
\end{proof}

We are now ready to derive an exact expression for the probability that for $x\in V_0[i]$ the value $\vW_x^{+}$ gets (almost) as large as the value $W^+$ that we would expected to see for an infected individual.
Recall the values $q_{1,j}^\pm$ from \eqref{eqq1}.

\begin{claim}\label{lem_happen1}
	Let 
	\begin{align}\label{eqM+}
		\cM^+=&\min\frac1s\sum_{j=1}^{s} 
		\KL{z_j}{\frac{q_{0,j}^+}{q_0^+}}\qquad \mbox{s.t.}\quad\sum_{j=1}^{s}\bc{z_j-\frac{q_{1,j}^+}{p_{11}}} w_j^+ = 0,\qquad z_1,\ldots,z_{s}\in(0,1).
	\end{align}
	Then for all $s < i \leq \ell$ and all $x\in V[i]$ we have
	\begin{align*}
		\pr\brk{\wx x^+ \geq (1-2\zeta)W^+ \mid \fU,\,x\in V_0^+[i]} &\leq \exp(-(1+o(1))p_{11}\Delta\cM^+).
	\end{align*}
\end{claim}
\begin{proof} 
	Since $\Delta=O(\log n)$ and $s=O(\log\log n)$ by \eqref{eqs} we can write
	\begin{align}\nonumber
		\Pr&\brk{\wx x^+(\SIGMA)\geq(1-2\zeta)W^+ \mid \fU,\,x\in \zeroplusi {i}}\\
		&\leq \sum_{0\leq v_1, \ldots, v_s\leq\Delta/s}\vecone\cbc{\sum_{j=1}^nv_jw_j^+\geq(1-2\zeta)W^+}\pr\brk{\forall j\in[s]:\vW^+_{x,j}(\SIGMA)=v_j\mid \fU,\,x\in \zeroplusi {i}}
		\nonumber\\
		   & \leq \exp(o(\Delta))\max_{0\leq v_1, \ldots, v_s\leq\Delta/s}\vecone\cbc{\sum_{j=1}^nv_jw_j^+\geq(1-2\zeta)W^+}\pr\brk{\forall j\in[s]:\vW^+_{x,j}(\SIGMA)=v_j\mid \fU,\,x\in \zeroplusi {i}}.
		   \label{eq_lagrange_1}
	\end{align}
	Further, given $\fU$ and $x\in V_0^+[i]$ the random variables $(\vW^+_{x,j})_{j\in[s]}$ are mutually independent because $x$ joins tests in the compartments $F[i+j-1]$, $j\in[s]$, independently.
	Hence, Claim~\ref{Lemma_wxj_xj}  shows that for any $v_1,\ldots,v_s$,
	\begin{align}
		\Pr\brk{\forall j\in[s]:\vW^+_{x,j}(\SIGMA)=v_j\mid \fU,\,x\in \zeroplusi {i}} &=\prod_{j=1}^s\Pr\brk{\vW^+_{x,j}(\SIGMA)=v_j\mid \fU,\,x\in \zeroplusi {i}}. 
		   \label{eq_lagrange_1a}
	\end{align}
	Thus, consider the optimisation problem
	\begin{align}\label{eq_lagrange_9}
		\cM^+_t=&\min\frac1s\sum_{j=1}^{s} \KL{z_j}{\frac{q_{0,j}^+}{q_0^+}}\qquad\mbox{s.t.}\quad\sum_{j=1}^{s}z_j w_j^+ \geq(1-t)W^+,\quad z_1,\ldots,z_{s}\in[q_{0,j}^+/q_0^+,1].
	\end{align}
	Then combining \eqref{eq_lagrange_1} and \eqref{eq_lagrange_1a} with Claim~\ref{Lemma_wxj_xj} and \Lem~\ref{lem_hyperchernoff} and using the substitution $z_j=v_j/(p_{11}\Delta)$, we obtain
	\begin{align}
		\Pr&\brk{\wx x^+\geq(1-2\zeta)W^+ \mid \fU,\,x\in \zeroplusi {i}}\leq \exp(-\Delta p_{11}\cM^+_{2\zeta}+o(\Delta)). \label{eq_lagrange_1b}
	\end{align}

	We claim that \eqref{eq_lagrange_1b} can be sharpened to
	\begin{align}
		\Pr&\brk{\wx x^+\geq(1-2\zeta)W^+ \mid \fU,\,x\in \zeroplusi {i}}
		\leq \exp(-\Delta p_{11}\cM^+_0+o(\Delta)). \label{eq_lagrange_1c}
	\end{align}
	Indeed, consider any feasible solution $z_1,\ldots,z_s$ to $\cM_\zeta^+$ such that $\sum_{j\geq s}z_jw_j^+<W^+$.
	Obtain $z_1',\ldots,z_s'$ by increasing some of the values $z_1,\ldots,z_s$ such that $\sum_{j\leq s}z_j'w_j^+=W^+$.
	Then because the functions $z\mapsto\Kl z{q_{0,j}^+/q_0^+}$ with $j\in[s]$ are equicontinuous on the compact interval $[0,1]$, we see that
	\begin{align*}
		\frac1s\sum_{j=1}^s\KL{z_j}{q_{0,j}^+/q_0^+}\geq\frac1s\sum_{j=1}^s\KL{z_j'}{q_{0,j}^+/q_0^+}+o(1)
	\end{align*}
	uniformly for all $z_1,\ldots,z_s$ and $z_1',\ldots,z_s'$.
	Thus, \eqref{eq_lagrange_1c} follows from \eqref{eq_lagrange_1b}. 

	Finally, we notice that the condition $z_j\geq q_{0,j}^+/q_0^+$ in \eqref{eq_lagrange_9} is superfluous.
	Indeed, since  $\Kl{q_{0,j}^+/q_0^+}{q_{0,j}^+/q_0^+}=0$  there is nothing to be gained from choosing $z_j<q_{0,j}^+/q_0^+$.
	Furthermore, since the Kullback-Leibler divergence is strictly convex and \eqref{eqnoise} ensures that $q_{0,j}^+/q_0^+<q_{1,j}^+/p_{11}$ for all $j$, the optimisation problem $\cM^+_0$ attains a unique minimum, which is not situated at the boundary of the intervals $[q_{0,j}^+/q_0^+,1]$.
	That said, the unique minimiser satisfies the weight constraint $\sum_{j\geq s}z_jw_j^+$ with equality; for otherwise we could reduce some $z_j$, thereby decreasing the objective function value.
	In summary, we conclude that $\cM_0^+=\cM^+$.
	Thus, the assertion follows from \eqref{eq_lagrange_1c}.
\end{proof}

\begin{claim}\label{lem_happen0}
	Let
	\begin{align}\label{eqM-}
		\cM^-=&
		\min\frac1s\sum_{j=1}^{s} \KL{z_j}{\frac{q_{0,j}^-}{q_0^-}}\qquad\mbox{s.t.}\quad\sum_{j=1}^{s}\bc{z_j-\frac{q_{1,j}^-}{p_{01}}} w_j^- = 0,\qquad z_1,\ldots,z_{s}\in(0,1).
	\end{align}
	Then for all $s < i \leq \ell$ and all $x\in V[i]$ we have
	\begin{align*}
		\pr\brk{\wx x^- \leq (1+2\zeta)W^- \mid \fU,\,x\in V_0^+[i]} &\leq \exp(-(1+o(1))p_{10}\Delta\cM^-).
	\end{align*}
\end{claim}
Recall that by convention $0 \cdot \infty = 0$. 
Thus for $p_{10}=0$ the condition of \eqref{eqM-} boils down to  $z_j = {q_{1,j}^-}/{p_{01}}$ and the optimisation becomes trivial.

\begin{proof}
	Analogous to the proof of Claim~\ref{lem_happen1}.
\end{proof}

Clearly, as a next step we need to solve the optimisation problems \eqref{eqM+} and \eqref{eqM-}.
We defer this task to \Sec~\ref{sec_calc}, where we prove the following.

\begin{lemma} \label{lem_calc}
Let $\vX$ have distribution $\Be(\exp(-d))$ and let $\vY$ be the (random) channel output given input $\vX$.
Then
	\begin{align*}
		p_{11}\cM^++p_{10}\cM^-=&\frac{I(\vX,\vY)}d-\KL{p_{11}}{q_0^+}.
	\end{align*}
\end{lemma}

\begin{proof}[Proof of \Lem~\ref{lem_0dev}]
In light of Claims~\ref{lem_happen1} and \ref{lem_happen0} and \Lem~\ref{lem_calc} to work out that all but $o(k)$ positive individuals are identified correctly, using Markov's inequality we need to verify that
	\begin{align}\label{eqLaqBegin}
		\abs{V_0^+} \exp \left(-\Delta \left(p_{11} \cM^+ + p_{10}\cM^-\right) \right) <&  k
	\end{align}
	Taking the logarithm of \eqref{eqLaqBegin} and simplifying, we arrive at
	\begin{align}
		\log(n)\left(1-cd(1-\theta) \left(\KL{p_{11}}{(1-\exp(-d))p_{11}+\exp(-d)p_{01}}+p_{11}\cM^++p_{10}\cM^-\right) \right) < \theta \log(n).
	\end{align}
	Thus we need to show that 
	\begin{align}\label{eqLagFinal}
		cd\brk{\KL{p_{11}}{(1-\exp(-d))p_{11}+\exp(-d)p_{01}}+p_{11}\cM^++p_{10}\cM^-}>&1.
	\end{align}
	This boils down to $cI(\vX,\vY)>1$, which in turn is identical to \eqref{eqprop_dist_psi}.
\end{proof}

\subsection{Proof of \Lem~\ref{lem_calc}}\label{sec_calc}

We tackle the optimisation problems $\cM^{\pm}$ via the method of Lagrange multipliers.
Since the objective functions are strictly convex, these optimisation problems possess unique stationary points.
As the parameters from \eqref{eqq1} satisfy $q_{1,j}^+/p_{11}=\exp(d(j-s)/s)$, the optimisation problem \eqref{eqM+} gives rise to the following Lagrangian.

\begin{claim} \label{lem_lagrange}
The Lagrangian
\begin{align*}
	\cL^+ =& \sum_{j=1}^s \KL{z_j}{\frac{q_{0,j}^+}{q_0^+}} + \lambda w_j^+\bc{z_j - \exp(-d(s-j)/s)}
\end{align*}
has the unique stationary point $z_j=\exp(-d(s-j)/s)$, $\lambda=-1$.
\end{claim}
\begin{proof}
	Since the objective function $\sum_{j=1}^s \KL{z_j}{q_{1,j}^+/p_{11}}$ is strictly convex, we just need to verify that $\lambda=-1$ and $z_j=\exp(-d(s-j)/s)$ is a stationary point.
	To this end we calculate
	\begin{align}\label{eqlem_lagrange1}
		\frac{\partial\cL^+}{\partial z_j} = & \log\frac{z_j^+} {1-z_j^+}- \log\frac {q_{1,j}^+} {p_{11}-q_{1,j}^+} + \lambda w_j^+,&
	\frac{\partial\cL^+}{\partial \lambda} =& \sum_{j=1}^s \bc{z_j - \exp\bc{-d(s-j)/s}}  w_j^+.
\end{align}
Substituting in the definition \eqref{eqweights} of the weights $w_j^+$ 
and the definitions \eqref{eqq1} of $p_{11},q_{1,j}^+$ and simplifying, we obtain
\begin{align*}
	\frac{\partial L^+}{\partial z_j}\bigg|_{\substack{z_j=\exp(-d(s-j)/s)\\\lambda=-1}}= 
	& 
	\log\frac {\exp(-d(s-j)/s)}{1-\exp(-d(s-j)/s)}
	- \log\frac{p_{11}(\exp\bc{dj/s} -1) + p_{01}}{p_{11}(\exp(d)-1)+p_{01}}\\&
		+\log \bc{1-\frac{p_{11}(\exp(d  j/s) -1) + p_{01}}{p_{11}(\exp(d)-1)+p_{01}} } 
	-\log\frac{p_{11}}{p_{11} + (p_{01}-p_{11}) \exp(-d  j/s )}=0. 
\end{align*}
Finally, \eqref{eqlem_lagrange1} shows that setting $z_j=\exp(-d(s-j)/s)$ ensures that $\partial\cL^+/\partial\lambda=0$ as well.
\end{proof}

Analogous steps prove the corresponding statement for $\cM^-$.

\begin{claim} \label{lem_lagrange-}
	Assume $p_{10} > 0$. The Lagrangian
\begin{align*}
	\cL^- =& \sum_{j=1}^s \KL{z_j}{\frac{q_{0,j}^-}{p_{01}}} + \lambda w_j^-\bc{z_j - \exp(-d(s-j)/s)}
\end{align*}
has the unique stationary point $z_j=\exp(-d(s-j)/s)$, $\lambda=-1$.
\end{claim}

Having identified the minimisers of $\cM^{\pm}$, we proceed to calculate the optimal objective function values. 
Note that for $\cM^-$ the minimisers $z_j$ for the cases $p_{10}>0$ and $p_{10} = 0$ coincide.

\begin{claim} \label{lem_int}
	Let
	\begin{align*}
		\lambda^+&=\log(q_0^+)=\log(p_{01}\exp(-d)+p_{11}(1-\exp(-d))),&\lambda^-&=\log(q_0^-)=\log(p_{00}\exp(-d)+p_{10}(1-\exp(-d))).
	\end{align*}
	Then 
	\begin{align*}
		p_{11}d\exp(d)\cM^+=&(\lambda^++d)\bc{p_{11}\bc{(d-1)\exp(d)+1}-p_{01}}+p_{01}\log(p_{01})\\&\qquad-p_{11}\bc{(d-1)\exp(d)+1}\log(p_{11})-d(d-1)\exp(d)p_{11}+O(1/s),\\
		p_{10}d\exp(d)\cM^-=&(\lambda^-+d)\bc{p_{10}\bc{(d-1)\exp(d)+1}-p_{00}}+p_{00}\log(p_{00})\\&\qquad-p_{10}\bc{(d-1)\exp(d)+1}\log(p_{10})-d(d-1)\exp(d)p_{10}+O(1/s)\quad\mbox{if $p_{10}>0$}.
	\end{align*}
\end{claim}
\begin{proof}
	We perform the calculation for $\cM^+$ in detail. 
	Syntactically identical steps yield the expression for $\cM^-$, the only required change being the obvious modification of the indices of the channel probabilities.
	Substituting the optimal solution $z_j=\exp(-d(s-j)/s)$ and the definitions \eqref{eqq1} and \eqref{eqq}, \eqref{eqq1} and \eqref{eqqj} of $q^+_{1,j},q_0^+,q_{0,j}^+$ into \eqref{eqM+}, we obtain 
	\begin{align}\label{eq_lem_int_1}
		\cM^+=&\frac1s\sum_{j=1}^{s}\KL{\exp(-d(s-j)/s)}{\frac{p_{01}+(\exp(dj/s)-1)p_{11}}{p_{01}+(\exp(d)-1)p_{11}}}=\cI^++O(1/s),\qquad\mbox{where}\\
		\cI^+=&\int_0^1\KL{\exp(d(x-1))}{\frac{p_{11}(\exp(dx)-1)+p_{01}}{p_{11}(\exp(d)-1)+p_{01}}}\dd x;\nonumber
	\end{align}
	the $O(1/s)$-bound in \eqref{eq_lem_int_1} holds because the derivative of the integrand $x\mapsto\KL{\exp(d(x-1))}{\frac{p_{11}(\exp(dx)-1)+p_{01}}{p_{11}(\exp(d)-1)+p_{01}}}$ is bounded on $[0,1]$.
	Replacing the Kullback-Leibler divergence by its definition, we obtain $\cI^+=\cI^+_1+\cI^+_2$, where
	\begin{align*}
		\cI^+_1&=\int_{0}^{1}  \exp\bc{d(x-1)} \log \frac{\exp\bc{d(x-1)} ( p_{11}(\exp\bc{d}-1)+p_{01} ) }  {p_{11}(\exp\bc{dx}-1)+p_{01} } \dd x,\\
		\cI^+_2&=\int_{0}^{1} (1-\exp\bc{d(x-1)}) \log \brk{ \frac{ 1- \exp\bc{d(x-1)} }{ 1 - \frac{p_{11}(\exp\bc{dx}-1)+p_{01}}{p_{11}(\exp\bc{d}-1)+p_{01}}  }  } \dd x.
	\end{align*}
	Splitting the logarithm in the first integrand, we further obtain
	$\cI^+_1=\cI_{11}^++\cI_{12}^+$, where
\begin{align*}
	\cI^+_{11}&=\int_{0}^{1}  \exp\bc{d(x-1)} \log  \brk{ \exp\bc{d(x-1)} ( p_{11}(\exp\bc{d}-1)+p_{01} )}  \dd x,\\
\cI^+_{12}&=-\int_{0}^{1}   \exp\bc{d(x-1)} \log  \brk{ p_{11}(\exp\bc{dx}-1)+p_{01} }  \dd x.
\end{align*}
	Setting $\Lambda^+=\log( p_{11}(\exp\bc{d}-1)+p_{01})=\lambda^++d$ and introducing $ u =  \exp\bc{d(x-1)}$, we calculate
	\begin{align}
		\cI_{11}^+&=\frac{1}{d} \brk{\int_{\exp\bc{-d}}^{1} \log( u)\dd u + \int_{\exp\bc{-d}}^{1} \Lambda^+ \dd u  }=\frac{1}{d} \brk{(d+1)\exp(-d) -1  + (1 - \exp\bc{-d} ) \Lambda^+}.\label{eq_lem_int_2}
	\end{align}
	Concerning $\cI_{12}^+$, we once again substitute $ u =  \exp\bc{d(x-1)}$ to obtain
	\begin{align}\nonumber
		\cI_{12}^+&=-\frac1d\int_{\exp(-d)}^1\log\bc{p_{11}\exp(d)u+p_{01}-p_{11}}\dd u\\
				  &= - \frac{1}{d} \brk{ \bc{\frac{ p_{01}}{p_{11}}\exp\bc{-d}-\exp\bc{-d}+1}\Lambda^+ -  \exp\bc{-d}\frac{p_{01}}{p_{11}}  \log(p_{01}) + \exp\bc{-d}-1 }\label{eq_lem_int_3}
	\end{align}
	Proceeding to $\cI_2$, we obtain
	\begin{align}\nonumber
		\cI_2^+ &= \int_{0}^{1} (1- \exp\bc{d(x-1)}) \log\frac{ p_{11} (\exp\bc{d} - 1) + p_{01}   }{ p_{11}\exp(d)  }   \dd x\\
& = \int_{0}^{1} (1- \exp\bc{d(x-1)}) \log \bc{  p_{11} (\exp(d) - 1) + p_{01} } \dd x  - \int_{0}^{1} (1- \exp\bc{d(x-1)})  \log(  p_{11}\exp(d) ) \dd x\nonumber \\
&=\frac{ (\Lambda^+ -  \log(  p_{11}) - d ) ( d-1 + \exp\bc{-d} )}{d}.\label{eq_lem_int_4}
	\end{align}
	Finally, recalling that $\cI^+=\cI_1+\cI_2=\cI^+_{11}+\cI^+_{12}+\cI_2^+$ and combining \eqref{eq_lem_int_1}--\eqref{eq_lem_int_4} and simplifying, we arrive at the desired expression for $\cM^+$.
\end{proof}

\begin{proof}[Proof of \Lem~\ref{lem_calc}]
	We have
	\begin{align*}
		I(\vX,\vY)&=H(\vY)-H(\vY|\vX)=h(p_{00}\exp(-d)+p_{10}(1-\exp(-d)))-\exp(-d)h(p_{00})-(1-\exp(-d))h(p_{11}).
	\end{align*}
	Hence, Claim~\ref{lem_int} yields
	\begin{align*}
		p_{11}\cM^++p_{10}\cM^-
		=&-\frac{h(p_{00})}{d\exp(d)}-\frac{(1-\exp(-d))h(p_{11})}{d}+h(p_{11})\\
		 &+\frac{(p_{11}-p_{01})\lambda^++(p_{10}-p_{00})\lambda^-}{d\exp(d)}+\frac{d-1}d\brk{p_{11}\lambda^++p_{10}\lambda^-}\\
		=&-\frac{h(p_{00})}{d\exp(d)}-\frac{(1-\exp(-d))h(p_{11})}{d}-\KL{p_{11}}{q_0^+}-\frac{\lambda^+}d\exp(\lambda^+)-\frac{\lambda^-}d\exp(\lambda^-)\\
		=&\frac{I(\vX,\vY)}d-\KL{p_{11}}{q_0^+},
	\end{align*}
	as desired.
\end{proof}

\section{Analysis of the exact recovery algorithm}\label{sec_exact}

\noindent
In this section we establish \Prop s~\ref{lem_yz} and~\ref{prop_endgame}, which are the building blocks of the proof of \Thm~\ref{thm_alg_ex}.
We continue to work with the spatially coupled design $\Gsc$ from \Sec~\ref{sec_gsc} and keep the notation and assumptions from \eqref{eqell}--\eqref{eqtotaltests}.

\subsection{Proof of \Prop~\ref{lem_yz}}\label{sec_lem_yz}

Assume that $c>\cbound(d,\theta)+\eps$ and let $c'=\cbound(d,\theta)+\eps/2$.
Since $c> c'+\eps/2$ the definitions \eqref{eqcex} of $\cbound(d,\theta)$ and \eqref{eqYinterval} of $\cY(c',d,\theta)$ ensure that for small enough $\delta>0$ we can find an open interval $\cI\subset \cY(c',d,\theta)$ with rational boundary points such that {\bf Z1} is satisfied.

Let $\bar\cI$ be the closure of $\cI$.
Then by the definition of $\cbound(d,\theta)$ there exists a function $y\in\bar\cI\mapsto z_y\in[p_{01},p_{11}]$ such that for all $y\in\bar\cI$ we have
\begin{align}
	cd(1-\theta)\brk{\KL y{\exp(-d)}+y\KL{z_y}{p_{11}}}&=\theta.\label{eq_lem_yz_1}
\end{align}
In fact, because the Kullback-Leibler divergence is strictly convex the equation \eqref{eq_lem_yz_1} defines $z_y$ uniquely.
The inverse function theorem implies that the function $y\mapsto z_y$ is continuous and therefore uniformly continuous on $\bar\cI$.
Additionally, once again because the Kullback-Leibler divergence is convex and $c>\cbound(d,\theta)$, for all $y\in\bar\cI$ we have
\begin{align*}
	cd(1-\theta)\brk{\KL y{\exp(-d)}+y\KL{z_y}{p_{01}}}&>1.
\end{align*}
Therefore, there exists $\hat\delta=\hat\delta(c,d,\theta)>0$ such that for all $y\in\bar\cI$ we have
\begin{align}\label{eq_lem_yz_2}
	cd(1-\theta)\brk{\KL y{\exp(-d)}+y\KL{z_y}{p_{01}}}&>1+\hat\delta.
\end{align}
Combining \eqref{eq_lem_yz_1} and \eqref{eq_lem_yz_2}, we find a continuous $y\in\bar\cI\mapsto \hat z_y$ such that for small enough $\delta>0$ for all $y\in[0,1]$ we have
\begin{align}\label{eq_lem_yz_3}
	cd(1-\theta)\brk{\KL y{\exp(-d)}+y\KL{\hat z_y}{p_{11}}}&>\theta+2\delta,\qquad\mbox{and}\\
	cd(1-\theta)\brk{\KL y{\exp(-d)}+y\KL{\hat z_y}{p_{01}}}&>1+2\delta.\label{eq_lem_yz_4}
\end{align}
Additionally, by uniform continuity for any given $0<\eps'<\delta/2$ (that may depend arbitrarily on $\delta$ and $\cI$) we can choose $\delta'>0$ small enough so that
\begin{align}\label{eq_lem_yz_5}
	|\hat z_y-\hat z_{y'}|&<\eps'/2&\mbox{for all $y,y'\in\bar\cI$ with $|y-y'|<\delta'$}.
\end{align}

Finally, let $y_0,\ldots,y_\nu$ with $\nu=\nu(\delta',\eps')>0$ be a large enough number of equally spaced points in $\bar\cI=[y_0,y_\nu]$.
Then for each $i$ we pick $\cZ(y_i)\in[p_{01},p_{11}]\cap\QQ$ such that $|\hat z_{y_i}-\cZ(y_i)|$ is small enough.
Extend $\cZ$ to a step function $\bar\cI\to\QQ\cap[0,1]$ by letting $\cZ(y)=\cZ(y_{i-1})$ for all $y\in(y_{i-1},y_i)$ for $1\leq i\leq \nu$.
Since $y\mapsto \hat z_y$ is uniformly continuous, we can choose $\nu$ large enough so that \eqref{eq_lem_yz_3}--\eqref{eq_lem_yz_5} imply
\begin{align*}
	cd(1-\theta)\brk{\KL y{\exp(-d)}+y\KL{\cZ(y)}{p_{11}}}&>\theta+\delta&&\mbox{for all }y\in\bar\cI,\\
	cd(1-\theta)\brk{\KL y{\exp(-d)}+y\KL{\cZ(y)}{p_{01}}}&>1+\delta&&\mbox{for all }y\in\bar\cI,\\
	|\cZ(y)-\cZ(y')|&<\eps'&&\mbox{for all $y,y'\in\bar\cI$ such that $|y-y'|<\delta'$},
\end{align*}
as claimed.

\subsection{Proof of \Prop~\ref{prop_endgame}}\label{sec_prop_endgame}
As in the proof of \Prop~\ref{prop_dist_psi} in \Sec~\ref{sec_prop_dist_psi} we will first investigate an idealised scenario where we assume that the ground truth $\SIGMA$ is known.
Then we will use the expansion property provided by \Lem~\ref{lemma_endgame_misclassified} to connect this idealised scenario with the actual steps of the algorithm.

In order to study the idealised scenario, for $x\in V[i]$ and $j\in[s]$ we introduce
\begin{align*}
	\vY_{x,j}&= \abs{ \cbc{ a \in  F[i+j-1]\cap\partial x: V_1\cap\partial a\subseteq\cbc x}},&
	\vZ_{x,j}&= \abs{ \cbc{ a \in  F^+[i+j-1]\cap\partial x: V_1\cap\partial a\subseteq\cbc x}},\\
	\vY_x&= \sum_{j=1}^s\vY_{x,j},&\vZ_x&= \sum_{j=1}^s\vZ_{x,j}.
\end{align*}
Thus, $\vY_{x,j}$ is the number of untainted tests in compartment $F[i+j-1]$ that contain $x$, i.e.\ test that do not contain another infected individual.
Moreover, $\vZ_{x,j}$ is the number of positively displayed untainted tests.
Finally, $\vY_x,\vZ_x$ are the sums of these quantities on $j\in[s]$.
The following lemma provides an estimate of the number of individuals $x$ with a certain value of $\vY_x$.

\begin{lemma}\label{lem_ykl}
	\whp for all $0\leq Y\leq\Delta$ and all $i\in[\ell]$ we have
	\begin{align}\label{eqlem_ykl1}
		\sum_{x\in V_0[i]}\vecone\{\vY_x=Y\}\leq n\exp\bc{-\Delta\KL{Y/\Delta}{\exp(-d)}+o(\Delta)},\\
		\sum_{x\in V_1[i]}\vecone\{\vY_x=Y\}\leq k\exp\bc{-\Delta\KL{Y/\Delta}{\exp(-d)}+o(\Delta)}.\label{eqlem_ykl2}
	\end{align}
\end{lemma}
\begin{proof}
	Let $1\leq i\leq\ell$ and consider any $x\in V[i]$.
	Further, obtain $\Gsp-x$ from $\Gsp$ by deleting individual $x$ (and, of course, removing $x$ from all tests).
	Additionally, obtain $\Gsp'$ from $\Gsp-x$ by re-inserting $x$ and assigning $x$ to $\Delta/s$ random tests in the compartments $F[i+j-1]$ for $j\in[s]$ as per the construction of the spatially coupled test design.
	Then the random test designs $\Gsp$ and $\Gsp'$ are identically distributed.

	Let $\fE$ be the event that $\Gsp$ enjoys properties {\bf G1} and {\bf G2} from \Prop~\ref{prop_basic}.
	Then \Prop~\ref{prop_basic} shows that
	\begin{align}\label{eq_lem_ykl_1}
		\pr\brk{\fE}&=1-o(n^{-2}).
	\end{align}
	Moreover, given $\fE$ for every $j\in[s]$ the number of tests in $F[i+j-1]$ that contain no infected individual aside from $x$ satisfies
	\begin{align}\label{eq_lem_ykl_2}
		\sum_{a\in F[i+j-1]}\vecone\{\partial a\cap V_1\setminus\{x\}=\emptyset\}=(1+O(n^{-\Omega(1)}))\frac m\ell\exp(-d);
	\end{align}
	this follows from the bounds on $F_0[i+j-1]$ provided by {\bf G2} and the fact that discarding $x$ can change the numbers of actually positive/negative tests by no more than $\Delta$.

	Now consider the process of re-inserting $x$ to obtain $\Gsp'$.
	Then \eqref{eq_lem_ykl_2} shows that given $\fE$ we have
	\begin{align*}
	\vY_{x,j}&\sim\Hyp\bc{\frac m\ell+O(1),(1+O(n^{-\Omega(1)}))\frac m\ell\exp(-d),\frac\Delta s}&&(j\in[s]).
	\end{align*}
	These hypergeometric variables are mutually independent given $\Gsp-x$.
	Therefore, \Lem~\ref{lem_hyperchernoff} implies that on $\fE$,
	\begin{align}\label{eq_lem_ykl_3}
		\pr\brk{\vY_x=Y\mid\Gsp-x}&\leq\exp\bc{-\Delta\KL{Y/\Delta}{\exp(-d)}+o(\Delta)}.
	\end{align}
	This estimate holds independently of the infection status $\SIGMA_x$.
	Thus, the assertion follows from \eqref{eq_lem_ykl_1}, \eqref{eq_lem_ykl_3} and Markov's inequality.
\end{proof}

As a next step we argue that for $c$ beyond the threshold $\cbound(d,\theta)$ the function $\cZ$ from \Prop~\ref{lem_yz} separates the infected from the uninfected individuals \whp.

\begin{lemma} \label{lem_distphixstar}
	Assume that $c>c^*(d,\theta)+\eps$.
	Let $\cI=(l,r),\delta>0$ be the interval and the number from \Prop~\ref{lem_yz}, choose $\eps'>0$ sufficiently small and let $\delta',\cZ$ be such that {\bf Z1}--{\bf Z4} are satisfied.
	Then \whp the following statements are satisfied.
	\begin{enumerate}[(i)]
		\item For all $x\in V_1$ we have $\vY_x/\Delta\in(l+\eps',r-\eps')$ and $\vZ_x/\Delta>\cZ(\vY_x/\Delta)+3\eps'$.
		\item For all $x\in V_0$ with $\vY_x/\Delta\in\cI$ we have $\vZ_x/\Delta<\cZ(\vY_x/\Delta)-3\eps'$.
	\end{enumerate}
\end{lemma}
\begin{proof}
	Let $\fE$ be the event that the bounds \eqref{eqlem_ykl1}--\eqref{eqlem_ykl2} hold for all $0\leq Y\leq\Delta$.
	Then \eqref{eqlem_ykl2} and \Prop~\ref{lem_yz}, {\bf Z1} show that \whp $\vY_x/\Delta\in(l+\eps',r-\eps')$ for all $x\in V_1$, provided $\eps'>0$ is small enough.
	Moreover, for a fixed $0\leq Y\leq\Delta$ such that $Y/\Delta\in\cI$ and $i\in[\ell]$ let $\vX_1(Y)$ be the number of variables $x\in V_1[i]$ such that $\vY_x=Y$ and $\vZ_x\leq\Delta\cZ(Y/\Delta)+3\eps'\Delta$.
	Since $x$ itself is infected, all tests $a\in\partial x$ are actually positive. 
	Therefore, $a$ is displayed positively with probability $p_{11}$.
	As a consequence, \Lem~\ref{lem_chernoff} shows that
	\begin{align}\label{eq_lem_distphixstar_1}
		\pr\brk{\vZ_x\leq\Delta\cZ(\vY_x/\Delta)+3\eps'\Delta\mid\vY_x=Y}&\leq\exp\bc{-Y\KL{\cZ(Y/\Delta)+3\eps'\Delta/Y}{p_{11}}+o(\Delta)}.
	\end{align}
	Combining \eqref{eqlem_ykl2} and \eqref{eq_lem_distphixstar_1}, recalling that $\kdef$ and choosing $\eps'>0$ sufficiently small, we obtain
	\begin{align}\nonumber
		\ex&\brk{\sum_{x\in V_1[i]}\vecone\cbc{\vY_x=Y,\,\vZ_x\leq\Delta\cZ(\vY_x/\Delta)+3\eps'\Delta}\mid\fE}\\
		   &\leq k\exp\bc{-\Delta\KL{Y/\Delta}{\exp(-d)}-Y\KL{\cZ(Y/\Delta)+3\eps'\Delta/Y}{p_{11}}+o(\Delta)}\\&\leq n^{-\Omega(1)}\qquad[\mbox{due to \Prop~\ref{lem_yz}, \bf Z2}].\label{eq_lem_distphixstar_2}
	\end{align}
	Taking a union bound on the $O(\log^2n)$ possible combinations $(i,Y)$, we see that (i) follows from \eqref{eq_lem_distphixstar_2}.
	A similar argument based on \Prop~\ref{lem_yz}, {\bf Z3} yields (ii).
\end{proof}
			 
\begin{proof}[Proof of \Prop~ \ref{prop_endgame}]
	For $t = 1 \dots \ceil{\ln n}$ consider the set of misclassified individuals after $t-1$ iterations:
	\begin{align*}
		\cM_t = \cbc{ x \in V[s+1]\cup\cdots V[\ell]: \tau^{(t)}_x \neq \SIGMA_x}.
	\end{align*}
	\Prop s~\ref{prop_plausible} and~\ref{prop_dist_psi} show that \whp the size of the initial set satisfies
	\begin{align}\label{eqprop_endgame1}
		\abs{\cM_1} \leq k\exp\bc{-\Omega(\log^{1/8} n)}.
	\end{align}

	We are going to argue by induction that $|\cM_t|$ decays geometrically.
	Apart from the bound \eqref{eqprop_endgame1}, this argument depends on only two conditions.
	First, that the random graph $\Gsp$ indeed enjoys the expansion property from \Lem~\ref{lemma_endgame_misclassified}.
	Second, that (i)--(ii) from \Lem~\ref{lem_distphixstar} hold.
	Let $\fE$ be the event that these two conditions are satisfied, and that \eqref{eqprop_endgame1} holds.
	\Prop s~\ref{prop_plausible} and~\ref{prop_dist_psi} and \Lem s~\ref{lemma_endgame_misclassified} and~\ref{lem_distphixstar} show that $\pr\brk\fE=1-o(1)$.

	To complete the proof we are going to show by induction on $t\geq2$ that on $\fE$,
	\begin{align}\label{eqprop_endgame2a}
		|\cM_t|\leq|\cM_{t-1}|/3.
	\end{align}
	Indeed, consider the set	
	\begin{align*}
		\cM_{t}^*=\cbc{x\in V[s+1]\cup\cdots V[\ell]: \sum_{a\in\partial x\setminus F[0]}\abs{\partial a \cap \cM_{t-1}\setminus\cbc x} \geq \Delta/\log\log n}.
	\end{align*}
	Since by \eqref{eqprop_endgame1} and induction we know that $\abs{\cM_{t-1}} \leq k\exp\bc{-\Omega(\log^{1/8} n)}$, the expansion property from \Lem~\ref{lemma_endgame_misclassified} implies that $\cM_t^*\leq\cM_{t-1}/3$.
	Therefore, to complete the proof of \eqref{eqprop_endgame2a} it suffices to show that $\cM_{t}\subseteq\cM_{t}^*$.

	To see this, suppose that $x\in\cM_{t}$.
	\begin{description}
		\item[Case 1: $x\in V_1$ but $Y_x(\tau^{(t-1)})/\Delta\not\in\cI$]
			\Lem~\ref{lem_distphixstar} (i) ensures that $\vY_x/\Delta\in(l+\eps',r-\eps')$.
			Therefore, the case $Y_x(\tau^{(t-1)})/\Delta\not\in\cI$ can occur only if at least $\eps'\Delta$ tests $a\in\partial x$ contain a misclassified individual $x'\in\cM_{t-1}$.
			Hence, $x\in\cM_t^*$.
		\item[Case 2: $x\in V_1$ and $Y_x(\tau^{(t-1)})/\Delta\not\in\cI$ but $Z_x(\tau^{(i)})/\Delta\leq\cZ(Y_x(\tau^{(i)})/\Delta)$] by \Lem~\ref{lem_distphixstar} (i) we have $\vZ_x/\Delta>\cZ(\vY_x/\Delta)+2\eps'$.
			Thus, if $Z_x(\tau^{(t-1)})/\Delta\leq\cZ(Y_x(\tau^{(t-1)})/\Delta)$, then by the continuity property {\bf Z4} we have $|Y_x(\tau^{(t-1)})-\vY_x|>\eps'\Delta$.
			Consequently, as in Case~1 we have $x\in\cM_t^*$.
		\item[Case 3: $x\in V_0$]
			as in the previous cases, due to {\bf Z4} and \Lem~\ref{lem_distphixstar} (ii) the event $x\in\cM_t$ can occur only if $|\vY_x-Y_v(\tau^{(t-1)})|>\eps'\Delta$.
			Thus, $x\in\cM_t^*$.
	\end{description}
		Hence, $\cM_{t}\subseteq\cM_{t}^*$, which completes the proof.
\end{proof}

\section{Lower bound on the constant-column design}\label{sec_cc_lower}

\subsection{Proof of \Prop~\ref{prop_cc_ground}}\label{sec_cc_ground}

The following lemma is an adaptation of \Prop~\ref{prop_basic} ({\bf G2}) to $\Gcc$.

\begin{lemma}\label{lemma_basic}
	The random graph $\Gcc$ enjoys the following properties with probability $1-o(n^{-2})$:
	\begin{align}
		m\exp(-d)p_{00}-\sqrt m\ln^3 n&\leq\abs{F_0^-}\leq m\exp(-d)p_{00}+\sqrt m\ln^3 n,\label{eqCC3_1}\\
		m\exp(-d)p_{01}-\sqrt m\ln^3 n&\leq\abs{F_0^+}\leq m\exp(-d)p_{01}+\sqrt m\ln^3 n,\label{eqCC3_2}\\
		m(1-\exp(-d))p_{10}-\sqrt m\ln^3 n&\leq\abs{F_1^-}\leq m(1-\exp(-d))p_{10}+\sqrt m\ln^3 n,\label{eqCC3_3}\\
		m(1-\exp(-d))p_{11}-\sqrt m\ln^3 n&\leq\abs{F_1^+}\leq m(1-\exp(-d))p_{11}+\sqrt m\ln^3 n.\label{eqCC3_4}
	\end{align}
\end{lemma}

\noindent
The proof of \Lem~\ref{lemma_basic} is similar to that of \Prop~\ref{prop_basic} (see \Sec~\ref{sec_prop_basic}).

\begin{proof}[Proof of \Prop~\ref{prop_cc_ground}]
	The definition \eqref{eqpsidef} of the weight functions ensures that
	\begin{align*}
		\log\psi_{\Gcc,\dSIGMA}(\SIGMA)&=|F_0^-|\log p_{00}+|F_0^+|\log p_{01}+|F_1^-|\log p_{10}+|F_1^+|\log p_{11}.
	\end{align*}
	Substituting in the estimates from \eqref{eqCC3_1}--\eqref{eqCC3_4} completes the proof.
\end{proof}

\subsection{Proof of \Prop~\ref{prop_XYZ}}\label{sec_cc_stab}

Let $\cX_r(Y)$ be the set of individuals $x\in V_r$ such that
\begin{align*}
	\sum_{a\in\partial x}\vecone\cbc{\partial a\setminus \{x\}\subseteq V_0}=Y.
\end{align*}
Hence, $x$ participates in precisely $Y$ tests that do not contain another infected individual.

\begin{lemma}\label{lem_EXY}
	Let $y\in\cY(c,d,\theta)$ be such that $y\Delta$ is an integer.
	Then \whp we have
	\begin{align}\label{eq_lem_EXY_X0}
		|\cX_0(y\Delta)|&=n\exp\bc{-\Delta\KL y{\exp(-d)}+o(\Delta)},\\
		|\cX_1(y\Delta)|&=k\exp\bc{-\Delta\KL y{\exp(-d)}+o(\Delta)}.\label{eq_lem_EXY_X1}
	\end{align}
\end{lemma}
\begin{proof}
	Let $Y=y\Delta$ and let $\cE$ be the event that the bounds \eqref{eqCC3_1}--\eqref{eqCC3_4} hold.
	We begin by computing $\ex|\cX_1(Y)|$.
	By exchangeability we may condition on the event $\cS=\{\SIGMA_{x_1}=\cdots=\SIGMA_{x_k}=1\}$, i.e.\ precisely the first $k$ individuals are infected.
	Hence, by the linearity of expectation it suffices to prove that
	\begin{align}\label{eq_lem_EXY_1}
		\pr\brk{x_1\in\cX_1(Y)\mid\cE,\cS}=\exp\bc{-\Delta\KL y{\exp(-d)}+o(\Delta)}.
	\end{align}
	Let $\G'=\Gcc-x_1$ be the random design without $x_1$ and let $F_0'$ be the set of actually negative tests of $\G'$.
	Given $\vm_0'=|F_0'|$ the number of tests $a\in\partial x_1$ such that $\partial a\setminus\{x_1\}\subseteq V_0$ has distribution $\Hyp(m,\vm_0',\Delta)$, because $x_1$ joins precisely $\Delta$ tests independently of all other individuals.
	Hence, \eqref{eqHyp} yields
	\begin{align}\label{eq_lem_EXY_2}
		\pr\brk{x_1\in\cX_1(Y)\mid\cS,\vm_0'}=\binom{\vm_0'}{Y}\binom{m-\vm_0'}{\Delta-Y}\binom{m}{\Delta}^{-1}.
	\end{align}
	Expanding \eqref{eq_lem_EXY_2} via Stirling's formula and using the bounds \eqref{eqCC3_1}--\eqref{eqCC3_2}, we obtain \eqref{eq_lem_EXY_1}, which implies that 
	\begin{align}\label{eq_lem_EXY_3}
		\ex\brk{|\cX_1(y\Delta)|\mid\cE}&=k\exp\bc{-\Delta\KL y{\exp(-d)}+o(\Delta)}.
	\end{align}
	Since the above argument does not depend on the infection status of $x_1$, analogously we obtain
	\begin{align}\label{eq_lem_EXY_4}
		\ex\brk{|\cX_0(y\Delta)|\mid\cE}&=(n-k)\exp\bc{-\Delta\KL y{\exp(-d)}+o(\Delta)}.
	\end{align}

	To turn \eqref{eq_lem_EXY_3}--\eqref{eq_lem_EXY_4} into ``with high probability''-bounds we resort to the second moment method.
	Specifically, we are going to show that
	\begin{align}\label{eq_lem_EXY_5}
		\ex\brk{|\cX_1(y\Delta)|(|\cX_1(y\Delta)|-1)\mid\cE}&\sim\ex\brk{|\cX_1(y\Delta)|\mid\cE}^2,\\
		\ex\brk{|\cX_0(y\Delta)|(|\cX_0(y\Delta)|-1)\mid\cE}&\sim\ex\brk{|\cX_0(y\Delta)|\mid\cE}^2.\label{eq_lem_EXY_6}
	\end{align}
	Then the assertion is an immediate consequence of \eqref{eq_lem_EXY_3}--\eqref{eq_lem_EXY_6} and \Lem~\ref{lemma_basic}.

	For similar reasons as above it suffices to prove \eqref{eq_lem_EXY_5}.
	More precisely, we merely need to show that
	\begin{align}\label{eq_lem_EXY_10}
		\pr\brk{x_1,x_2\in\cX_1(Y)\mid\cE,\cS}\sim\pr\brk{x_1\in\cX_1(Y)\mid\cE,\cS}^2.
	\end{align}
	To compute the probability on the l.h.s.\ obtain $\G''=\Gcc-x_1-x_2$ by removing $x_1,x_2$.
	Let $\vm_0''$ be the number of actually negative tests of $\G''$.
	We claim that on $\cE$,
	\begin{align}\label{eq_lem_EXY_11}
		\pr\brk{x_1,x_2\in\cX_1(Y)\mid\cS,\vm_0''}=\sum_{I=0}^{\Delta-Y}\binom{\vm_0''}{I}\binom{\vm_0''}{Y}\binom{\vm_0''-Y}{Y}\binom{m-\vm_0''}{\Delta-Y-I}\binom{m-\vm_0''}{\Delta-Y-I} \binom{m}{\Delta}^{-2}.
	\end{align}
	Indeed, we first choose $0\leq I\leq\Delta-Y$ tests that are actually negative in $\G''$ that both $x_1,x_2$ will join.
	Observe that these tests $a$ do {\em not} satisfy $\partial a\setminus\{x_{1/2}\}\subseteq V_0$.
	Then we choose $Y$ distinct actually negative tests for $x_1$ and $x_2$ to join.
	Finally, we choose the remaining $\Delta-Y-I$ tests for $x_1,x_2$ among the actually positive tests of $\G''$.

	Since on $\cE$ the total number $\vm_0''$ is much bigger than $\Delta$, it is easily verified that the sum \eqref{eq_lem_EXY_11} is dominated by the term $I=0$; thus, on $\cE$ we have
\begin{align}\label{eq_lem_EXY_12}
		\pr\brk{x_1,x_2\in\cX_1(Y)\mid\cS,\vm_0''}=(1+O(\Delta^2/m))\binom{\vm_0''}{Y}\binom{\vm_0''-Y}{Y}\binom{m-\vm_0''}{\Delta-Y}^2\binom{m}{\Delta}^{-2}.
	\end{align}
	Furthermore, a careful expansion of the binomial coefficients from \eqref{eq_lem_EXY_12}  shows that uniformly for all $m_0',m_0''=m\exp(-d)+O(\sqrt m\log^3n)$ we have
	\begin{align*}
		\frac{\pr\brk{x_1,x_2\in\cX_1(Y)\mid\cS,\vm_0''=m_0''}}{\pr\brk{x_1\in\cX_1(Y)\mid\cS,\vm_0'=m_0'}^2}\sim1,
	\end{align*}
	whence we obtain \eqref{eq_lem_EXY_10}.
	A similar argument applies to $|\cX_0(Y)|$.
\end{proof}

As a next step consider the set $\cX_r(Y,Z)$ of all $x\in\cX_r(Y)$ such that
\begin{align*}
	\sum_{a\in\partial x\cap F^+}\vecone\cbc{\partial a\setminus \{x\}\subseteq V_0}=Z.
\end{align*}

\begin{corollary}\label{cor_EXY}
	Let $y\in\cY(c,d,\theta)$ be such that $y\Delta$ is an integer and let $z\in(p_{01},p_{11})$ be such that $z\Delta$ is an integer and such that \eqref{eq_prop_XYZ_1}--\eqref{eq_prop_XYZ_2} are satisfied.
Then \whp we have
	\begin{align}\label{eq_cor_EXY_X0}
		|\cX_0(y\Delta,z\Delta)|&=n\exp\bc{-\Delta\bc{\KL y{\exp(-d)}+y\KL{z}{p_{01}}}+o(\Delta)},\\
		|\cX_1(y\Delta,z\Delta)|&=k\exp\bc{-\Delta\bc{\KL y{\exp(-d)}+y\KL{z}{p_{11}}}+o(\Delta)}.\label{eq_cor_EXY_X1}
	\end{align}
\end{corollary}
\begin{proof}
	Let $Y=y\Delta$ and $Z=z\Delta$.
	We deal with $|\cX_0(y\Delta,z\Delta)|$ and $|\cX_1(y\Delta,z\Delta)|$ by two related but slightly different arguments.
	The computation of 	$|\cX_1(y\Delta,z\Delta)|$ is pretty straightforward.
	Indeed, \Lem~\ref{lem_EXY} shows that \whp the set $\cX_1(Y)$ has the size displayed in \eqref{eq_lem_EXY_X1}.
	Furthermore, since \eqref{eqnoisemodel} provides that tests are subjected to noise independently, \Lem~\ref{lem_chernoff} shows that
	\begin{align}\label{eq_cor_EXY_1}
			\ex\brk{|\cX_1(Y,Z)|\mid|\cX_1(Y)|}&=|\cX_1(Y)|\exp\bc{-\Delta\bc{\KL y{\exp(-d)}+y\KL{z}{p_{11}}}+o(\Delta)}.
	\end{align}
	Moreover, we saw in the proof of \Lem~\ref{lem_EXY}, any $x_1,x_2\in\cX_1(Y)$ have disjoint sets of untainted tests.
	Hence, in perfect analogy to \eqref{eq_cor_EXY_1} we obtain
\begin{align}\label{eq_cor_EXY_2}
	\ex\brk{|\cX_1(Y,Z)|(|\cX_1(Y,Z)|-1)\mid|\cX_1(Y)|}&=|\cX_1(Y)|(|\cX_1(Y)|-1)\exp\bc{-2\Delta\bc{\KL y{\exp(-d)}+y\KL{z}{p_{11}}}+o(\Delta)}.
	\end{align}
	Thus, \eqref{eq_cor_EXY_X1} follows from \eqref{eq_cor_EXY_1}--\eqref{eq_cor_EXY_2} and Chebyshev's inequality.

	Let us proceed to prove \eqref{eq_cor_EXY_X0}.
	As in the case of $|\cX_1(y\Delta,z\Delta)|$ we obtain
\begin{align}\label{eq_cor_EXY_3}
		\ex\brk{|\cX_0(Y,Z)|\mid|\cX_0(Y)|}&=|\cX_0(Y)|\exp\bc{-\Delta\bc{\KL y{\exp(-d)}+y\KL{z}{p_{01}}}+o(\Delta)},
	\end{align}
	Hence, as in \eqref{eq_lem_EXY_4} from the proof of \Lem~\ref{eq_lem_EXY_4},
	\begin{align}\label{eq_cor_EXY_4}
		\ex\brk{|\cX_0(Y,Z)|\mid\cE}&=n\exp\bc{-\Delta\bc{\KL y{\exp(-d)}+y\KL{z}{p_{01}}}+o(\Delta)}.
	\end{align}
	With respect to the second moment calculation, it is not necessarily true that $x_i,x_j\in\cX_0(Y,Z)$ with $k<i<j\leq n$ have disjoint sets of untainted tests.
	Thus, as in the expression \eqref{eq_lem_EXY_11} let $\vm_0''$ be the number of actually negative tests of $\G''=\Gcc-x_i-x_j$ and introduce $0\leq I\leq\Delta$ to count the untainted tests that $x_i,x_j$ have in common.
	Additionally, write $0\leq I_1\leq\min\{I,Z\}$ for the number of common untainted tests that display a negative result.
	Then 
	\begin{align}\nonumber
		\pr\brk{x_i,x_j\in\cX_0(Y,Z)\mid\cS,\vm_0''}=\sum_{I,I_1}&\binom{\vm_0''}{I}\binom{\vm_0''-I}{Y-I}\binom{\vm_0''-Y}{Y-I}\binom{m-\vm_0''}{\Delta-Y}^2\binom{m}{\Delta}^{-2}\\
																 &\qquad\cdot\binom I{I_1}p_{00}^{I-I_1}p_{01}^{I_1}\brk{\binom{Y-I}{Z-I_1}p_{00}^{Y-Z-I+I_1}p_{01}^{Z-I_1}}^2\label{eq_cor_EXY_5}.
	\end{align}
	As in the proof of \Lem~\ref{lem_EXY} it is easily checked that the summand $I=I_1=0$ dominates \eqref{eq_cor_EXY_5}, and that therefore
	\begin{align}\label{eq_cor_EXY_6}
		\ex\brk{|\cX_0(Y,Z)|(|\cX_0(Y,Z)|-1)\mid\cE}\sim\ex\brk{|\cX_0(Y,Z)|\mid\cE}^2.
	\end{align}
	Thus, \eqref{eq_cor_EXY_X0} follows from \eqref{eq_cor_EXY_4}, \eqref{eq_cor_EXY_6} and Chebyshev's inequality.
\end{proof}

\begin{proof}[Proof of \Prop~\ref{prop_XYZ}]
	By continuity we can find $y,z$ that satisfy \eqref{eq_prop_XYZ_1}--\eqref{eq_prop_XYZ_2} such that $y\Delta$, $z\Delta$ are integers, provided that $n$ is large enough.
	Now, if \eqref{eq_prop_XYZ_1}--\eqref{eq_prop_XYZ_2} are satisfied, then \Cor~\ref{cor_EXY} shows that $$|\cX_1(y\Delta,z\Delta)\times\cX_0(y\Delta,z\Delta)|=n^{\Omega(1)}.$$
	Hence, take any pair $(v,w)\in\cX_1(y\Delta,z\Delta)\times\cX_0(y\Delta,z\Delta)$.
	Then $\{a\in\partial v:\partial a\setminus\{v\}\subseteq V_0\}$ and $\{a\in\partial w:a\in F_0\}$ are disjoint, because $v\in V_1$.
	Therefore, any such pair $(v,w)$ satisfies \eqref{eq_prop_XYZ_3}.
\end{proof}

\subsection{Proof of \Prop~\ref{prop_mmt}}\label{sec_cc_mmt}
We are going to lower bound the partition function $Z_{\Gcc,\dSIGMA}$ by way of a moment computation.
To this end we are going to couple the constant column design $(\Gcc,\dSIGMA)$ with the displayed test results $\dSIGMA$ with another random pair $(\Gcc,\DSIGMA)$ where the test results indicated by the vector $\DSIGMA$ are purely random, i.e.\ do not derive from an actual vector $\SIGMA$ of infected individuals.
One can think of $(\Gcc,\DSIGMA)$ as a `null model'.
Conversely, in the language of random constraint satisfaction problems~\cite{CKPZ}, ultimately $(\Gcc,\dSIGMA)$ will turn out to be the `planted model' associated with $(\Gcc,\DSIGMA)$.

Hence, let $\vm^+$ be the number $\|\dSIGMA\|_1$ of positively displayed tests of $(\Gcc,\dSIGMA)$.
Moreover, for a given integer $0\leq m^+\leq m$ let $\DSIGMA\in\{0,1\}^F$ be a uniformly random vector of Hamming weight $m^+$, drawn independently of $\Gcc,\SIGMA,\dSIGMA$.
In other words, in the null model $(\Gcc,\DSIGMA)$ we simply choose a set of uniformly random tests to display positively.

Let $\hat F^+=\{a\in F:\DSIGMA_a=1\}$ and $\hat F^-=F\setminus F^+$.
Moreover, just as in \eqref{eqpsidef} define weight functions
\begin{align}\label{eqpsinulldef}
	\psi_{\Gcc,\DSIGMA,a}&:\{0,1\}^{\partial a}\to\RRpos,&
	\sigma_{\partial a_i}&\mapsto
	\begin{cases}
		\vecone\{\|\sigma\|_{1}=0\}p_{00}+\vecone\{\|\sigma\|_{1}>0\}p_{10}&\mbox{ if $\DSIGMA_a=0$},\\
		\vecone\{\|\sigma\|_{1}=0\}p_{01}+\vecone\{\|\sigma\|_{1}>0\}p_{11}&\mbox{ if $\DSIGMA_a=1$}.
		\end{cases}
\end{align}
In addition, exactly as in \eqref{eqpsiG}--\eqref{eqBoltzmann} let
\begin{align}\label{eqpsiGnull}
	\psi_{\Gcc,\DSIGMA}(\sigma)&=\vecone\cbc{\|\sigma\|_1=k}\prod_{a\in F}\psi_{\Gcc,\DSIGMA,a}(\sigma_{\partial a}),
							   &Z_{\Gcc,\DSIGMA}&=\sum_{\sigma\in\{0,1\}^V}\psi_{\Gcc,\DSIGMA}(\sigma),&
	\mu_{\Gcc,\DSIGMA}(\sigma)&=\psi_{\Gcc,\DSIGMA}(\sigma)/Z_{\Gcc,\DSIGMA}.
\end{align}
We begin by computing the mean of the partition function (aka.\ the `annealed average').

\begin{lemma}\label{lem_mmt}
	For any $0\leq m^+\leq m$ we have $\ex\brk{Z_{\Gcc,\DSIGMA}}=\binom nk\binom{m}{m^+}^{-1}{\pr\brk{\vm^+=m^+}}.$ 
\end{lemma}
\begin{proof}
	Writing out the definitions of $\Gcc,\DSIGMA$, we obtain
	\begin{align*}
		\ex\brk{\vZ_{\Gcc,\DSIGMA}}&=
		\binom{m}{m^+}^{-1}\sum_G\sum_{\substack{\sigma''\in\{0,1\}^F:\|\sigma''\|_1=\vm^+\\\sigma\in\{0,1\}:\|\sigma\|_1=k}}\pr\brk{\Gcc=G}\psi_{\Gcc,\sigma''}(\sigma)\\
												&=\binom{m}{m^+}^{-1}\sum_{G,\sigma'',\sigma:\|\sigma''\|_1=m^+,\|\sigma\|_1=k}\pr\brk{\Gcc=G}\pr\brk{\dSIGMA=\sigma''\mid \Gcc=G,\,\SIGMA=\sigma}&&\mbox{[by \eqref{eqpsinulldef}--\eqref{eqpsiGnull}]}\\
												&=\binom nk\binom{m}{m^+}^{-1}\sum_{G,\sigma'',\sigma:\|\sigma''\|_1=m^+,\|\sigma\|_1=k}\pr\brk{\Gcc=G}\pr\brk{\SIGMA=\sigma}\pr\brk{\dSIGMA=\sigma''\mid \Gcc=G,\,\SIGMA=\sigma}\\
												&=\binom nk\binom{m}{m^+}^{-1}\sum_{\sigma'':\|\sigma''\|_1=m^+}\pr\brk{\dSIGMA=\sigma''}=\binom nk\binom{m}{m^+}^{-1}{\pr\brk{\vm^+=m^+}},
	\end{align*}
	as claimed.
\end{proof}

As a next step we sort out the relationship of the null model and of the `real' group testing instance.

\begin{lemma}\label{lem_nishicc}
	Let $0\leq m^+\leq m$ be an integer.
	Then for any $G$ and any $\sigma''\in\{0,1\}^F$ with $\|\sigma''\|_1=m^+$ we have
	\begin{align}\label{eqlem_nishicc}
		\pr\brk{\Gcc=G,\dSIGMA=\sigma''\mid\vm^+=m^+}&=
		\frac{\pr\brk{\Gcc=G,\DSIGMA=\sigma''}Z_{G,\sigma''}}{\ex\brk{Z_{\Gcc,\DSIGMA}}}.
	\end{align}
\end{lemma}
\begin{proof}
	We have
	\begin{align*}
		\pr\brk{\Gcc=G,\dSIGMA=\sigma''\mid\vm^+=m^+}&=\sum_{\sigma:\|\sigma\|_1=k}\frac{\pr\brk{\Gcc=G}\pr\brk{\dSIGMA=\sigma''\mid\Gcc=G,\SIGMA=\sigma}}{\binom nk\pr\brk{\vm^+=m^+}}\\
													 &=\sum_{\sigma:\|\sigma\|_1=k}\frac{\pr\brk{\Gcc=G}\psi_{G,\sigma''}(\sigma)}{\binom nk\pr\brk{\vm^+=m^+}}&&\mbox{[by \eqref{eqpsidef}--\eqref{eqpsiG}]}\\
													 &=\frac{\pr\brk{\Gcc=G}Z_{G,\sigma''}}{\binom nk\pr\brk{\vm^+=m^+}}&&\mbox{[by \eqref{eqBoltzmann}]}\\
													 &=\frac{\pr\brk{\Gcc=G}\pr\brk{\DSIGMA=\sigma''}Z_{G,\sigma''}}{\binom nk\binom m{m^+}^{-1}\pr\brk{\vm^+=m^+}}&&\mbox{[as $\DSIGMA$ is uniformly random]}\\
													 &=\frac{\pr\brk{\Gcc=G,\DSIGMA=\sigma''}Z_{G,\sigma''}}{\ex\brk{Z_{\Gcc,\DSIGMA}}}&&\mbox{[by \Lem~\ref{lem_mmt}],}
	\end{align*}
	as claimed.
\end{proof}

Combining \Lem s~\ref{lem_mmt}--\ref{lem_nishicc}, we obtain the following lower bound on $Z_{\Gcc,\dSIGMA}$.

\begin{corollary}\label{cor_mmt}
	Let $0\leq m^+\leq m$ be an integer.
	For any $\delta>0$ we have
	\begin{align*}
		\pr\brk{Z_{\Gcc,\dSIGMA}<\delta\binom nk\binom{m}{m^+}^{-1}\pr\brk{\vm^+=m^+}\mid\vm^+=m^+}<\delta.
	\end{align*}
\end{corollary}
\begin{proof}
	\Lem s~\ref{lem_mmt} and~\ref{lem_nishicc} yield
	\begin{align*}
		\pr&\brk{Z_{\Gcc,\dSIGMA}<\delta\binom nk\binom{m}{m^+}^{-1}\pr\brk{\vm^+=m^+}\mid\vm^+=m^+}\\
		   &=\sum_{G,\sigma'':\|\sigma''\|_1=m^+}\frac{\vecone\{Z_{G,\sigma''}<\delta\binom nk\binom{m}{m^+}^{-1}\pr\brk{\vm^+=m^+}\}Z_{G,\sigma''}\pr\brk{\Gcc=G,\DSIGMA=\sigma''\mid\vm^+=m^+}}{\binom nk\binom{m}{m^+}\pr\brk{\vm^+=m^+}}<\delta,
	\end{align*}
	as desired.
\end{proof}

\begin{proof}[Proof of \Prop~\ref{prop_mmt}]
	Since \Prop~\ref{prop_basic} shows that $\pr\brk{\vm^+=mq_0^++O(\sqrt m\log^3n)}=1-o(1)$, the proposition follows immediately from \Cor~\ref{cor_mmt}.
\end{proof}

\begin{appendix}

\section{Proof of \Thm~\ref{thm_inf_apx}}\label{sec_inf_apx}

\noindent
The basic idea is to compute the mutual information of $\SIGMA$ and $\dSIGMA$.
What makes matters tricky is that we are dealing with the {\em adaptive} scenario where tests may be conducted one by one.
To deal with this issue we closely follow the arguments from~\cite{Maurice,Aldridge_2019}.
Furthermore, the displayed test results are obtained by putting the actual test results through the noisy channel.

As a first step we bound the mutual information between $\SIGMA$ and $\dSIGMA$ from above under the assumption that the statistician applies a adaptive scheme where the next test to be conducted depends deterministically on the previously displayed test results.
Let $m$ be the total number of tests that are conducted.

\begin{lemma}\label{lem_mutual_upper}
	For a deterministic adaptive algorithm we have $I(\SIGMA,\dSIGMA)\leq m/\cchan$.
\end{lemma}
\begin{proof}
	Let $\aSIGMA$ be the vector of actual test results.
	Then
	\begin{align*}
		I(\SIGMA,\dSIGMA)&=\sum_{s,s''}\pr\brk{\dSIGMA=s''\mid\SIGMA=s}\pr\brk{\SIGMA=s}\log\frac{\pr\brk{\dSIGMA=s''\mid\SIGMA=s}}{\pr\brk{\dSIGMA=s''}}\\
		&=\sum_{s,s',s''}\pr\brk{\dSIGMA=s''\mid\aSIGMA=s'}\pr\brk{\aSIGMA=s'\mid\SIGMA=s}\pr\brk{\SIGMA=s}\log\frac{\pr\brk{\dSIGMA=s''\mid\aSIGMA=s'}\pr\brk{\aSIGMA=s'\mid\SIGMA=s}}{\pr\brk{\dSIGMA=s''}}\\
		&=I(\dSIGMA,\aSIGMA)-H(\aSIGMA\mid\SIGMA)\leq I(\dSIGMA,\aSIGMA).
	\end{align*}
	Furthermore,
	\begin{align*}
		I(\aSIGMA,\dSIGMA)&=\sum_{s'',s'}\pr\brk{\dSIGMA=s'',\aSIGMA=s'}\log\frac{\pr\brk{\dSIGMA=s'',\aSIGMA=s'}}{\pr\brk{\dSIGMA=s''}\pr\brk{\aSIGMA=s'}}.
	\end{align*}
	Since the tests are conducted adaptively, we obtain
	\begin{align*}
		\pr\brk{\dSIGMA=s''\mid\aSIGMA=s'}&=\prod_{i=1}^m\pr\brk{\dSIGMA_i=s''_i\mid\forall j<i:\dSIGMA_j=s''_j,\,\aSIGMA_i=s'_i}.
	\end{align*}
	Hence,
	\begin{align*}
		I(\aSIGMA,\dSIGMA)&=\sum_{i=1}^m\sum_{s''_1,\ldots,s''_i,s'_i}\pr\brk{\forall j<i:\dSIGMA_j=s''_j,\aSIGMA_i=s'_i}\\&\qquad\qquad\qquad\qquad\cdot\pr\brk{\dSIGMA_i=s''_i\mid\forall j<i:\dSIGMA_j=s''_j,\aSIGMA_i=s'_i}\log\frac{\pr\brk{\dSIGMA_i=s''_i\mid\forall j<i:\dSIGMA_j=s''_j,\aSIGMA_i=s'_i}}{\pr\brk{\dSIGMA_i=s''_i\mid\forall j<i:\dSIGMA_j=s''_j}}.
	\end{align*}
	In the last term $\aSIGMA_i$ is a Bernoulli random variable (whose distribution is determined by $\dSIGMA_j$ for $j<i$), and $\dSIGMA_i$ is the output of that variable upon transmission through our channel.
	Furthermore, the expression in the second line above is the mutual information of these quantities.
	Hence, the definition of the channel capacity implies that $I(\aSIGMA,\dSIGMA)\leq m/\cchan$.
\end{proof}

\begin{proof}[Proof of \Thm~\ref{thm_inf_apx}]
	As a first step we argue that it suffices to investigate deterministic adaptive group testing algorithms (under the assumption that the ground truth $\SIGMA$ is random).
	Indeed, a randomised adaptive algorithm $\cA(\nix)$ can be modeled as having access to a (single) sample $\vec\omega$ from a random source that is independent of $\SIGMA$.
	Now, if we assume that for an arbitrarily small $\delta>0$ we have
	\begin{align*}
		\ex\norm{\cA(\SIGMA'',\vec\omega)-\SIGMA}_1<\delta k,
	\end{align*}
	where the expectation is on both $\vec\omega$ and $\SIGMA$, then there exists some outcome $\omega$ such that
	\begin{align*}
		\ex\norm{\cA(\SIGMA'',\omega)-\SIGMA}_1<\delta k,
	\end{align*}
	where the expectation is on $\SIGMA$ only.

	Thus, assume that $\cA(\nix)$ is deterministic.
	We have $I(\SIGMA,\dSIGMA)=H(\SIGMA)-H(\SIGMA\mid\dSIGMA)$.
	Furthermore, $H(\SIGMA)\sim k\log(n/k)$.
	Hence, \Lem~\ref{lem_mutual_upper} yields
	\begin{align*}
		H(\SIGMA\mid\dSIGMA)&=H(\SIGMA)-I(\SIGMA,\dSIGMA)\geq H(\SIGMA)-m/\cchan,
	\end{align*}
	which implies the assertion.
\end{proof}

\section{Proof of \Lem~\ref{lemma_endgame_misclassified}}\label{sec_lemma_endgame_misclassified}

\noindent
This proof is a straightforward adaption of~\cite[proof of \Lem~4.16]{opt}.
	Fix $T\subset V$ of size $t=|T|\leq \exp(-\log^{\alpha} n)k$ as well as a set $R\subset V$ of size $r=\lceil t/\lambda\rceil $ with $\lambda=8\log\log n$. 
	Let $\gamma=\lceil\ln^{\beta}n\rceil$.
	Further, let $U\subset F[1]\cup\cdots\cup F[\ell]$ be a set of size $\gamma r\leq u\leq \Delta t$.
	Additionally, let $\cE(R,T,U)$ be the event that every test $a\in U$ contains two individuals from $R\cup T$.
	Then
	\begin{align}\label{eq_prop_dist_psi_9}
		\pr\brk{R\subset\cbc{x\in V:\sum_{a\in\partial x\setminus F[0]}\vecone\cbc{T\cap\partial a\setminus \cbc x\neq\emptyset}\geq\gamma}}
		\leq\pr\brk{\cE(R,T,U)}.
	\end{align}
	Hence, it suffices to bound $\pr\brk{\cE(R,T,U)}$.

	For a test $a\in U$ there are no more than $\binom{r+t}2$ ways to choose distinct individuals $x_a,x_a'\in R\cup T$.
	Moreover, \eqref{eqax} shows that the probability of the event $\{x_a,x_a'\in\partial a\}$ is bounded by $(1+o(1))(\Delta\ell/(ms))^2$; in fact, this probability might be zero if we choose an individual that cannot join $a$ due to the spatially coupled construction of $\Gsp$.
	Hence, due to negative correlation
	\begin{align*}
		\pr\brk{\cE(R,T,U)}&\leq \brk{\binom{r+t}2\bcfr{(1+o(1))\Delta\ell}{ms}^2}^u.
	\end{align*}
	Consequently, by the union bound the event $\fE(r,t,u)$ that there exist sets $R,T,U$ of sizes $|R|=r,|T|=t,|U|=u$ such that $\cE(R,T,U)$ occurs has probability
	\begin{align*}
		\pr\brk{\fE(r,t,u)}&\leq \binom nr\binom nt\binom mu\brk{\binom{r+t}2\bcfr{(1+o(1))\Delta\ell}{ms}^2}^u.
	\end{align*}
	Hence, the bounds $\gamma t/\lambda\leq\gamma r\leq u\leq \Delta t$ yield
	\begin{align*}
		\pr\brk{\fE(r,t,u)}&\leq \binom nt^2\binom m{u}\brk{\binom{2t}2\bcfr{(1+o(1))\Delta\ell}{ms}^2}^{u}
		\leq\bcfr{\eul n}{t}^{2t}\bcfr{2\eul \Delta^2\ell^2t^2}{ms^2u}^u\\
						 &\leq\brk{\bcfr{\eul n}{t}^{\lambda/\gamma}\frac{2\eul\lambda\Delta^2\ell^2t}{\gamma ms^2}}^u \leq\brk{\bcfr{\eul n}{t}^{\lambda/\gamma}\cdot\frac{t\ln^4n}m}^u&&\mbox{[due to \eqref{eqell}, \eqref{eqs}]}.
	\end{align*}
	Further, since $\gamma=\Omega(\log^{\beta}n)$ and $m=\Omega(k\log n)$ while $t\leq \exp(-\log^{\alpha} n)k$ and $\alpha+\beta>1$, we obtain 
	\begin{align*}
		\pr\brk{\fE(r,t,u)}\leq \exp(-u\log^{\Omega(1)} n).
	\end{align*}
	Thus, summing on $1\leq t\leq\exp(-\ln^\alpha n)k$, $\gamma r\leq u\leq\Delta t$ and recalling $r=\lceil t/\lambda\rceil$, we obtain
	\begin{align}\label{eq_lemma_endgame_misclassified_2}
		\sum_{t,u}\pr\brk{\fE(r,t,u)}
&\leq\sum_{u\geq1}u\exp(-u\log^{\Omega(1)} n)=o(1).
	\end{align}
	Finally, the assertion follows from \eqref{eq_prop_dist_psi_9} and \eqref{eq_lemma_endgame_misclassified_2}.

\section{Parameter optimisation for the binary symmetric channel}\label{sec_lukas_bsc}

\noindent
Let $\optDs(\theta, \vec p) = \{d > 0 : \max\{\cbound(d,\theta), \ccentre(d)\} = \cex(\theta)\}$ be the set of those $d$ where the minimum in the optimisation problem \eqref{eqccentre} is attained.
The goal in this section is to show that for the binary symmetric channel this minimum is not always attained at the information-theoretic value $\dchan=\log2$ that minimises the term $\ccentre(d)$ from \eqref{eqccentre}.

\begin{proposition}\label{prop:bsc_example_not_log2}
    For any binary symmetric channel $\vec p$ given by $0 < \poi = \pio < 1/2$
        there is $\hat{\theta}(\poi)$ such that for all $\theta > \hat{\theta}(\poi)$ we have $\dchan = \log(2) \not\in \optDs(\theta, \vec p)$.
\end{proposition}

In order to show that $\dchan$ is suboptimal, we will use the following analytic bound on $\cbound(d, \theta)$ for binary symmetric channels.

\begin{lemma}\label{lem:bsc_c_localstab_bounds}
    For a binary symmetric channel $\vec p$ given by $\poi < 1/2$ we have
        \begin{align}
			\frac\theta{-(1-\theta)d \log(1 - (1-a)\exp(-d))}&<
			\cbound(d, \theta) \leq \frac{1}{-(1-\theta)d \log(1 - (1-a)\exp(-d))},\qquad\textup{where}\\ a &= \exp\bc{-\KL{1/2}{\poi}} = 2\sqrt{\poi(1-\poi)}.\label{eqlembsc_c_localstab_bounds_a}
	\end{align}
\end{lemma}

The proof of \Lem~\ref{lem:bsc_c_localstab_bounds} uses the following fact, which can be verified by elementary calculus.

\begin{fact}\label{fact:argmin_of_condition2}
    For all $d > 0$, $0 < p < 1$ and $0 \leq z \leq 1$ we have
	\begin{align*}\argmin_{0 \leq y \leq 1} \cbc{\KL{y}{\exp(-d)} + y \KL{z}{p}} &= \frac{a\exp(-d)}{1 - (1-a)\exp(-d)},\quad\textup{where }a = \exp\bc{-\KL{z}{p}},\quad\mbox{and}\\
		\min_{0 \leq y \leq 1} \cbc{\KL{y}{\exp(-d)} + y \KL{z}{p}} &= -\log(1-(1-a)\exp(-d)).
	\end{align*}
\end{fact}
%

\begin{proof}[Proof of \Lem~\ref{lem:bsc_c_localstab_bounds}]
    Fact~\ref{fact:argmin_of_condition2} shows that
        \[\min_{0 \leq y \leq 1} \cbc{\KL{y}{\exp(-d)} + y\KL{1/2}{\poi}} = - \ln(1 - (1-a)\exp(-d)).\]
		with $a$ as in \eqref{eqlembsc_c_localstab_bounds_a}.
    Hence, it is sufficient to show that
	\begin{align}\label{eqlembsc_c_localstab_bounds_1}
		\theta < \min_{0 \leq y \leq 1} \cbc{\cbound(d, \theta) \cdot d (1-\theta) \bc{\KL{y}{\exp(-d)} + y\KL{1/2}{\poi}}} \leq 1.\end{align}

    Let us first prove the upper bound.
    Choose $\hat{c}(d, \theta)$ such that
	\[\min_{0 \leq y \leq 1} \cbc{\hat{c}(d, \theta) d (1-\theta) \bc{\KL{y}{\exp(-d)} + y\KL{1/2}{\poi}}} = 1;\]
        we need to show that $\cbound(d, \theta) \leq \hat{c}(d, \theta)$.
    By the channel symmetry and because $\poi < 1/2 < \pii$, we have $$\KL{1/2}{\poi} = \KL{1/2}{\pii} \leq \KL{\poi}{\pii}.$$
    Therefore, the definition of $\hat{c}(d, \theta)$ ensures that for all $y \in [0, 1]$ we have
	\begin{align*}\hat{c}(d, \theta) \cdot d(1-\theta) \bc{\KL{y}{\exp(-d)} + y\KL{\poi}{\pii}} &\geq \hat{c}(d, \theta) \cdot d(1-\theta) \bc{\KL{y}{\exp(-d)} + y\KL{1/2}{\poi}} \geq 1 > \theta,\end{align*}
        and thus $\hat{c}(d, \theta) \geq \cyz(d, \theta)$.
		Once again by channel symmetry and the definitions of $\hat{c}(d, \theta)$ and $\fz(y)$ (see \eqref{eq_def_zy}),
        we see that $\KL{\fz(y)}{\pii} \leq \KL{1/2}{\pii} = \KL{1/2}{\poi}$ and hence $\KL{\fz(y)}{\poi} \geq \KL{1/2}{\poi}$ for all $y$.
    Consequently, the definition of $\hat{c}(d, \theta)$ ensures we see that for all $y \in [0, 1]$,
        \[\hat{c}(d, \theta) \cdot d(1-\theta) \bc{\KL{y}{\exp(-d)} + y\KL{\fz(y)}{\poi}} \geq
            \hat{c}(d, \theta) \cdot d(1-\theta) \bc{\KL{y}{\exp(-d)} + y\KL{1/2}{\poi}} \geq 1.\]
			Hence, $\hat{c}(d, \theta) \geq \cbound(d, \theta)$, which is the right inequality in \eqref{eqlembsc_c_localstab_bounds_1}.
        
    Moving on to the lower bound, choose $\check{c}(d, \theta)$ such that
	\begin{align}\label{eqlembsc_c_localstab_bounds_2}\min_{0 \leq y \leq 1} \cbc{\check{c}(d, \theta) d (1-\theta) \bc{\KL{y}{\exp(-d)} + y\KL{1/2}{\pii}}} = \theta;\end{align}
    we need to show that $\cbound(d, \theta) > \check{c}(d, \theta)$.
	The $y = \hat{y}$ where the minimum \eqref{eqlembsc_c_localstab_bounds_2} is attained satisfies $\hat{y} \in \cY(c,d,\theta)$, because $\KL{1/2}{\pii} > 0$ (due to $\pii > 1/2$).
    Moreover, $\fz(\hat{y}) = 1/2$ by the definition of $\fz(\nix)$.
    But since $\KL{1/2}{\poi} = \KL{1/2}{\pii}$ by symmetry of the channel,
        we have
        \[\check{c}(d, \theta) \cdot d(1-\theta) \bc{\KL{\hat{y}}{\exp(-d)} + \hat{y} \KL{\fz(\hat{y})}{\poi}} = \theta < 1.\]
		Hence, we obtain $\check{c}(d, \theta) < \cbound(d, \theta)$, which is the left inequality in \eqref{eqlembsc_c_localstab_bounds_1}.
\end{proof}

To complete the proof of \Prop~\ref{prop:bsc_example_not_log2} we need a second elementary fact.

\begin{fact}\label{fact_z_channel_f_is_concave}
    The function \(f(x, p) = \log(x) \log(1-p x)\) is concave in its first argument for $x, p \in (0, 1)$,
        and for any given $p \in (0, 1)$, any $x$ maximizing $f(x, p)$ is strictly less than $\frac{1}{2}$.
\end{fact}

\begin{proof}[Proof of \Prop~\ref{prop:bsc_example_not_log2}]
    We reparameterise the bounds on $\cbound(d, \theta)$ from \Lem~\ref{lem:bsc_c_localstab_bounds} in terms of $\exp(-d)$, obtaining
        \begin{align*}
			\frac{\theta}{-(1-\theta)f(\exp(-d), 1-a)}&<\cbound(d, \theta) \leq \frac{1}{-(1-\theta)f(\exp(-d), 1-a)},\quad\textup{where}\\ 
f(x, p) &= \log(x)\log(1 - px),\quad a = \exp\bc{-\KL{1/2}{\poi}}.\end{align*}
    As $0 < a < 1$ Fact~\ref{fact_z_channel_f_is_concave} shows that any $x$ maximizing $f(x, 1-a)$ is strictly less than $1/2$.
    Hence, for $\hat{d} > \log(2)$ minimizing $f(\exp(-d), 1-a)$,
        we have $f(\exp(-d), 1-a) < f(1/2, 1-a) = f(\exp(-\dchan), 1-a)$. 
		In particular, the value
        \[\hat{\theta}(\poi) = \inf\cbc{0 < \theta < 1 : \frac{\ccentre(\hat{d})}{\theta} < \frac{1}{-\hat{d}(1-\theta)\log(1-(1-a)\exp(-\hat{d}))} < \frac{\theta}{-\dchan(1-\theta)\log(1-(1-a)\exp(-\dchan))}}\]
		is well defined.
        Hence, for all $\theta > \hat{\theta}(\poi)$ we have
        \[\max\cbc{\cbound(\hat{d}, \theta), \ccentre(\hat{d})} = \cbound(\hat{d}, \theta) < \cbound(\dchan, \theta) = \max\cbc{\cbound(\dchan, \theta), \ccentre(\dchan)},\]
        and thus $\dchan \not\in \optDs(\theta, \vec p)$.
\end{proof}

\section{Parameter optimisation for the $Z$-channel}\label{sec_lukas_Z}

\noindent
Much as in Appendix~\ref{sec_lukas_bsc} the goal here is to show that also for the $Z$-channel the value $\dchan$ from~\eqref{eqdcentre} at which $\ccentre(d)$ from \eqref{eqccentre} attains its minimum is not generally the optimal choice to minimise $\max\{\cbound(d,\theta),\ccentre(d)\}$ and thus obtain the optimal bound $\cex(\theta)$.
%
To this end we derive the explicit formula \eqref{eqLukZ1} for $\cbound(d,\theta)$; the derivation of the second formula \eqref{eqLukZ2} is elementary.

\begin{proposition}\label{z_channel_cbound_closed_form}
    For a $Z$-channel $\vec p$ given by $\poi = 0$ and $0 < \pii < 1$ we have $\cbound(d, \theta) = \frac{\theta}{-(1-\theta) d \log(1 - \exp(-d) \pii)}$.
\end{proposition}

\begin{proof}
    We observe that for a $Z$-channel we have $\cbound(d, \theta) = \cyz(d, \theta)$.
    Indeed, fix any $c > \cyz(d, \theta)$.
	Then by the definitions \eqref{eqcyz} of $\cyz(d, \theta)$ and of $\fz(y)$ we have $\fz(y) > \poi$.
    Since the $Z$-channel satisfies $\poi = 0$, the value $\KL{\fz(y)}{\poi}$ diverges for all $c > \cyz(d, \theta)$, rendering the condition in the definition of $\cbound(d, \theta)$ void for all $y > 0$.
    Moreover, since $c > \cyz(d,\theta)$, we also have $0 \not\in \cY(c,d,\theta)$, and thus $c \geq \cbound(d, \theta)$.
    Since $\cbound(d, \theta) \geq \cyz(d, \theta)$ by definition, this implies that $\cbound(d, \theta) = \cyz(d, \theta)$ on the $Z$-channel.

    Hence, it remains to verify that $\cbound(d,\theta) = \cyz(d,\theta)$ has the claimed value.
    This is a direct consequence of Fact~\ref{fact:argmin_of_condition2} (with $z=\poi$ and $p=\pii$) in combination with the fact that $\KL{\poi}{\pii} = \KL{0}{\pii} = -\log(\pio)$ and thus $1 - \exp(-\KL{\poi}{\pii}) = 1 - \pio = \pii$.
\end{proof}

The following proposition shows that indeed $d=\dchan$ is not generally the optimal choice.
Recall that $\optDs(\theta, \vec p)$ is the set of $d$ where the minimum in the optimisation problem defining $\cex(\theta)$ is attained for a given channel $\vec p$.

\begin{proposition}\label{z_channel_dchan_suboptimal}
    For a $Z$-channel $\vec p$ given by $\poi = 0$ and any $0 < \pii < 1$ there is a $\hat{\theta}(\pii) < 1$ such that for all $\theta > \hat{\theta}(\pii)$ we have $\dchan \not\in \optDs(\theta, \vec p)$.
\end{proposition}

Towards the proof of \Prop~\ref{z_channel_dchan_suboptimal} we state the following fact whose proof comes down to basic calculus.

\begin{fact}\label{fact_z_channel_negtest_more_than_half}
    For a $Z$-channel $\vec p$ where $\poi = 0$ and $0 < \pio < 1$, we have $\exp(-\dchan(\vec p)) > \frac{1}{2}$.
\end{fact}
%

\begin{proof}[Proof of \Prop~\ref{z_channel_dchan_suboptimal}]
    First we show that $d = \dchan$ does not minimize $\cbound(d, \theta)$ for all $0 < \pii < 1$.
    We reparameterise the expression for $\cbound(d, \theta)$ from \Prop~\ref{z_channel_cbound_closed_form} in terms of $\exp(-d)$, obtaining
        \[\cbound(d, \theta) = \frac{\theta}{(1-\theta)f(\exp(-d), \pii)},\quad\textup{where}\quad f(x, \pii) = \log(x)\log(1 - \pii x).\]
    Hence, for any given $\pii$ the value of $\cbound(d, \theta)$ is minimized when $f(\exp(-d), \pii)$ is minimized.
    Using elementary calculus we check that $f$ is concave in its first argument on the interval $(0, 1)$, and that for all $0 < \pii < 1$ the value of $x$ maximizing $f(x, \pii)$ is strictly smaller than $\frac{1}{2}$ (see Fact~\ref{fact_z_channel_f_is_concave}).
    Now for any $Z$-channel with $0 < \pio < 1$ we have $\exp(-\dchan) > \frac{1}{2}$ (using Fact~\ref{fact_z_channel_negtest_more_than_half}).
    Hence, for the $Z$-channel, $\dchan$ does not minimize $\cbound(d, \theta)$.

    Now let $d_1$ be a $d$ minimizing $\cbound(d, \theta)$; in particular, $d_1 \neq \dchan$.
    Since $\frac{\theta}{1-\theta}$ is increasing in $\theta$ and unbounded as $\theta\to1$, the same holds for $\cbound(d_1, \theta)$ and $\cbound(\dchan, \theta)$.
    Hence, we may consider
        \[\hat{\theta}(\pii) = \inf\cbc{0 < \theta < 1 : \cbound(d_1, \theta) > \ccentre(d_1), \cbound(\dchan, \theta) > \ccentre(\dchan)},\]
        check that it is strictly less than $1$, and that by definition of $\hat{\theta}(\pii)$, it holds for all $\theta > \hat{\theta}(\pii)$ that
        \[\max\cbc{\cbound(d_1, \theta), \ccentre(d_1)} = \cbound(d_1, \theta) < \cbound(\dchan, \theta) = \max\cbc{\cbound(\dchan, \theta), \ccentre(\dchan)}.\]
        Consequently, $\dchan \not\in \optDs(\theta, \vec p)$.
\end{proof}


\section{Comparison with the results of Chen and Scarlett on the symmetric channel}\label{apx_scarlettchen}

\noindent
Chen and Scarlett~\cite{scarlettchen} recently derive the precise information-theoretic threshold of the constant column design $\Gcc$ for the symmetric channel (i.e.\ $p_{11}=p_{00}$).
The aim of this section is to verify that their threshold coincides with $m\sim\cex(\theta)k\log(n/k)$, with $\cex(\theta)$ from \eqref{eqcex} on the symmetric channel.
The threshold quoted in~\cite[\Thm s~3 and~4]{scarlettchen} reads $m\sim \min_{d>0}\cchenges(d,\theta)k\log(n/k)$, where $\cchenges(d,\theta)$ is the solution to the following optimisation problem:
\begin{align}\label{eq:cscarlettchen}
\cchenges(d,\theta) =& \max\{\ccentre(d,\theta), \cchen(d,\theta)\},\qquad\mbox{where}\\
	\cchen(d, \theta) = &\Bigg[(1-\theta) d \min_{y \in (0,1), z \in (0,1)} \max\Bigg\{\frac 1 \theta \left( \KL{y}{\eul^{-d} + y \KL{z}{\eul^{-d}}}\right),\nonumber\\ 
						&\hspace{4cm} \min_{y'\in[\abs{y(2z-1)}, 1]}\left(  \KL{y'}{\eul^{-d}} + y' \KL{z'(z,y,y')}{p_{01}}\right) \Bigg\} \Bigg]^{-1}\label{eq:cscarlettchen1}\\
	z'(z,y,y') =& \frac 1 2 + \frac{y(2z-1)}{2y'}.&\nonumber
\end{align}

\begin{lemma}
	For symmetric noise ($p_{00} = p_{11}$) and any $d>0$ we have $\cchenges(d,\theta)= \max\{\cbound(d,\theta),\ccentre(d)\}$. 

\end{lemma}

\begin{proof}
	The definition \eqref{eq:cscarlettchen1} of $\cchen(d)$ can be equivalently rephrased as follows:
	\begin{align*}
		\cchen(d,\theta) =& \inf \big\{c>0: \forall y,z\in\left[0,1\right] \forall y'\in\brk{\abs{y(2z-1)}, 1}: \\&\hspace{20mm} \left( f_1(c,d,\theta, y, z) \geq \theta  \lor  f_2(c,d,\theta, y, z, y', z'(z,y,y')) \geq 1 \right) \big\}, \qquad \text{ where}\\	
		f_1(c,d,\theta, y, z) = & c d (1-\theta) \left( \KL{y}{\exp(-d)} + y \KL{z}{p_{11}} \right),\\
		f_2(c,d,\theta, y, z, y', z') = & c d (1-\theta) \left( \KL{y'}{\exp(-d)} + y' \KL{z'(z,y,y')}{p_{01}} \right).
	\end{align*}
	Consequently, $c<\cchen(d,\theta)$ iff
	\begin{align*}
		\exists y, z, y' \text{ s.t. }
		f_1(c,d,\theta, y, z) < \theta \text { and } 	
		f_2(c,d,\theta, y, z, y', z'(z,y,y')) < 1.
	\end{align*}
	Recall that $c<\cbound(d,\theta)$ iff
	\begin{align*}
		\exists y, z \text{ s.t } f_1(c,d,\theta, y, z) < \theta \text { and } f_2(c,d,\theta, y, z, y,z) < 1.
	\end{align*}	
	Since $z'(z,y,y) = z$ we conclude that if $c<\cbound(d,\theta)$, then $c<\cchen(d,\theta)$.
	Hence, $\cbound(d,\theta)\leq\cchen(d,\theta)$.

	To prove the converse inequality, we are going to show that any $c<\cchen(d,\theta)$ also satisfies $c<\cbound(d,\theta)$.
%
	Hence, assume for contradiction that $c < \cchen(d,\theta)$ and $c \geq \cbound(d,\theta)$. 
	Then the following four inequalities hold:
	\begin{align}
		c d (1-\theta) \left( \KL{y}{\eul^{-d}} + y \KL{z}{p_{11}} \right) <& \theta, &
		c d (1-\theta) \left( \KL{y}{\eul^{-d}} + y \KL{z}{p_{01}} \right) \geq& 1,  \label{eq:ineqy}\\
			c d (1-\theta) \left( \KL{y'}{\eul^{-d}} + y' \KL{z'(z,y,y')}{p_{01}}\right) < & 1, & 
			c d (1-\theta) \left( \KL{y'}{\eul^{-d}} + y' \KL{z'}{p_{11}} \right) \geq& \theta.\label{eq:ineqy'} 
	\end{align}
	The two inequalities on the left are a direct consequence of $c < \cchen(d,\theta)$. 
	Note that if one of the right two inequalities is violated then $c < \cbound(d,\theta)$.
	Combining the inequalities of \eqref{eq:ineqy} leads us to
	\begin{align}\nonumber
		c d (1-\theta) \left( \KL{y}{\eul^{-d}} + y \KL{z}{p_{01}} \right) - &  c d (1-\theta) \left( \KL{y}{\eul^{-d}} + y \KL{z}{p_{11}} \right) 
		\\ = c d (1- \theta )  y \left(  \KL{z}{p_{01}} - \KL{z}{1- p_{01}}\right) = & c d (1-\theta) y (1-2z) \log\left(\frac{p_{01}}{1-p_{01}}\right) >  1-\theta  \label{eq:cSCProofFirst}.
	\end{align} 
	The remaining inequalities given by \eqref{eq:ineqy'} imply
	\begin{align*}	
		c d (1-\theta) \left( \KL{y'}{\eul^{-d}} + y' \KL{z'(z,y,y')}{p_{01}}\right) -&  c d (1-\theta) \left( \KL{y'}{\eul^{-d}} + y' \KL{z'}{p_{11}} \right)\\
		=  c d (1-\theta) y' \left( \KL{z'(z,y,y')}{p_{01}} - \KL{z'(z,y,y')}{1-p_{01}} \right)  = &  c d (1-\theta) y' \frac y {y'} (1-2z) \log\left(\frac{p_{01}}{1-p_{01}}\right)
		 < 1-\theta,
	\end{align*}
	which contradicts \eqref{eq:cSCProofFirst}.
\end{proof}

\end{appendix}

\end{document}